\documentclass[prx,aps,
    twocolumn,
    superscriptaddress, tightenlines,
    nofootinbib, 10pt,
    balancelastpage
]{revtex4-2}

\usepackage[margin=0.75in]{geometry}

\usepackage{physics}
\usepackage{amsmath}
\usepackage{amsfonts}
\usepackage{amssymb}
\usepackage{mathtools}
\usepackage{cancel}
\usepackage{graphicx}
\usepackage{tikz}
\definecolor{myBlue}{RGB}{0, 48, 143}
\definecolor{neededBlue}{RGB}{105,130,182}
\definecolor{tg}{RGB}{51,255,51}
\usepackage{pgfplots}
\usetikzlibrary{arrows.meta, angles, graphs, graphs.standard,  fit, decorations.pathmorphing, decorations.pathreplacing, decorations.shapes, patterns,  patterns.meta, shadows, shapes.geometric, positioning,quotes, decorations.text}
\tikzset{snake it/.style={decorate, decoration=snake}}
\pgfplotsset{compat=1.17} 

\usepackage[LGR,T1]{fontenc}
\usepackage{lmodern}
\usepackage{slantsc}
\usepackage{fix-cm}
\usepackage{url}
\usepackage{microtype}
\usepackage{hyperref}
\hypersetup{
    colorlinks,
    linkcolor=[RGB]{0, 48, 143},
    citecolor=[RGB]{0, 48, 143},
    urlcolor=[RGB]{0, 48, 143}
}
\usepackage[capitalise]{cleveref}

\usepackage{amsthm}
\usepackage{thmtools}
\usepackage{thm-restate}
\newtheorem{theorem}{Theorem}

\newtheorem{conjecture}{Conjecture}
\newtheorem{corollary}{Corollary}
\newtheorem{lemma}{Lemma}
\newtheorem{proposition}{Proposition}

\theoremstyle{definition}
\newtheorem{definition}{Definition}
\theoremstyle{remark}
\newtheorem{remark}{Remark}

\usepackage{paralist}
\usepackage{multirow}

\newcommand{\B}{\{0,1\}}
\newcommand{\poly}[1]{\textnormal{{\rm poly}}\!\left(#1\right)}
\newcommand{\cl}[1]{\textnormal{{\bf #1}}}
\newcommand{\clw}[2]{{\bf #1}{\rm -}{\rm #2}}
\newcommand{\Gate}[1]{{\fontfamily{cmr}\selectfont\textsc{#1}}}
\newcommand{\tn}[1]{\textnormal{#1}}
\newcommand{\e}{{\rm e}}
\renewcommand{\i}{{\rm i}}
\newcommand{\g}{\textsl{g}}
\newcommand{\hw}{\textnormal{hw}}
\renewcommand{\ket}[1]{|#1\rangle}
\renewcommand{\bra}[1]{\langle#1|}
\renewcommand{\sc}[1]{\textnormal{\textsc{#1}}}

\renewcommand{\H}{\mathcal{H}}
\DeclareMathOperator{\Gauss}{\mathcal{G}}
\DeclareMathOperator{\Fend}{\mathcal{F}}

\DeclareMathOperator{\pd}{\phantom{\dagger}}

\input{tikz}

\usepackage{newtxtext}
\makeatletter
\input{t1ntxtlf.fd}
\DeclareFontShape{T1}{ntxtlf}{m}{sl}{<-> \ntx@scaled ptmro8t}{}
\DeclareFontShape{T1}{ntxtlf}{b}{sl}{<-> \ntx@scaled ptmbo8t}{}
\DeclareFontShape{T1}{ntxtlf}{bx}{sl}{<-> ssub * ntxtlf/b/sl}{}
\makeatother

\usepackage{setspace}
\newcommand{\manualtocentry}[2]{%
  \noindent%
  \hyperref[#1]{\ref*{#1}.\quad #2}%
  \nobreak{\dotfill}\nobreak%
  \pageref{#1}%
  \\[0.1cm]
}

\newcommand{\manualtocsubentry}[2]{%
  \phantom{~}\hspace{3em}%
  \hyperref[#1]{\ref*{#1}.\quad #2}%
  \nobreak{\dotfill}\nobreak%
  \pageref{#1}%
  \\[0.1cm]
}

\newcommand{\altmanualtocentry}[2]{%
  \noindent%
  \hyperref[#1]{#2}%
  \nobreak{\dotfill}\nobreak%
  \pageref{#1}%
  \\[0.1cm]
}

\newcommand{\manualtocentryapp}[2]{%
  \noindent%
  \hyperref[#1]{\ref*{#1}.\quad #2}%
  \nobreak{\dotfill}\nobreak%
  \pageref{#1}%
  \\[0.1cm]
}

\newcommand{\manualtocsubentryapp}[2]{%
  \phantom{~}\hspace{3em}%
  \hyperref[#1]{\ref*{#1}.\quad #2}%
  \nobreak{\dotfill}\nobreak%
  \pageref{#1}%
  \\[0.1cm]
}

\newcommand{\manualtocsubsubentryapp}[2]{%
  \phantom{~}\hspace{6em}%
  \hyperref[#1]{\ref*{#1}.\quad #2}%
  \nobreak{\dotfill}\nobreak%
  \pageref{#1}%
  \\[0.1cm]
}
\makeatletter
\def\l@subsubsection#1#2{}
\makeatother
\begin{document}

\title{Physically-Motivated Guiding States for Local Hamiltonians}

\author{Gabriel Waite}\email{gabrielwaite1999@outlook.com}
\affiliation{Centre for Quantum Computation and Communication Technology}
\affiliation{Centre for Quantum Software and Information, 
School of Computer Science, Faculty of Engineering \& Information Technology, University of Technology Sydney, NSW 2007, Australia
}

\author{Karl Lin}
\affiliation{Centre for Quantum Computation and Communication Technology}
\affiliation{Centre for Quantum Software and Information, 
School of Computer Science, Faculty of Engineering \& Information Technology, University of Technology Sydney, NSW 2007, Australia
}

\author{Samuel J Elman}\email{samuel.elman@uts.edu.au}
\affiliation{Centre for Quantum Computation and Communication Technology}
\affiliation{Centre for Quantum Software and Information, 
School of Computer Science, Faculty of Engineering \& Information Technology, University of Technology Sydney, NSW 2007, Australia
}

\author{Michael J Bremner}\email{michael.bremner@uts.edu.au}
\affiliation{Centre for Quantum Computation and Communication Technology}
\affiliation{Centre for Quantum Software and Information, 
School of Computer Science, Faculty of Engineering \& Information Technology, University of Technology Sydney, NSW 2007, Australia
}

\begin{abstract}
    We study the computational complexity of the \textsc{Guided Local Hamiltonian} problem: 
    given a local Hamiltonian $H$ together with a classical description of a guiding state that has non-negligible overlap with the ground state of $H$, estimate the ground-state energy within inverse-polynomial precision.
    This setting captures real-world scenarios in quantum chemistry and many-body physics, where trial states derived from classical heuristics can be used to guide quantum algorithms.
    We identify families of physically-motivated guiding states for which the computational hardness of ground-state energy estimation persists in the guided setting.
    Extending prior results for semi-classical subset states, we prove \textbf{BQP}-hardness for classes including fixed-weight states, matrix product states, Gaussian states, and Fendley states.
    Our hardness results are obtained via refined Feynman-Kitaev circuit-to-Hamiltonian constructions that explicitly expose the structural role of the guiding state in the reduction.
    Complementing these results, we give a constructive proof of \textbf{BQP} containment when the guiding state admits a polynomial-size classical description, establishing \textbf{BQP}-completeness for the canonical formulation of the problem.
    Our results show that quantum advantage persists for the newly introduced state classes, and classical methods also remain viable when said guiding states admit appropriate descriptions.
    Together, our results identify a \emph{Goldilocks zone} of guiding states that are efficiently preparable, succinctly described, and sample-query accessible, within which quantum advantage for ground-state estimation can be meaningfully assessed.
    We additionally formalise the \textsc{Guided Fermi-Hubbard Hamiltonian} problem and prove \textbf{BQP}-completeness on 2D square and triangular lattices, both with and without magnetic fields, when provided with an appropriate fermionic guiding state.
\end{abstract}
\maketitle
{
\begin{spacing}{0.8}
\begin{center}
    \manualtocentry{sec:intro}{Introduction}
    \manualtocentry{sec:bqp_hardness}{Hardness of the Guided Local Hamiltonian Problem}
        \manualtocsubentry{sec:history_state_hardness}{Warm Up with the History State}
        \manualtocsubentry{sec:guiding_state_framework}{Guiding State Construction Framework}
        \manualtocsubentry{sec:2-local_reduction}{Reduction to 2-Local Hamiltonians}
    \manualtocentry{sec:state_type_variations}{State Type Variations and Implications}
        \manualtocsubentry{sec:fixed_weight_states}{Fixed-Weight States}
        \manualtocsubentry{sec:mps_states}{Matrix Product States}
        \manualtocsubentry{sec:gaussian_states}{Gaussian States}
        \manualtocsubentry{sec:main_result_a}{Main Result}
        \manualtocsubentry{sec:implications_extensions}{Implications and Extensions}
    \manualtocentry{sec:bqp_containment}{Containment of the Guided Local Hamiltonian Problem in BQP}
        \manualtocsubentry{sec:qpe}{Quantum Phase Estimation}
        \manualtocsubentry{sec:efficient_state_preparation}{Efficient State Preparation}
        \manualtocsubentry{sec:main_result_b}{Main Result}
    \manualtocentry{sec:classical_tractability}{Classical Tractability within the Goldilocks Zone}
        \manualtocsubentry{sec:goldilocks_zone}{The Goldilocks Zone}
        \manualtocsubentry{sec:dequantisation_algorithm}{Dequantisation Algorithm}
    \manualtocentry{sec:fh_model}{Complexity of the Fermi-Hubbard Model with a Guiding State}
        \manualtocsubentry{sec:duality_mapping}{Heisenberg-Fermi-Hubbard Duality}
        \manualtocsubentry{sec:fermionic_guiding_states}{Fermionic Guiding States}
        \manualtocsubentry{sec:fermi_hubbard_local_fields}{Fermi-Hubbard with Local Fields}
        \manualtocsubentry{sec:main_result_c}{Main Results}
    \manualtocentry{sec:conclusion}{Conclusion}
    \altmanualtocentry{sec:acknowledgments}{Acknowledgments}
    \altmanualtocentry{sec:refs}{References}
    \altmanualtocentry{app:toc}{Appendices}
\end{center}
\end{spacing}
}
\newpage
\section{Introduction}\label{sec:intro}
A result of quantum complexity theory is the intractability of computing the ground-state energy of an arbitrary local Hamiltonian~\cite{kitaev2002classical}.
This complexity persists even for physically-motivated Hamiltonians, such as those relevant to quantum chemistry~\cite{schuch2009computational,cubitt2018universal,ogorman2021electronic}. 
Under the widely believed assumption that $\cl{BQP} \neq \cl{QMA}$, no efficient quantum algorithm is expected to solve this problem without additional information.
To address this challenge, the problem \sc{Guided Local Hamiltonian} was introduced~\cite{richter2007two,gharibian2023dequantizing,cade2023improved}, where the input includes a classical description of a \emph{guiding state} that has non-trivial overlap with the ground state.
The computational problem is to estimate the ground-state energy of a local Hamiltonian, given access to this guiding state.

Trial states to guide ground-state energy searches are often constructed using classical heuristic algorithms.
To narrow the search space, methods like density functional theory (DFT)~\cite{tsuneda2014density,bogojeski2020quantum}, density matrix renormalisation group (DMRG)~\cite{white1992density}, and Hartree-Fock~\cite{arute2020hartree} leverage structural approximations such as mean-field ansatze or interaction strength bounds.
Additional methods have used active space truncations, based on physical intuition, to reduce the computational overhead~\cite{clinton2024towards}.
Yet, selecting optimal bases with which to perform these pre-computations is generally \clw{QMA}{hard}~\cite{ogorman2021electronic}, underscoring the challenge of circumventing worst-case complexity bounds.

The practical success of guided energy estimation raises the question of whether a theoretical framework can be developed to analyse its performance --- particularly in worst-case settings and across varying parameter regimes.
To study this meaningfully, we must place natural restrictions on the class of guiding structures considered.
In this work, we focus primarily on guiding states that are physically motivated, admit succinct classical descriptions, and can be efficiently prepared.
Richter~\cite{richter2007two} was the first to formally study the use of guiding states in a complexity-theoretic context.
It was shown that both polynomial-time preparable states and simple product-states could serve as guiding states; provided a sufficient overlap with the ground state, this renders Hamiltonian ground-state energy estimation \clw{BQP}{complete}.
Building on this, recent work has further explored \emph{assisted} variants of the \sc{Local Hamiltonian} problem.
\citet{gharibian2023dequantizing} introduced the \sc{Guided Local Hamiltonian} problem, where the input includes a state with (at least) inverse-polynomial overlap with the true ground state.
It was shown this problem is \clw{BQP}{hard}, even for $2$-local Hamiltonians~\cite{cade2023improved} and specific local stoquastic Hamiltonians~\cite{waite2025guided}, yet classically tractable under bounded-precision and overlap constraints~\cite{gharibian2023dequantizing}.
These results suggest a quantum advantage in regimes requiring inverse-polynomial precision, which is relevant to the ``chemical accuracy'' in quantum chemistry.\footnote{Chemical accuracy is not universal; the required precision typically depends on experimental or computational needs.}
\citet{zhang2024dequantized} recently extended the classical tractability to a broader class of Hamiltonians via the use of randomised imaginary-time evolution.
This approach, however, requires the guiding state to satisfy both overlap and circuit-depth constraints.

In this work, we study the \sc{Guided Local Hamiltonian} problem under physically-motivated restrictions of the guiding state.
This problem augments standard ground-state energy estimation by including, as part of the input, a guiding state via a succinct classical description that is promised to have non-negligible overlap with the true ground state --- see \cref{def:guided_local_hamiltonian} for a formal statement of this computational problem.
While prior work established hardness and classical tractability results for broadly defined guiding states, it remained unclear which aspects of these results persist under structural and physical constraints that arise naturally in quantum many-body systems.
We introduce guiding states with classical encodings that extend beyond those previously considered~\cite{gharibian2023dequantizing,cade2023improved,waite2025guided}, necessitating constructive proofs of polynomial-time decidability and polynomial-time reductions to establish \cl{BQP} containment and \clw{BQP}{hardness} respectively.
Our work provides a broader set of parameters to explore for both theoretical and practical efforts~\cite{nguyen2025theoretical}.

Our primary technical contribution concerns the robustness of \clw{BQP}{hardness} when restricting the structure of the guiding state.
Extending earlier results for semi-classical subset states, we prove \clw{BQP}{hardness} for new families of structured guiding states, including fixed-weight states, matrix product states, Gaussian states, and Fendley states --- see \cref{sec:state_type_variations} for formal definitions of the aforementioned states and remarks on their physical relevance.
Our hardness proofs rely on refined Feynman-Kitaev circuit-to-Hamiltonian constructions that expose how the structure of the guiding state constrains, for example, admissible clock encodings and local transitions.
Rather than exploiting specialised properties of individual state families, our analysis shows that hardness persists under general structural requirements, demonstrating that quantum hardness is not an artefact of unphysical or highly unstructured guiding states.
As an initial application of this framework, we first analyse the \clw{BQP}{hardness} construction using the standard history state as the guiding state.
This serves both to illustrate the general reduction strategy and to isolate the structural features of the history state that are essential for preserving hardness.
To motivate the subsequent guiding state constructions, we introduce a windowing technique that restricts attention to carefully chosen support sectors of the full history state.
This perspective allows us to identify which components of the circuit history can be modified or omitted, such as gates that act trivially or are sufficiently close to the identity, without compromising the validity of the reduction.
By characterising these permissible simplifications, we obtain a method for constructing alternative guiding states that retain the necessary overlap and structural properties required for \clw{BQP}{hardness}.
Subsequent use of perturbative gadgets~\cite{kempe2006complexity,oliveira2008complexity} allows us to reduce the locality of the Hamiltonians involved, yielding \clw{BQP}{hardness} results for $2$-local Hamiltonians under each guiding state family.
A high-level summary of the reduction procedure is provided in \cref{fig:hardness-summary}.

Another contribution of this work is to make explicit the containment of the \sc{Guided Local Hamiltonian} problem within \cl{BQP} when the guiding state is provided via a suitable polynomial-size classical description.
Although \clw{BQP}{completeness} was previously claimed for related formulations, earlier works did not provide a constructive proof of containment for the canonical input model~\cite{gharibian2023dequantizing,cade2023improved}.
We show that such containment is not automatic: a classical description of a state does not, by itself, guarantee efficient quantum preparation.
By identifying sufficient conditions under which guiding states can be efficiently prepared from their descriptions, we provide a concrete and rigorous \cl{BQP} containment algorithm that clarifies implicit assumptions made in prior literature and establishes \clw{BQP}{completeness} for the extended families we introduce.
More generally, a succinct classical description of a state family does not by itself imply the existence of a uniform and efficient preparation procedure.
In particular, such descriptions may be non-uniform across system sizes, with no known algorithm that generates the description or a corresponding preparation circuit as a function of $n$ (see \cref{rmk:uniformity} for further discussion).
This distinction is especially relevant for guiding state families such as Fendley states, where efficient and uniform preparation from known descriptions remains an open question~\cite{ruh2025furthering}.
It is thus important to clarify the assumptions under which \cl{BQP} containment holds.

We additionally investigate when the \sc{Guided Local Hamiltonian} problem admits efficient classical algorithms.
We show that several of the physically-motivated guiding state families introduced in this work also preserve known classical tractability results, provided the guiding state admits efficient sample- and query-access~\cite{vandennest2011simulating} under comparable input assumptions.
This motivates the definition of a shared regime in which quantum and classical algorithms operate on the same form of input information \cite{kieferova2022assume}.
We refer to this regime as the \emph{Goldilocks zone}: guiding states that are succinctly described, efficiently preparable, and classically accessible.
Within this zone, meaningful comparisons between quantum and classical computational power can be made without relying on asymmetric input models.

Our results on the constructive \cl{BQP} containment also clarify the relationship to the classical result of Ref.~\cite{zhang2024dequantized}, which proves the existence of a classical approximation algorithm for the ground-state energy of local Hamiltonians whose operator norms scale physically, provided the guiding state is prepared by a constant-depth quantum circuit.
For the guiding state families introduced in this work, as well as for several previously studied cases, we show that efficient preparation generally requires circuit depth growing with system size, placing them outside the regime covered by constant-depth dequantisation algorithms.
While this does not preclude the existence of other classical algorithms or special cases where known guiding states admit constant-depth preparations, our results highlight the importance of circuit depth as a parameter governing the classical tractability of constructive state preparation.

Finally, we consider the extension of the guided problem to physically-motivated settings involving interacting fermionic systems.
Specifically, we study the \textsc{Guided Fermi-Hubbard Hamiltonian} problem, where the input is a Fermi-Hubbard Hamiltonian defined on a 2D square or triangular lattice, together with a fermionic guiding state that has non-negligible overlap with the ground state --- see \cref{def:guided-fermi-hubbard} for a formal statement.
We prove there exists a polynomial-time computable reduction establishing the \clw{BQP}{hardness} of this problem, both with and without local fields present, using a perturbative analysis of the models kinetic and potential terms; notably, we do not rely on the Feynman-Kitaev circuit-to-Hamiltonian construction.
The proof strategy utilises the standard duality mapping between the Fermi-Hubbard model and the antiferromagnetic Heisenberg model~\cite{auerbach1994interacting, hubbard1963electron}.
The guiding state families we construct have a natural association with the existing semi-classical subset states and subsequently admit succinct classical descriptions.
Our result closes an open problem raised in Ref.~\cite{cade2023improved} concerning the extension of the complexity results for the \sc{Guided Local Hamiltonian} beyond the standard spin models.

Open questions remain, particularly about whether there are rigorous algorithms for constructing guiding states relevant to this regime \cite{bansal2009classical,buhrman2025beating}.
Recent work suggests that a study on the complexity of finding extremal states from specific classes may be fruitful.
For example, \citet{kallaugher2024complexity} proved \clw{NP}{completeness} for deciding the energy of extremal product states.
Interestingly, our results on Gaussian guiding states complement this result by proving the \sc{Gauss Local Hamiltonian} problem is \clw{NP}{complete} (see \cref{app:STLHP} for details).
Using more fine-grained tools from parametrised complexity theory, \citet{bremner2025parameterized} showed that estimating the lowest energy for superposition states parameterised by Hamming weight is in \cl{QW$[1]$} and hard for \cl{QM$[1]$}.
Additionally, recent results have demonstrated \clw{MA}{completeness} for Hamiltonians with succinctly represented ground states~\cite{waite2025complexityb,jiang2025local}.  
These variants expand the potential scope for the applicability of classical heuristics to problems beyond the standard \sc{Local Hamiltonian} problem.

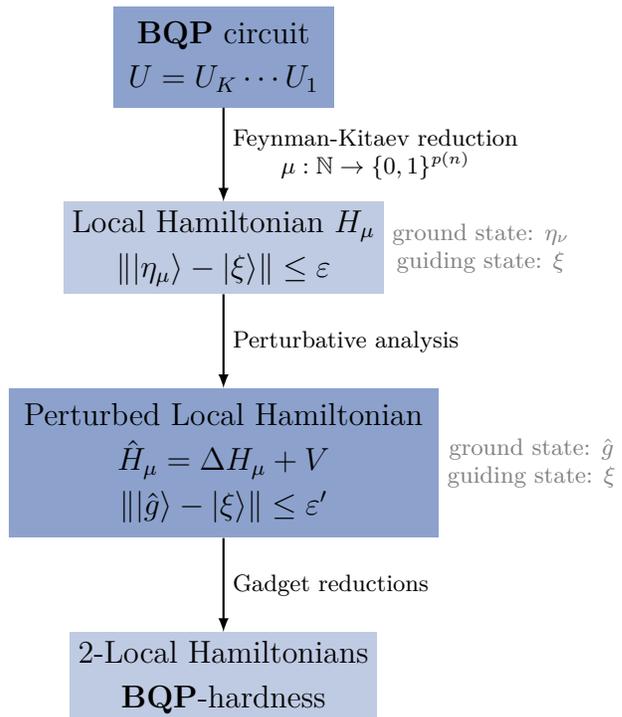
\begin{figure}[!ht]
    \centering
    \begin{tikzpicture}
        \pic{summary};
    \end{tikzpicture}
    \caption{
        Summary of the steps taken to prove \clw{BQP}{hardness} for the \sc{Guided Local Hamiltonian} problem. 
        We reduce from arbitrary \cl{BQP} circuits to a local Hamiltonian $H_{\mu}$, defined for a specific clock register encoding $\mu$. 
        Next, we construct a guiding state $\xi$ that is guaranteed to have overlap with the ground state $\eta_\nu$ of the Hamiltonian $H_{\mu}$. 
        Then we perform a perturbative analysis to show that the ground state of the perturbed Hamiltonian $\hat{H}_{\mu}$ is close to the guiding state $\xi$. 
        Further gadget reductions are employed to reduce the locality to $2$, concluding the proof of \clw{BQP}{hardness}.
        }
    \label{fig:hardness-summary}
\end{figure}

\subparagraph{Prior Work.}
Richter~\cite{richter2007two} proposed the problem \sc{Local Hamiltonian$^*$} as an extension of the standard \sc{Local Hamiltonian} problem, where the input includes a state which has at least inverse-polynomial overlap with a low-energy state of the Hamiltonian.
Standard arguments concerning eigenvalue estimation and small extensions to the Feynman-Kitaev circuit-to-Hamiltonian reduction show the problem to be \cl{BQP}-complete.
Not much structure was placed on the type of guiding state, other than that its efficiency concerning preparability.
The \sc{Guided Local Hamiltonian} problem was later introduced by \citet{gharibian2023dequantizing}, who established its definition and initial hardness result. 
Subsequent works, such as Ref.~\cite{cade2023improved}, extended these results to lower locality and physically relevant Hamiltonian families. 
The guiding state types considered in these works were motivated by a result showing \cl{QMA} is unchanged if proof states are replaced with subset states~\cite{grilo2015qma}.
Such states admit a specific structure, different from the state types previously considered.

Ref.~\cite{cade2022complexity} further explored the use of guiding states in the context of excited states as well as the canonical setting; though this problem altered the input format of the guiding state.
Ref.~\cite{zhang2024dequantized} proposed a dequantisation algorithm, based on standard cluster expansion techniques, broadening the scope of classical tractability~\cite{gharibian2023dequantizing}. 
A step in this algorithm performs the mapping $H \mapsto U^\dagger H U$, where $U$ is the quantum circuit preparing the guiding state.
Simple light cone arguments demonstrate that for classical efficiency, $U$ must be constant-depth.
Beyond this, we would need to perform a \emph{quantum polynomial-time reduction} to transform the Hamiltonian.
Our work builds on these results, offering a detailed analysis of the complexity of guiding state variations from the perspective of both classical and quantum algorithms.

This work differs from Ref.\cite{deshpande2022importance}'s \sc{Ground State Description} and Ref.\cite{weggemans2024guidable}'s \sc{Guidable Local Hamiltonian} problems, which extend Ref.~\cite{bravyi2015monte}, but do not treat the guiding state as part of the input of the problem's instance; rather, a guiding state is promised to exist.
Ref.~\cite{weggemans2024guidable} focuses on guidable states motivated by physical constraints, e.g., matrix product, stabiliser, and IQP states, but within a different complexity framework.
In contrast, our problem formulation is distinct: we analyse a different class of guiding states, under a different set of assumptions, and with a different problem input model.
While the settings share commonalities, our results are logically independent and complementary to theirs, with no contradictions.
We also highlight that our construction admits guiding state families that have not been previously addressed.

\subparagraph{Outline.}
The remainder of this work is structured as follows.
In \cref{sec:bqp_hardness}, we present our framework for establishing \clw{BQP}{hardness} results for the \sc{Guided Local Hamiltonian} problem.
We begin by illustrating the reduction strategy using the standard history state as the guiding state.
This serves both to illustrate an easy example of a valid guiding state and to isolate the structural features of the history state that are essential for preserving hardness.
Then, in \cref{sec:state_type_variations}, we present our main \clw{BQP}{hardness} results for various physically-motivated guiding state families.
\cref{sec:bqp_containment} provides a constructive proof of the \cl{BQP} containment for the canonical problem formulation, under suitable assumptions on the guiding state description.
\cref{sec:classical_tractability} discusses the classical tractability of the \sc{Guided Local Hamiltonian} problem under comparable input assumptions and introduces the \emph{Goldilocks zone} of guiding states.
\cref{sec:fh_model} presents our results on the \textsc{Guided Fermi-Hubbard Hamiltonian} problem, establishing \clw{BQP}{completeness} for the model in a specific regime.
Finally, we conclude in \cref{sec:conclusion} with some remarks and open problems.
Where appropriate, we have deferred proofs to the appendices.
Furthermore, in the appendices, we provide supplementary material including formal definitions, proofs, and additional discussions that support the main text.
In particular, we define polynomial-time generated quantum circuits, the complexity class \cl{BQP} in \cref{app:complexity}, and various access models for quantum states in \cref{app:sample-query}.

\section{Hardness of the Guided Local Hamiltonian Problem}\label{sec:bqp_hardness}
In what follows we denote the eigensystem of a Hamiltonian $H$ as $\{(\lambda_j, \ket{\phi_j})\}_{j=0}^{2^n-1}$, where $\lambda_0$ and $\ket{\phi_0}$ are the ground-state energy and ground state respectively.
We define the spectral gap as $\gamma \coloneqq \min_j \{\lambda_j - \lambda_0 : \lambda_j > \lambda_0\}$.
A Hamiltonian over $n$ qubits is said to be $k$-local if it can be expressed as a sum of (self-adjoint) terms, each acting non-trivially on at most $k$ qubits.
Denote the fidelity between two states $\ket{\psi}$ and $\ket{\phi}$ as $F_{\psi,\phi} = \abs{\braket{\psi}{\phi}}^2$.
We formally define the main computational problem studied in this work.

\begin{restatable}[\sc{[State Type] Guided Local Hamiltonian} problem]{definition}{guidedlocalhamiltonian}\label{def:guided_local_hamiltonian}
    Given a $k$-local Hamiltonian $H$ acting on $n$ qubits such that $\norm{H} \leq 1$, parameters $a,b \in [0,1]$ such that $b-a \geq 1/\poly{n}$ and a classically efficient description of a state $\ket{\xi}$, with the promise that $\norm{\ketbra{\phi_0}\ket{\xi}}^2 \geq \delta$ for some $0 <\delta < 1$, decide whether $\lambda_0 \leq a$ or $\lambda_0 \geq b$, promised one is true.
\end{restatable}

A \clw{BQP}{hardness} result for the above problem entails encoding an arbitrary \cl{BQP} computation into a $k$-local Hamiltonian $H$ such that the ground-state energy $\lambda_0$ reflects the outcome of the computation.
In particular, if the computation accepts, then $\lambda_0 \leq a$, and if it rejects, then $\lambda_0 \geq b$, for a suitable choice of $a$ and $b$.
Equivalently, we embed the circuit's acceptance probability into the low-energy spectrum of $H$.
In addition to this encoding, we must also construct a guiding state $\ket{\xi}$ of the specified type that has non-negligible overlap with the ground state $\ket{\phi_0}$ of $H$.
Furthermore, we must express a classically efficient description of $\ket{\xi}$ that forms part of the input to the problem instance.
Since we are aiming for polynomial-time Karp reductions, all steps of the construction must be computable in polynomial time.

We begin by illustrating a sketch of the general reduction strategy using the standard history state as the guiding state.
This serves both to illustrate an easy example of a valid guiding state and to isolate the structural features of the history state that are essential for preserving hardness.
To establish this result, we overview the Feynman-Kitaev circuit-to-Hamiltonian construction starting from an arbitrary \cl{BQP} circuit.
We then identify which components of the circuit history can be modified or omitted, such as gates that act trivially or are sufficiently close to the identity, without compromising the validity of the reduction.
As a result, we obtain a general formula for constructing alternative guiding states that retain the necessary overlap and structural properties required for \clw{BQP}{hardness}.
In the subsequent sections we apply this framework to various physically-motivated guiding state families.

\subsection{Warm Up with the History State}\label{sec:history_state_hardness}
Let $x \in \B^*$ be a problem instance for an arbitrary promise problem in \cl{BQP} with $\mathcal{U}_x$ the corresponding polynomial-time generated quantum circuit that decides $x$.
Via repetition and majority voting, we can assume that $\mathcal{U}_x$ accepts with probability at least $1 - \epsilon$ if $x$ is a \sc{yes} instance, and at most $\epsilon$ if $x$ is a \sc{no} instance, for some $\epsilon = 2^{-\poly{|x|}}$.
We assume that $\mathcal{U}_x = U_K \cdots U_1$ is a quantum circuit composed of $K = \poly{|x|}$ gates acting on $|x| + \poly{|x|}$ qubits, where each gate $U_t$ is drawn from a universal gate set $\mathcal{G}$ such that each $U_t$ acts on at most $2$ qubits.
Take the first qubit, denoted as $1$, to be the output qubit of the computation --- i.e., measuring this qubit in the computational basis reveals the outcome of the computation.
The idea is to construct a local Hamiltonian $\hat{H}$ such that the ground state encodes the history of the computation performed by $\mathcal{U}_x$.
Note that in contrast to standard constructions working from Merlin-Arthur circuits~\cite{kitaev2002classical}, we do require a proof state (see \cref{app:bqp_hardness_framework} for further discussion).

Define the local Hamiltonian over the space $(\mathbb{C}^2)^{\otimes(n + m + K+1)}$ where $n \coloneqq |x|$ is the number of input qubits, $m \coloneqq \poly{|x|}$ is the number of ancilla qubits, and the remaining $K+1$ qubits form the clock space that tracks the time step of the computation, encoded in unary.
Following the Feynman-Kitaev construction~\cite{kitaev2002classical,feynman1986quantum}, we define the local Hamiltonian $\hat{H} = \Delta H + V$, where
\begin{align*}
    H &\coloneqq H_{\tn{in}} + H_{\tn{prop}} + H_{\tn{clock}}, &
    V &\coloneqq H_{\tn{out}}, 
\end{align*}
where each term is defined to penalise invalid configurations of the computation, and $\Delta > 0$ is a scaling parameter.
The specific forms of these terms are presented in a moment; for now, we focus on constructing the guiding state.
The frustration-free ground state of $H$ is the well-known history state:
\begin{equation*}
    \ket{\eta} \coloneqq \frac{1}{\sqrt{K+1}} \sum_{t=0}^{K} \ket{\varphi_t}\ket{t},
\end{equation*}
where $\ket{\varphi_t} \coloneqq U_t \ket{\varphi_{t-1}}$ is the state of the computation after $t$ time steps and $\ket{\varphi_0} \coloneqq \ket{x}\ket{0^m}$ is the initial state.
Perturbing $H$ by $V$ introduces an energy penalty for rejecting computations, thereby shifting the eigensystem of $\hat{H}$ to reflect the acceptance probability of $\mathcal{U}_x$.
Importantly, the history state $\ket{\eta}$ is not the ground state of $\hat{H}$, but has non-negligible overlap with it.
The following lemma summarises these properties.

\begin{restatable}{lemma}{lmahistorystateguiding}\label{lma:history_state_guiding}
    Let $L \subseteq \B^*$ be a \cl{BQP} promise problem.
    Take $\mathcal{U}_x$ be a sequence of $K = \poly{|x|}$ unitary gates for instance $x \in L$.
    Using the Feynman-Kitaev construction, define the Hamiltonian
    \begin{equation*}
        \hat{H} = \Delta(H_{\tn{in}} + H_{\tn{clock}} + H_{\tn{prop}}) + V_{\tn{out}}.
    \end{equation*}
    Taking $\Delta > 112 K^3$ and $V_{\tn{out}} = \ketbra{0}_1 \ketbra{K}_C$ such that $\norm{V_{\tn{out}}} = 1$, results in $\hat{H}$ having a one-dimensional ground space spanned by $\ket{\hat{\phi}_0}$.
    Furthermore, it follows that 
    \begin{equation*}
        \norm{\ket{\hat{\phi}_0} - \ket{\eta}} = O(1/\poly{|x|}),
    \end{equation*}
    where $\ket{\eta}$ is the history state of the circuit.
\end{restatable}

The proof of \cref{lma:history_state_guiding} is provided in \cref{app:bqp_hardness_framework}.\footnote{In fact, we prove a stronger result using the clock mapping introduced in the next subsection.}
Since the history state and the ground state of $\hat{H}$ are close in norm, it follows that they also have non-negligible overlap.
Thus, the history state $\ket{\eta}$ can serve as a valid guiding state for the \sc{Guided Local Hamiltonian} problem.

The history state is a member of a family of states we call the \emph{Unitarily Transformed Subset States}~\cite{waite2025guided}, defined as:
\begin{equation*}
    \ket{S_{U}} \coloneqq \frac{1}{\sqrt{\abs{S}}}\sum_{z \in S} U_z \ket{z},
\end{equation*}
where $U = \{U_z\}_{z \in S}$ is a set of unitaries that act on the state $\ket{z}$ and $S$ is a subset of at most a polynomial number of computational basis states.
Given the classical description of the set $S$ and the unitaries $U_z$, it is possible to prepare the state $\ket{S_{U}}$ efficiently.
In \cref{app:history-state-prep} we construct a polynomial-size quantum circuit that prepares $\ket{\eta}$ from its classical description, i.e., the list of gates in $\mathcal{U}_x$, the unary clock encoding, and the input $x$.
As we will see in the sequel, this family of guiding states is necessarily different than those previously considered in the literature~\cite{gharibian2023dequantizing,cade2023improved}, thus requiring a constructive proof of \cl{BQP} containment in order to establish \clw{BQP}{completeness}.
It is unsurprising that the history state can serve as a guiding state, given that it completely encodes the \cl{BQP} computation.
Equivalently, it is trivial to know the history state of a given \cl{BQP} circuit and thus by fixing a time encoding, we can prove there exists a local Hamiltonian whose ground state is close to this history state.

Having established that the history state is a valid guiding state, proving \clw{BQP}{completeness} follows by concluding valid choices of the parameters $a$ and $b$, appropriately renormalising the Hamiltonian, preparing the guiding state, and performing quantum energy estimation.
In what follows we demonstrate how these steps can be carried out for more specific families of guiding states; since the history state is recovered as a special case, we defer explicit treatment of completeness. 

We now identify which structural properties of the history state are essential for preserving hardness.
Moreover, we can determine which components of the circuit history can be modified or omitted, such as gates that act trivially or are sufficiently close to the identity, without compromising the validity of the reduction.
For example, if a large number of gates in $\mathcal{U}_x$ are close to the identity, then we can consider a \emph{windowed} version of the history state that only includes time steps where non-trivial gates are applied.
On the other hand, if there is a region of sequential near-identity gates, we can consider a windowed history state that only focuses on this region.
Alternate choices of windowing are discussed in \cref{app:alt_idling_reductions}.
In the following subsection we formalise this idea by introducing one possible minimal representation that suffices as a framework for constructing alternative guiding states while retaining the necessary overlap and structural properties required for \clw{BQP}{hardness}.

\subsection{Guiding State Construction Framework}\label{sec:guiding_state_framework}
We now present a general framework for constructing guiding states that retain the necessary overlap and structural properties required for \clw{BQP}{hardness}.
The key idea is to introduce a windowing technique that restricts attention to carefully chosen subspaces of the full history state.
This perspective allows us to identify which components of the circuit history can be modified or omitted, such as gates that act trivially or are sufficiently close to the identity, without compromising the validity of the reduction.
By systematically characterising these permissible simplifications, we obtain a method for constructing alternative guiding states that retain the necessary overlap and structural properties required for \clw{BQP}{hardness}.
We begin from the same setting as in \cref{sec:history_state_hardness}, with $\mathcal{U}_x = U_K \cdots U_1$ a quantum circuit composed of $K = \poly{|x|}$ gates acting on $n + m$ qubits, where $n \coloneqq |x|$ is the number of input qubits and $m \coloneqq \poly{|x|}$ is the number of ancilla qubits.
In this part we formally define the components of the Feynman-Kitaev construction.

\subsubsection{Selecting a Clock Encoding}
Suppose the clock register is encoded via a mapping $\mu : \mathbb{N} \to \B^{p(n)}$, where $p(n)$ is a polynomially-bounded function yielding an integer value.
The mapping $\mu$ defines a chosen decimal-to-binary encoding of the clock register.
Let $\boldsymbol{0}, \boldsymbol{1}, \dots, \boldsymbol{t}, \dots, \boldsymbol{K}$ denote the decimal values of the clock states.\footnote{We use bold notation to denote decimal values of clock states and non-bold notation for indexing time steps, i.e., $\ket{\boldsymbol{1}} + \cdots + \ket{\boldsymbol{K}} = \sum_{t=1}^{K} \ket{\boldsymbol{t}}$.}
We denote each clock qubit as $C_k$ (collectively as $C$), e.g.,
\begin{equation*}
    \ket{\boldsymbol{t}}_C = \bigotimes_{k \in [p(n)]} \ket{\mu_k(\boldsymbol{t})}_{C_k} = \ket{\mu(\boldsymbol{t})}_C,
\end{equation*}
where the binary element $\mu_k(\boldsymbol{t})$ is the $k$-th bit of the binary representation of $\boldsymbol{t}$ under encoding $\mu$.

We define the Hamiltonian $\hat{H}_\mu = \Delta H_{\mu} + V_\mu$ over the space $(\mathbb{C}^2)^{\otimes(n + m + p(n))}$.
Let us denote each register as: the input register $I$ (of size $n$), the ancilla register $A$ (of size $m$), and the clock register $C$ (of size $p(n)$, where $p(n)$ is a polynomially-bounded function).
It follows that
\begin{align*}
    H_{\mu} &\coloneqq H_{\tn{in},\mu} + H_{\tn{prop},\mu} + H_{\tn{clock}, \mu}, &
    V_\mu &\coloneqq H_{\tn{out},\mu}, 
\end{align*}
with each term defined as
\begin{align*}
    H_{\tn{in}} &= \Pi^{(\tn{in})}_{I,A} \ketbra{\boldsymbol{0}}_{C}, 
        & H_{\tn{out}} &= \ketbra{0}_1 \ketbra{\boldsymbol{K}}_{C},\\
    H_{\tn{clock}} &= \Pi^{(\tn{clock})}_{C}, 
        & H_{\tn{prop}} &= \sum_{t=1}^{K} H_{\tn{prop},t},
\end{align*}
where 
\begin{align*}
    \Pi^{(\tn{in})}_{I,A} &= \sum_{j \in [n]} \ketbra{\bar{x}_j}_j + \sum_{j \in [m]} \ketbra{1}_j, \\
    H_{\tn{prop},t} &= \ketbra{\boldsymbol{t}}_{C} + \ketbra{\boldsymbol{t+1}}_{C} \\
        &\qquad- U_{t} \ketbra{\boldsymbol{t}}{\boldsymbol{t-1}}_{C} + U_{t}^\dagger \ketbra{\boldsymbol{t-1}}{\boldsymbol{t}}_{C}.
\end{align*}
We have not yet specified $\Pi^{(\tn{clock})}$, as this will depend on the specific clock encoding $\mu$.
Moreover, it is not necessarily true that any term above is local; this will also depend on the choice of clock encoding $\mu$.
Letting $\chi$ denote the labels $L = \{\tn{in}, \tn{prop}, \tn{clock}, \tn{out}\}$, each term $H_{\chi,\mu}$ has the Hilbert space decomposition $\mathcal{H}^\chi = \mathcal{L}_{-}^\chi \oplus \mathcal{L}_{+}^\chi$, where $\mathcal{L}_{-}^\chi$ is the null space of $H_{\chi,\mu}$ and $\mathcal{L}_{+}^\chi$ is its orthogonal complement.
The common null space of $H_{\mu}$ is given by $\mathcal{L}_{-} = \bigcap_{\chi \in L \setminus \{\tn{out}\}} \mathcal{L}_{-}^\chi$ which is spanned by the history state
\begin{equation*}
    \ket{\eta_\mu} \coloneqq \frac{1}{\sqrt{K+1}} \sum_{t=0}^{K} \ket{\varphi_t}_{I,A} \ket{\boldsymbol{t}}_C.
\end{equation*}
Note that $H_{\mu}$ is frustration-free by construction, i.e., $\mel{\eta_\mu}{H_{\chi,\mu}}{\eta_\mu} = 0$ for all $\chi \in L \setminus \{\tn{out}\}$, and has a non-degenerate ground space $\mathcal{L}_{-}$.
Furthermore, the spectral gap $\gamma$ of $H_{\mu}$, i.e., the difference between the smallest and second smallest eigenvalues, is lower bounded as $\Omega(1/K^3)$.

\subsubsection{Perturbative Analysis}
The energy penalty term $V_\mu$ perturbs the eigensystem of the unperturbed Hamiltonian $H_\mu$ so that its low-energy spectrum reflects the acceptance probability of the circuit $\mathcal{U}_x$.
We analyse this perturbation using first-order perturbation theory and the Schrieffer-Wolff (SW) transformation to control both the low-energy eigenvalues and eigenstates of the perturbed Hamiltonian $\hat{H}_\mu = \Delta H_\mu + V_\mu$.

Let $P$ denote the projector onto the ground space of $H_\mu$.
Recall that $H_\mu$ is frustration-free and has a unique ground state given by the history state $\ket{\eta_\mu}$, hence $P = \ketbra{\eta_\mu}{\eta_\mu}$.
By construction, the expectation value of the penalty term on the history state satisfies
\begin{equation*}
    \mel{\eta_\mu}{V_\mu}{\eta_\mu} = \frac{\Pr[\tn{reject}(\mathcal{U}_x)]}{K+1}.
\end{equation*}
Thus, to first order, the ground-state energy of $\hat{H}_\mu$ is proportional to the rejection probability of the underlying quantum circuit.

We employ the Schrieffer-Wolff transformation~\cite{bravyi2016complexity} to analyse how the perturbation $V_\mu$ affects both the low-energy eigenvalues and eigenstates of $\hat{H}_\mu$.
This transformation block-diagonalises $\hat{H}_\mu$ with respect to the decomposition $\mathcal{H} = \mathcal{L}_{-} \oplus \mathcal{L}_{+}$, where $\mathcal{L}_{-}$ is the low-energy subspace of $H_\mu$.
Concretely, one constructs an anti-Hermitian operator $S$ such that the unitary transformation ${\rm e}^{S}$ approximately eliminates couplings between $\mathcal{L}_{-}$ and its orthogonal complement.
Truncating the series
\begin{equation*}
    \e^{-S} \hat{H}_\mu \e^{S} = \hat{H}_\mu + \comm{\hat{H}_\mu}{S} + \frac{1}{2}\comm{\comm{\hat{H}_\mu}{S}}{S} + \cdots
\end{equation*}
and projecting onto $\mathcal{L}_{-}$ yields an effective Hamiltonian
\begin{equation*}
    H_{\mathrm{eff}} \coloneqq P \bigl({\rm e}^{-S} \hat{H}_\mu {\rm e}^{S}\bigr) P,
\end{equation*}
whose spectrum approximates the low-energy spectrum of $\hat{H}_\mu$ up to inverse-polynomial error, provided the perturbation strength is chosen appropriately (see the proof of \cref{lma:history_state_guiding} in \cref{app:bqp_hardness_framework}).

This analysis allows us to choose the scaling parameter $\Delta$ such that $V_\mu$ does not significantly mix $\mathcal{L}_{-}$ with $\mathcal{L}_{+}$, ensuring that the ground state of $\hat{H}_\mu$ remains close to the history state $\ket{\eta_\mu}$ and preserving the promise gap.
Consequently, \cref{lma:history_state_guiding} extends directly to the perturbed Hamiltonian $\hat{H}_\mu$ for any clock encoding $\mu$.
The qualitative effect of scaling and perturbation on the spectrum is illustrated in \cref{fig:energy_spectrum}.

\begin{figure}[!ht]
    \centering
    \begin{tikzpicture}
        \pic[]{altspectrum};
    \end{tikzpicture}
    \caption{
        Visualisation of spectral gap changes under scaling and perturbations. 
        The vertical axis represents energy levels, while horizontal bars are for reference. 
        The left diagram represents the scaled system $\Delta H$ with gap $\Delta \gamma$. 
        The right diagram illustrates the effect of a perturbation $V$ on $\Delta H$, where the gap may shrink to $\hat{\gamma}$.}
        \label{fig:energy_spectrum}
\end{figure}

\subsubsection{Windowed History States}
We now introduce a windowing technique that restricts attention to carefully chosen subspaces of the full history state.
This perspective allows us to identify which components of the circuit history will allow for the construction of alternative guiding states while retaining the necessary overlap and structural properties required for \clw{BQP}{hardness}.
As we have just seen, the choice of clock encoding $\mu$ can affect several aspects of the construction, including the locality of the Hamiltonian terms and the structure of the history state.
It is clear that non-trivial choices of $\mu$ can lead to non-local Hamiltonian terms, which are not desirable for understanding the complexity of local Hamiltonian problems.
Thus we proceed under the assumption that $\mu$ is chosen such that all terms of $\hat{H}_\mu$ are $k$-local for some constant $k = O(1)$, though we do not place further restrictions on $\mu$, yet.

The following lemma formalises the windowing technique by considering a \emph{pre-idled} \cl{BQP} circuit, where the first $N$ gates are identity operations.
This allows us to decompose the history state into two distinct regions: an initial idle region where the state remains unchanged, and a subsequent active region where the computation evolves:
\begin{multline*}
    \ket{\eta_\mu} = \frac{1}{\sqrt{K + N + 1}} \sum_{t=0}^{N} \ket{x,0^m}\ket{\mu(\boldsymbol{t})}  \\ 
        + \frac{1}{\sqrt{K + N + 1}}\sum_{t=N+1}^{K+N} \ket{\varphi_t}\ket{\mu(\boldsymbol{t})} .
\end{multline*}
For a sufficiently large idle region, we can consider a windowed version of the history state that focuses on the initial idle region.
More specifically, we characterise guiding states via uniform superpositions over subsets of the idle region of the history state.
The key set is defined as
\begin{equation}\label{eq:R-state-subset}
        R_{\sigma, \mu'}^\nu = \{\sigma\} \times \{\mu'(\boldsymbol{t}) : t\in[\nu]\},
\end{equation}
where $\sigma \in \B^{|x| + m}$, $\nu \in \mathbb{N}$ and $\mu' : \mathbb{N} \to \B^{r(|x|)}$, for some polynomially-bounded function $r$.
Guiding states whose computational basis state support overlap sufficiently with this set can be shown to have non-negligible overlap with the ground state of $\hat{H}_\mu$.
In later sections and in the Appendices, we will explore more exotic choices that extend beyond computational basis states~\cite{waite2025guided}.

Though it is clear that we do not necessarily require a one-to-one correspondence between the guiding state and the idle region, as we will see later, this construction provides a useful framework for identifying valid guiding states.
The following lemma formalises this idea.

\begin{restatable}{lemma}{lmageneralguidingstate}\label{lma:general_guiding_state}
    Let $\hat{H}_\mu = \Delta H_\mu + V_\mu$ be a Hamiltonian with a one-dimensional ground space spanned by $\ket{\hat{\phi}_0}$.
    Let the ground state of $H_\mu$ be the history state
    \begin{equation*}
        \ket{\eta_\mu} = \frac{1}{\sqrt{K+1}} \sum_{t=0}^{K} U_{t} \cdots U_{1} \ket{x,0^m}\ket{\mu(\boldsymbol{t})},
    \end{equation*}
    where $m = \poly{|x|}$ and $K = \poly{|x|}$. 
    Assume the first $N$ unitaries are identity gates, i.e., $U_{t} = I$ for $t\in[N]$.

    Suppose
    \begin{equation*}
        \ket{\hat{R}_{\sigma, \mu'}^\nu} \coloneqq \frac{1}{\sqrt{|R_{\sigma, \mu'}^\nu}|} \sum_{x \in R_{\sigma, \mu'}^\nu} \ket{x},
    \end{equation*}
    is a state defined over a set
    \begin{equation}
        R_{\sigma, \mu'}^\nu = \{\sigma\} \times \{\mu'(\boldsymbol{t}) : t\in[\nu]\},
    \end{equation}
    where $\sigma \in \B^{|x| + m}$, $\nu \in \mathbb{N}$ and $\mu' : \mathbb{N} \to \B^{r(|x|)}$, for some polynomially-bounded function $r$ yielding an integer value.
    
    Then, there exists a choice of $R_{\sigma, \mu'}^\nu$ such that
    \begin{equation*}
        F_{\hat{R}_{\sigma, \mu'}^\nu, \g} \geq 1 - \frac{1}{q(|x|)},
    \end{equation*}
    for some polynomial $q$.
\end{restatable}

The proof of \cref{lma:general_guiding_state} is provided in \cref{app:bqp_hardness_framework}.

\subsubsection{Scope of the Framework}
Combining the above results, we have established a general framework for constructing guiding states that retain the necessary overlap and structural properties required for \clw{BQP}{hardness}.
It can be shown that for an arbitrary promise problem in \cl{BQP}, there exists a choice of clock encoding $\mu$, yielding a $k$-local Hamiltonian $\hat{H}_\mu$, and a set $R_{\sigma, \mu'}^\nu$ such that the corresponding guiding state $\ket{\hat{R}_{\sigma, \mu'}^\nu}$ has non-negligible overlap with the ground state of $\hat{H}_\mu$ and choices of $a,b \in [0,1]$ with $b-a \geq 1/\poly{n}$ such that the \sc{$R_{\sigma, \mu'}^\nu$ Guided Local Hamiltonian} problem is \clw{BQP}{hard}.
Moreover, this reduction is Karp as all steps of the construction can be computed in polynomial time.

\begin{restatable}{theorem}{thmRstatehardness}\label{thm:R_state_hardness}
    For any $\delta \in (0,1-1/\poly{n})$, there exists $a,b \in [0,1]$ with $b-a \geq 1/\poly{n}$ such that the \sc{$R_{\sigma, \mu'}^\nu$ Guided Local Hamiltonian} problem is \clw{BQP}{hard}.
\end{restatable}

The proof of \cref{thm:R_state_hardness} is outlined in \cref{app:bqp_hardness_framework}.
Information about the structure and form of potential guiding states is encoded in the set $R_{\sigma, \mu'}^\nu$.
For example:
\begin{inparaenum}[(a)]
    \item states with minimal structure: identifying the tightest state that recovers \clw{BQP}{hardness},\footnote{This is non-obvious as restricting states in certain manners does not necessarily imply that \clw{BQP}{hardness} is preserved.}
    \item gadget states: additions to $R$ that allow the use of perturbative gadgets reductions,
    \item states with interesting attributes: restricted entanglement, no entanglement, or close to a Haar random state, and
    \item physically motivated states: states relating to what we see in quantum computing ansatz and physical considerations on systems.
\end{inparaenum}
However, certain constructions may yield unphysical attributes or fail to retain the required properties for efficient state preparation; for example, the history state expressed in the same manner as semi-classical subset states will likely not be an efficient encoding~\cite{gharibian2023dequantizing}.
An additional attribute we explore further in \cref{sec:fh_model} is the extension of the framework to fermionic systems, specifically the Fermi-Hubbard model.
It remains an open question to expand these results to more general electronic structure Hamiltonians.

\subsection{Reduction to 2-Local Hamiltonians}\label{sec:2-local_reduction}
Recall that the locality of the Hamiltonian constructed in \cref{thm:R_state_hardness} is determined by the encoding $\mu$.
We prove how perturbative gadgets can be used to reduce the locality of $\hat{H}_\mu$ to $2$-local while preserving the overlap properties of the guiding state.
More specifically, we perform polynomial-time reductions that preserve \clw{BQP}{hardness} for the \sc{$R_{\sigma, \mu'}^\nu$ Guided Local Hamiltonian} problem and the non-negligible overlap between the guiding state and the ground state of the resultant Hamiltonian.
Such reductions entail the following criteria: 
\begin{inparaenum}[(i)] 
    \item efficient classical description of the resultant guiding state, and
    \item sufficient overlap between the guiding state and the ground state of the resultant Hamiltonian.
\end{inparaenum}
Note that there is no reason to preserve the family of guiding states under these reductions, only the above two criteria; however, as we will see later, it is possible to preserve the family of guiding states in certain cases.
To prove the second condition we will invoke the following result.

\begin{lemma}[{\cite[Lemma~3]{waite2025guided}}]\label{lma:geometric-lemma}
    Let $\ket{a}, \ket{b}, \ket{c} \in (\mathbb{C}^2)^{\otimes n}$ be normalised states, such that $\norm{\ket{a} - \ket{b}} \leq x$ and $\abs{\braket{b}{c}}^2 \geq y$. Then:
    \begin{equation*}
        \begin{cases}
            (\sqrt{y} - x)^2 \leq \abs{\braket{a}{c}}^2 \leq (\sqrt{y} +x)^2 & \tn{if}~ x \leq \sqrt{y}, \\[0.25cm]
            0 \leq \abs{\braket{a}{c}}^2 \leq (\sqrt{y} + x)^2 & \tn{if}~ x > \sqrt{y}.
            \end{cases}
    \end{equation*}
\end{lemma}

To perform these reductions, we leverage the concept of Hamiltonian simulation~\cite{cubitt2018universal,bravyi2016complexity} and perturbation gadgets~\cite{oliveira2008complexity,kempe2006complexity}.
Consider two Hamiltonians $H_A$ and $H_B$ acting on Hilbert spaces $\mathcal{H}_A$ and $\mathcal{H}_B$, respectively.
An isometric encoding map $\mathcal{E} : \mathcal{H}_A \to \mathcal{H}_B$ and Hamiltonian $H_B$ are said to be an $(\eta,\epsilon)$-simulator for $H_A$ if there exists an isometry $\mathcal{V} : \mathcal{H}_A \to \mathcal{H}_B$ such that: the image of $\mathcal{V}$ is the low-energy subspace of $H_A$, $\norm{H_A - \mathcal{V}^\dagger H_B \mathcal{V}} \leq \epsilon$, and $\norm{\mathcal{E} - \mathcal{V}} \leq \eta$.
Equivalently, the low-energy subspace of $H_B$ approximates $H_A$, up to controlled errors.
To prove our hardness results we always take $\eta, \epsilon = O(1/\poly{n})$.

A consequence of this simulation is that the $j$-th smallest eigenvalues of $H_A$ and $H_B$ are close, satisfying $\norm{\lambda_j(H_B) - \lambda_j(H_A)} \leq \epsilon$, where $\lambda_j(H)$ denotes the $j$-th smallest eigenvalue of $H$~\cite[Lemma 1]{bravyi2016complexity}.
Additionally, it can be shown that the ground states of $H_A$ and $H_B$ are close under the encoding~\cite[Lemma 2]{bravyi2016complexity}.

From a complexity-theoretic perspective, this notion of simulation ensures that any algorithm capable of solving the local Hamiltonian problem for $H_B$ can be used to solve it for $H_A$ with comparable accuracy.
This hardness preservation is formalised in the following proposition.

\begin{proposition}[{\cite[Proposition~3]{waite2025guided}}]\label{prop:simulator-hardness}
    Let $H_B$ be a $(\eta,\epsilon)$-simulator for $H_A$.
    Then $H_B$ is at least as hard as $H_A$.
\end{proposition}

\begin{proof}
    Take an instance of $H_A$ to be defined as $x\coloneqq(H_A, a,b)$ such that $b-a \geq 1/\poly{n}$.
    Since $H_B$ is a $(\eta,\epsilon)$-simulator for $H_A$, we set the parameters $b' = b - \epsilon$ and $a' = a + \epsilon$.
    Setting $\epsilon < (b-a)/2$ ensures that $b'-a' = b-a - 2\epsilon \geq 1/\poly{n}$, meeting the criterion for a valid instance of $H_B$.
    It is not hard to see that in the event of a \sc{yes} case for $H_A$, i.e.
    $\lambda(H_A) \leq a$ then it must be that $\lambda(H_B) \leq a'$.
    Similarly the converse holds for the \sc{no} case.
\end{proof}

To reduce the locality of $\hat{H}_\mu$, to $2$-local, we use repeated applications of the gadget constructions of Ref.~\cite{oliveira2008complexity}.
Specifically, the \textsl{subdivision} gadget is used to reduce the locality of $k$-local Hamiltonians to $3$-local Hamiltonians, and the \textsl{3-to-2-local} gadget is used to further reduce $3$-local Hamiltonians to $2$-local Hamiltonians.
For a $k$-local Hamiltonian we require $O(\log k)$ applications of the subdivision gadget followed by a single application of the \textsl{3-to-2-local} gadget.
The consequence of these reductions is a polynomial blow-up in interaction strengths and the number of qubits.

\begin{theorem}
    Let $\mu$ be a mapping such that $\hat{H}_\mu$ is a $O(1)$-local Hamiltonian with a one-dimensional ground space spanned by $\ket{\hat{\phi}_0}$ and a spectral gap $\hat{\gamma}$ (of the form defined in \cref{thm:R_state_hardness}).
    Suppose there exists a choice of $R_{\sigma, \mu'}^\nu$ such that $F_{\hat{R}_{\sigma, \mu'}^\nu, \hat{\phi}_0} \geq 1 - O(1/\poly{|x|})$.
    Then, there exists a polynomial-time reduction from $\hat{H}_\mu$ to a $2$-local Hamiltonian $\tilde{H}$ with a one-dimensional ground space spanned by $\ket{\tilde{\phi}_0}$ such that the state $\ket{\bar{R}_{\sigma, \mu'}^\nu} = \ket{\hat{R}_{\sigma, \mu'}^\nu}\ket{0^{\ell(|x|)}}$, for some polynomial $\ell$, satisfies
    \begin{equation*}
        F_{\bar{R}_{\sigma, \mu'}^\nu, \tilde{\phi}_0} \geq 1 - \frac{1}{q(|x|)},
    \end{equation*}
    for some polynomial $q$.
\end{theorem}

\begin{proof}
    The proof is a consequence of the gadget techniques of Ref.~\cite{oliveira2008complexity} applied via the Schrieffer-Wolf transformation~\cite{hamiltonianjungle2023}, and the use of \cref{lma:geometric-lemma} to bound the overlap between the guiding state and the ground state of the resultant Hamiltonian.
    The description of the resultant guiding state follows from the Cartesian product of the original guiding state support with the ancillae support; it is easy to see that this description is polynomial in size.
\end{proof}

Notice that even though a series of perturbation gadget reductions have been used, we have only attached a polynomial number of $\ket{0}$ ancillae to the system.
Additionally, due to the style of the reduction, we must end again with a manual renormalisation of the Hamiltonian to ensure $\norm{H} \leq 1$.

\begin{restatable}{corollary}{GadgetRstateCorollary}
    For any $\delta \in (0,1-1/\poly{n})$, there exists $a,b \in [0,1]$ with $b-a \geq 1/\poly{n}$ such that the \sc{$\bar{R}_{\sigma, \mu}^\nu$ Guided Local Hamiltonian} problem is \clw{BQP}{hard} for $2$-local Hamiltonians.
\end{restatable}

More exotic gadget constructions can be considered, such as those that reduce to geometrical structures~\cite{cubitt2016complexity,piddock2017complexity,oliveira2008complexity} or specific interaction types~\cite{piddock2017complexity,schuch2009computational,biamonte2008realizable}.
However, these constructions often introduce additional complexities that may affect the structure of the guiding state and break outside a given class of interest.
Specifically, the inclusion of quantum gates and/or constant-size isometries.
This is a consequence of requiring the attachment of polynomially-many ancillae such as $\ket{+}$ or $\ket{\Psi^-}$ states.
It is easy to see that states of the form $\ket{+}^{\otimes p(n)}$, for some polynomial $p$, are not easily efficiently described via computational basis state support.
It is an open question whether there exists a geometrical reduction that preserves the state type.

\subsubsection{Reduction to Lattice Geometries}\label{sec:lattice_reduction}
A straightforward reduction exists from $5$-Local Hamiltonians on a sparse graphical geometry to $2$-Local Hamiltonians on a square lattice~\cite{oliveira2008complexity}.
This requires a polynomial number of ancillae and extra (\Gate{Swap}) gates that interleave the original gates; \cref{fig:circuit} gives a small example of this idea.
From this geometrical configuration, the subsequent local Hamiltonian will have qubits interacting only in a localised region.
However, for unspecified clock mappings, this reduction may not be valid; those mappings that produce high vertex-degree graphs will not be appropriate.
This is because the spatially sparse geometry required for the defined local Hamiltonian requires very few long-range interactions.

The perturbative gadgets of Ref.~\cite{oliveira2008complexity} are designed to reduce the locality and geometry of the Hamiltonian.
A given circuit is first transformed to a sparse geometry, and then the locality is reduced to $2$-local.
Subsequent gadgets can be applied to reduce the degree of the interaction graph and the number of edge-crossings, ultimately yielding a $2$-local Hamiltonian on a square lattice.
At each step of the reduction, the guiding state can be preserved by attaching $\ket{0}$ ancillae, thereby retaining the overlap properties required for \clw{BQP}{hardness}.
Furthermore, only at most a polynomial number of gadgets are required, ensuring the overall reduction remains efficient.
An alternative approach to geometrical reductions was proposed by~\citet{zhou2021strongly}, however, this relies on the use of a local encoding that can destroy any physical relevance of the guiding state.
We proceed using the more general case of the reduction.

\begin{figure}[!ht]
    \centering
    \begin{tikzpicture}
        \pic[scale=0.9]{circuit};
    \end{tikzpicture}
    \caption{
        A small example of a circuit transformation from a general circuit to a circuit laid out on a sparse geometry. 
        In the latter, each qubit is acted on by at most three gates. 
        The dots represent the rest of the circuit, not shown; the pattern continues for all gates.}
    \label{fig:circuit}  
\end{figure}

\begin{proposition}[\cite{oliveira2008complexity,waite2025complexitya}]\label{prop:sparse_geometry}
    Let $\mathcal{U}_x$ be a sequence of $K = \poly{|x|}$ unitary gates, laid out on a sparse geometry, for instance $x \in \B^*$.
    There is an efficient algorithm that produces a new sequence $\mathcal{Q}_x$ comprising a sequence of $Q = \Theta(K)$ unitary gates and $\Theta(K)$ ancillae qubits such that the circuit is laid out on a sparse geometry.
    The new circuit preserves the acceptance statistics of the original circuit.
\end{proposition}

Following the procedure outlined in Ref.~\cite{oliveira2008complexity}, it is possible to preserve the family of guiding states $\bar{R}_{\sigma, \mu}^\nu$ and yield hardness results for $2$-local Hamiltonians on square lattice geometries; this can be improved to triangular lattices following either of Refs.~\cite{waite2025complexitya,cubitt2016complexity}.

\begin{corollary}
    For any $\delta \in (0,1-1/\poly{n})$, there exists $a,b \in [0,1]$ with $b-a \geq 1/\poly{n}$ such that the \sc{$\bar{R}_{\sigma, \mu}^\nu$ Guided Local Hamiltonian} problem is \clw{BQP}{hard} for $2$-local Hamiltonians on square or triangular lattice geometries.
\end{corollary}

Unfortunately, this reduction will not hold for the physically-motivated guiding states we consider later in the next section.
The reason is due to the form of the $H_{\text{clock}}$ Hamiltonian terms we use.

\section{State Type Variations and Implications}\label{sec:state_type_variations}
We focus on representative choices that yield guiding states with minimal structure, gadget-type constructions, and states possessing physically or computationally meaningful features.
For each case, we verify the conditions of \cref{lma:general_guiding_state}, thereby establishing \clw{BQP}{hardness} of the corresponding \sc{Guided Local Hamiltonian} instances.

We do not claim these constructions are exhaustive.
Alternative clock encodings, local Hamiltonian terms, and choices of $R_{\sigma, \mu'}^\nu$ may give rise to further guiding state families.
For example, Hartree-Fock states, simple product states defined by a product of $X$ gates, are a trivial but naturally occurring family of guiding states that can be shown to yield \clw{BQP}{hardness} via the framework, yielding weaker overlap guarantees~\cite{richter2007two}.
Moreover, exhibiting a single non-trivial guiding state within a family suffices to show the family does not trivialise the problem --- more complex or structured members of the same family will also retain \clw{BQP}{hardness}, albeit potentially with weaker overlap or fidelity guarantees.
Accordingly, we emphasise identifying minimally structured examples that preserve computational hardness.

Although guiding state families can exhibit mutual representations, this does not render the corresponding hardness results redundant.
Our aim is not to classify these families as abstract quantum states, but to show that distinct and physically meaningful \emph{input models} for guiding states each suffice to support \clw{BQP}{hardness}.
Hardness is a property of the problem encoding, including the form of the guiding state description and access model, rather than of state equivalence under unitary transformations.
In particular, our constructions proceed by modifying the clock encoding within the Feynman-Kitaev framework, thereby inducing different classical descriptions and structural constraints on the guiding state, rather than by applying unitary transformations to a fixed instance.
As illustrated in \cref{fig:hardness-growth}, our hardness results identify representative guiding states that serve as minimal origins for broader families whose members inherit the same complexity classification, without requiring equivalence of their classical descriptions.

\begin{figure}
    \centering
    \begin{tikzpicture}
        \pic[]{hardnessgrowth};
    \end{tikzpicture}
    \caption{
        Schematic illustration of representative guiding state families within the complexity landscape of the \sc{Guided Local Hamiltonian} problem.
        Each black point denotes a specific guiding state construction established in this work; here, a fixed-weight state $\ket{\xi_{\rm fw}} \in \mathcal{F}(\hat{X})$ and a matrix product state $\ket{\xi_{\rm MPS}} \in \mathcal{F}({\rm MPS}^*)$ --- for which \clw{BQP}{hardness} (and, where applicable, \clw{BQP}{completeness}) is proven.
        The shaded cones indicate upward regions within each family: more structured or expressive members of the same family inherit the corresponding hardness properties, though quantitative parameters such as overlap or fidelity may degrade.
        Intersections of cones represent families that are equivalent at the level of computational hardness, without implying equivalence of their classical descriptions or preparation procedures.
        The placement of the origin points emphasises that our results identify minimal, representative constructions rather than exhaustive characterisations of the families.
    }
    \label{fig:hardness-growth}
\end{figure}

For comparative purposes, we formally define the two main state families for which \clw{BQP}{hardness} is already known~\cite{gharibian2023dequantizing,cade2023improved,waite2025guided}.

\subparagraph{Semi-Classical Subset States.}
Let $S \subseteq \B^n$ be a subset of binary strings.
The \emph{subset state} over $S$ is defined as
\begin{equation*}
    \ket{\hat{S}} \coloneqq \frac{1}{\sqrt{\abs{S}}}\sum_{x\in S} \ket{x}.
\end{equation*}
A subset state is completely defined by the subset $S$.
A more general state defined by a subset $S$ is one of the form $\ket{S} = \sum_{x\in S} \alpha_x \ket{x}$.
To completely specify this state, it suffices to define the set of pairs $\{(\alpha_x,x)\}_{x\in S}$.
A \emph{semi-classical subset state} (SCSS) is a subset state over a subset $C \subset \B^n$ such that $|C| = \poly{n}$, i.e.,
\begin{equation*}
    \ket{\hat{C}} \coloneqq \frac{1}{\sqrt{\abs{C}}}\sum_{x\in C} \ket{x}.
\end{equation*}
A classically efficient description of a semi-classical subset state is a list of binary strings in $C$.

\subparagraph{Semi-classical Encoded Subset State.}
A simple modification to semi-classical subset states is the \emph{semi-classical encoded subset state} (SCESS).
An SCESS is a subset state over $C \subset \B^n$ such that $|C| = \poly{n}$, where we allow for a set of isometries $\mathcal{V} = \{V_j\}_j$, where for each $j$, $V_j : \mathbb{C}^2 \mapsto (\mathbb{C}^2)^{\otimes m_j}$ such that $m_j = O(1)$.
Thus, the SCESS over $C$ with $\mathcal{V}$ is defined as
\begin{equation*}
    \ket{C_{\mathcal{V}}} \coloneqq \frac{1}{\sqrt{\abs{C}}}\sum_{x\in C} \bigotimes_{j} V_j\ket{x_j}.
\end{equation*}
We say that the computational basis state images lie in $\B^M$, where $M = \sum_j m_j$.
The SCESS is completely defined by the subset $C$ and the set of isometries $\mathcal{V}$.
Furthermore, there is a classically efficient description of an SCESS, which is a list of binary strings in $C$ and a list of isometries in $\mathcal{V}$.
In the appendices we consider further generalisations of SCESSs.

We now introduce several new families of guiding states that fit within the framework established in \cref{sec:bqp_hardness}.
For each family, we demonstrate how the conditions of \cref{lma:general_guiding_state} are satisfied, thereby establishing \clw{BQP}{hardness} for the corresponding \sc{Guided Local Hamiltonian} problem.
Additionally, we provide insights into the physical and computational implications of these state choices.
Rather than concluding individual theorems for each state type, we state a general theorem at the end of the section.

\subsection{Fixed-Weight States}\label{sec:fixed_weight_states}
The first state type variation we consider is the fixed-weight state.
A fixed-weight state is a state that is a superposition of computational basis states with Hamming weight $k\in [n]$.
Let $X_{n,k}$ denote a subset of binary strings of length $n$ with Hamming weight $k$.
We formally denote these states as
\begin{equation*}
    \ket{X_{n,k}} \coloneqq \sum_{x\in X_{n,k}} \alpha_{x}\ket{x}.
\end{equation*}
For a given $k$, there are at most $\binom{n}{k}$ computational basis states in the superposition.
We also define \emph{uniform fixed-weight states}, denoted $\ket{\hat{X}_{n,k}}$.
Note that the coefficients $\alpha_x$ can be zero for certain $x \in X_{n,k}$.
Let $\mathcal{F}(X_{n,k})$ and $\mathcal{F}(\hat{X}_{n,k})$ denote the sets of all fixed-weight states and uniform fixed-weight states, respectively.

Fixed-weight states carry the physical interpretation of superposition states with a fixed number of excitations.
Such objects are natural to consider in the context of quantum chemistry and quantum many-body physics, and importantly, do not rule out any complicated structure, such as restricting the amount of entanglement.
It is known that entangled states~\cite{dur2000three}, certain atomic structures~\cite{weng2025high} and Hamiltonian ground states can be represented as fixed-weight states~\cite{gharibian2012hardness,yu2025samplebased}.
These states also allow for natural quantum generalisations of concepts in parameterised complexity theory~\cite{bremner2025parameterized,bremner2022quantum}.

Using the framework established in \cref{sec:bqp_hardness}, we can prove that the \sc{Fixed Weight Guided Local Hamiltonian} problem is \clw{BQP}{hard} by adapting the clock construction.
In the unary encoding case we have
\begin{equation*}
    H_{\tn{clock}} = \sum_{t=1}^{K-1} \ketbra{0}_{C_t}\ketbra{1}_{C_{t+1}},
\end{equation*}
where $C_t$ is the $t$-th clock qubit.
This penalises incorrect clock states such as $\ket{101}$ in the case where $T=3$.
For fixed-weight clock states, this clock-penalising term is not sufficient.
Since we only want weight-$1$ clock states, we must penalise states that have a Hamming weight greater than $1$.
For clarity, the progression of the clock state from decimal-to-binary is as follows:
\begin{align*}
    \ket{\boldsymbol{0}} = \ket{100\dots 0}, \ket{\boldsymbol{1}} = \ket{010\dots 0}, \dots, \ket{\boldsymbol{K}} = \ket{000\dots 1}.
\end{align*}
The register $C$ is therefore defined over $K+1$ qubits.
We denote this encoding as the \emph{fixed-weight clock encoding} and the corresponding mapping as $w$.

In this scenario, we rely on a long-range interaction, \footnote{Note that the energy of $H_{\tn{clock}}$, parameterised by Hamming weight $w$, is: $E_{\tn{clock}}(w) = (w-1)^2$.}
\begin{equation*}
    H_{\tn{clock}} = 2\sum_{t<t'} \ketbra{1}_{C_t}\ketbra{1}_{C_{t'}} - \sum_{t=1}^{K+1} \ketbra{1}_{C_t} + I.
\end{equation*}
Note that this term is still $2$-local, as it only acts on two clock qubits at a time.

Thus, using the additional following local Hamiltonian terms:
 \begin{align*}
    H_{\tn{in}} = \Pi^{(\tn{in})}_{I,A}\ketbra{1}_{C_1}, 
        &\quad H_{\tn{out}} = \ketbra{0}_1\ketbra{1}_{C_{K+1}}, \\
    H_{\tn{prop}} &= \sum_{t=1}^{K-1} H_{\tn{prop},t},
\end{align*}
where
\begin{align*}
    H_{\tn{prop},t} &= \ketbra{1}_{C_t}\ketbra{0}_{C_{t+1}}  + \ketbra{0}_{C_t}\ketbra{1}_{C_{t+1}}  \\
        &\qquad- U_t \ketbra{0}{1}_{C_{t}}\ketbra{1}{0}_{C_{t+1}} \\
        &\qquad- U_t^\dagger \ketbra{1}{0}_{C_{t}}\ketbra{0}{1}_{C_{t+1}}.
\end{align*}
This enforces the ground space of $H = H_{\tn{in}} + H_{\tn{prop}} + H_{\tn{clock}}$ to be spanned by history states with fixed-weight clock registers.
Moreover, defining $\sigma = x \| 0^m$, we choose a set 
\begin{equation*}
    R^{N}_{\sigma,w} \coloneqq \{\sigma\} \times \{w(\boldsymbol{t}) : t\in[N]\},
\end{equation*}
where $N$ is the number of initial identity gates in the circuit.
A valid guiding state is then
\begin{equation*}
    \ket{\xi_{\rm fw}} = \ket{\hat{R}^{N}_{\sigma,w}} \in \mathcal{F}(\hat{X}_{n,\hw(x)+1}),
\end{equation*}
where $\hw(x)$ is the Hamming weight of $x$.
The classical description of this state is efficient since $N = \poly{n}$.

It is clear from the setup that the weight of the guiding state is not constant.
This is the reason we refer to this problem as the \sc{Fixed Weight Guided Local Hamiltonian} problem and not the \sc{Weight-$k$ Guided Local Hamiltonian} problem since weight-$k$ states were defined with the idea that $k = O(1)$~\cite{bremner2025parameterized}.
As it turns out, the uniform fixed weight states are optimal in the sense of minimal structure required to achieve the desired complexity classifications.
There is not much more structure that can be removed from the state.
We have the uniform amplitude condition, the fixed Hamming weight, polynomially-sized support, and efficient classical description with efficient sampling-access.
This second condition was referred to above as hardening the structure and spread of the computational basis states.
It represents that a small connected sector of the Hilbert space is being used.
Since the Hamming weights are the same, on the Boolean hypercube the states are not nearest-neighbours.

\subsection{Matrix Product States}\label{sec:mps_states}
The second state type variation considered are matrix product states (MPSs).
The general structure for a matrix product state is of the form
\begin{equation*}
    \ket{\Psi} = \sum_{\underline{\sigma}} {\rm Tr}\big[(\prod_{j\in[n]} A^{\underline{\sigma}_j}) \big]\, \ket{\underline{\sigma}}.
\end{equation*}
MPS are completely specified by the set of tensors $\{A_j^{\underline{\sigma}_j}\}$ and the physical qudits $\underline{\sigma}_j$.
It requires a classical space complexity of $\Theta(n\,\chi^2\,\dim{(\underline{\sigma}_j)})$ to specify the state (where $\chi$ is the bond dimension).
This is, of course, efficient provided both the bond dimension and physical dimension are bounded by a polynomial in $n$.

Matrix product states are natural choices for guiding state types in quantum chemistry applications.
These states are ubiquitous in the study of quantum many-body systems and have been used extensively in the context of variational quantum algorithms~\cite{white1992density,perez2007matrix,guo2023differentiable,xu2024mps}.
It has also been demonstrated that MPSs are good candidate states for reaction chemistry simulations due to their straightforward preparability and applicability for certain types of multi-configuration molecules~\cite{morchen2024classification} and Heisenberg spin chains~\cite{jaderberg2025variational}.

Fairly simple structured states admit an MPS description.
Furthermore, it is straightforward to define the set of tensors required to construct an MPS that is a tensor product of basic states.
It is well-known that there is a constant-size bond dimension description of Dicke states (for constant $k$), specifically,
\begin{equation*}
    \ket{\tilde{D}^n_k} = \sum_{\substack{x \in \B^n \\ \hw(x) = k}} \ket{x} = \sum_{\underline{\sigma}} {\rm Tr}\big[(\prod_{j\in[n]} A^{\underline{\sigma}_j}) \,B\big]\, \ket{\underline{\sigma}},
\end{equation*}
where 
\begin{align*}
    A^{\underline{\sigma}_j} &= \begin{cases}
        I_{k+1} & \tn{if}~ \underline{\sigma}_j = 0\\
        J_{k+1}(0) & \tn{if}~ \underline{\sigma}_j = 1
    \end{cases}, \\
    B_{i,j} &= \begin{cases}
        1 & \tn{if}~ (i,j) = (1,k+1) ~\tn{or}~ (i,j) = (k+1,1)\\
        0 & \tn{otherwise}
    \end{cases}.
\end{align*}
The $J_{k+1}(0)$ is the $(k+1)$-dimensional Jordan block with eigenvalue $0$.
A translation-invariant MPS is one where the tensors $A^{\underline{\sigma}_j}$ are the same for all $j$.
The above MPS is an almost-translation-invariant MPS.

Recall that we defined a weight-preserving mapping $w$ used for the fixed weight states.
For this mapping, we want a uniform superposition of a subset of Dicke states for a fixed weight.
This implies we must find a way to truncate the MPS to account for the weight-preserving mapping --- let us index each of the tensors.
Let $\lambda$ represent the cut-off index, then the MPS 
\begin{equation}\label{eq:good-MPS}
    \sum_{\underline{\sigma}} {\rm Tr}\big[(\prod_{j=1}^{\lambda} A^{\underline{\sigma}_j}) \, (\prod_{j=\lambda+1}^{n} \tilde{A}_j^{\underline{\sigma}_j}) \, B\big]\, \ket{\underline{\sigma}},
\end{equation}
where $\tilde{A}_j^{\underline{\sigma}_j}$ are zero if $\underline{\sigma}_j = 1$ for $j > \lambda$, is sufficient to produce the desired superposition.
For the case of weight-$1$ states, the bond dimension is $2$, physical dimension is $2$.
We refer to the family of MPSs of the form \cref{eq:good-MPS} as $\mathcal{F}({\rm MPS}^*)$.
Under an appropriate normalisation, we can construct a valid guiding state $\ket{\xi_{\rm MPS}}$ from $\mathcal{F}({\rm MPS}^*)$.
The classical description of this state is efficient since $\lambda = \poly{n}$.

\subsection{Gaussian States}\label{sec:gaussian_states}
Our third class consists of Gaussian states. 
Gaussian states are physically motivated by their relationship to free-fermion Hamiltonians and the Jordan-Wigner transformation~\cite{jordan1928uber}. 
We provide a formal definition in \cref{app:gaussian} and only a brief overview here.
The ground states of Hamiltonians that admit a free fermion solution via a Jordan-Wigner-like map are precisely \textit{Gaussian states}. 
Such Hamiltonians can be diagonalised by circuits composed of match gates, a special set of unitaries that preserve fermionic bilinearity. 
Since match gates are efficiently simulable~\cite{valiant2008holographic}, these Hamiltonians admit classical solution methods, and Gaussian states can be prepared from basis states by the same circuits.

The physical motivation for considering Gaussian states as a guiding state comes from Ref.~\cite{bravyi2017complexity}, where it was shown that the overlap between the low-energy subspace of the free model, constructed by removing interaction terms, and the ground state of the interacting model is polynomially-small with the number of modes removed.
As we discuss in the sequel, these states are also natural candidates for extensions of the \sc{Guided Local Hamiltonian} problem to fermionic systems, such as electronic structure Hamiltonians~\cite{ogorman2021electronic}, as they can be represented using fermionic operators.
We additionally prove that finding an extremal Gaussian state for a given local Hamiltonian is hard for the class \cl{NP} (see \cref{lma:gauss_calc} in \cref{app:STLHP}), suggesting that obtaining a Gaussian guiding state is a non-trivial task.

We employ the same time encoding as used for the fixed weight states and via an identification that states of the form 
\begin{equation*}
    \ket{\psi}=\frac{1}{\sqrt{N+1}}\sum_{t=0}^{N}\ket{x}_I\ket{0^m}_A\ket{w(\boldsymbol{t})}_C
\end{equation*}
is a Gaussian state in the sense that it is the ground state of the following Hamiltonian:
\begin{multline*}
    H=-\sum_{j\in I}(-1)^{x_j}2Z_{I_j}-\sum_{j\in A}2Z_{A_j}-\sum_{j=N+1}^{L}2Z_{C_j}\\-\sum_{j=0}^{N}(X_{C_j}X_{C_{j+1}}+Y_{C_j}Y_{C_{j+1}}-2Z_{C_j})
\end{multline*}
where the addition in the final sum is modulo $N+1$. 
The frustration graph for this model is a line graph and is therefore simplicial and claw-free. 
The Jordan-Wigner mapping of this Hamiltonian is:
\begin{multline*}
    H=\sum_{j\in I}(-1)^{x_j}\, 2\, a_{I_j}^\dagger a_{I_j}^{\pd}+\sum_{j\in A} 2 \,a_{A_j}^\dagger a_{A_j}^{\pd} + \sum_{j=N+1}^{L} 2\,a_{C_j}^\dagger a_{C_j}^{\pd}\\
        -\sum_{j=0}^{N}\left(a^\dagger_{C_j} a_{C_{j+1}}^{\pd} +a^\dagger_{C_{j+1}}a_{C_j}^{\pd} +2\,a_{C_j}^\dagger a_{C_j}^{\pd}\right),
\end{multline*}
which is clearly bilinear in fermions.

\subsubsection{Fendley States} 
More recently, \citet{fendley2019free} discovered a model which provably admits no solution via a Jordan-Wigner-like transformation, and yet still has a free spectrum. 
This model has since been generalised to a whole class of models~\cite{elman2021free,chapman2023unified}, the ground states of which are precisely \textit{Fendley states}, see \cref{app:fendley} for a more rigorous definition.

We note that the set of Fendley states is a superset of the Gaussian states since these are a more general family of free-fermion states. 
Thus, the physical motivation for using Fendley states as guiding states can be inherited directly from above with the added boon that fewer Hamiltonian terms need be removed to turn a general Hamiltonian into one which may be solved using Fendley's method.

The following result therefore holds in the absence of the condition that the classical description can be used for sample-query access.
This implies that Fendley states are not useful for dequantisation algorithms and instead are more useful for quantum algorithms.

\subsection{Main Result}\label{sec:main_result_a}
Combining the constructions above with \cref{thm:R_state_hardness}, we obtain the following main result.

\begin{theorem}\label{thm:bqp_hardness}
     For any $\delta \in (0,1-1/\poly{n})$, there exists $a,b \in [0,1]$ with $b-a \geq 1/\poly{n}$ such that the \sc{Guided Local Hamiltonian} problem is \clw{BQP}{hard} for $2$-local Hamiltonians using either: 
     \begin{inparaenum}
         \item[{\rm (a)}] fixed-weight states,
         \item[{\rm (b)}] matrix product states,
         \item[{\rm (c)}] Gaussian states, and
         \item[{\rm (d)}] Fendley states.
     \end{inparaenum}
\end{theorem}

We summarise the implications of \cref{thm:bqp_hardness} and discuss potential extensions in the following subsection.

\subsection{Implications and Extensions}\label{sec:implications_extensions}
A central motivation for considering alternative guiding state families is to identify the minimal structural requirements necessary to retain the computational complexity of the \sc{Guided Local Hamiltonian} problem.
In particular, while semi-classical subset states provide a general and flexible framework, they may exhibit more structure than is strictly required to achieve the necessary overlap with the ground state.
Our results demonstrate that fixed-weight states constitute among the simplest guiding states sufficient to establish \clw{BQP}{completeness}.

Moreover, we show that guiding states with uniform amplitudes are optimal in the sense that they maximise overlap with the Feynman-Kitaev history state.
This follows from a straightforward optimisation: given a subset state $\ket{S}=\sum_{x\in S}\alpha_x\ket{x}$, the quantity
\begin{equation*}
    \abs{\sum_{x \in S \cap \mathrm{supp}(\ket{\eta})} \alpha_x}^2
\end{equation*}
is maximised under the normalisation constraint $\sum_{x\in S}|\alpha_x|^2=1$ precisely when the amplitudes are uniform.
A formal proof is provided in \cref{app:uniform}; the intuition for the result stems from the fact that, for large system sizes, the dominant portion of the history state itself is a semi-classical subset state.
Allowing larger support beyond the pre-idling region was previously exploited in proving classical hardness for stoquastic Hamiltonians~\cite{waite2025guided}.

An instructive aspect of our constructions is the relationship between semi-classical subset states and fixed-weight states.
For a subset of instances, one can define a bijective (unitary) correspondence between these guiding states, accompanied by an appropriate transformation of the Hamiltonian.
This correspondence is made explicit in our \clw{BQP}{hardness} proof and yields a polynomial-time reduction between the two formulations.
However, because the required Hamming weight depends on the instance, this does not lead to containment in parameterised classes such as \cl{XP}~\cite{bremner2022quantum}.
Further details concerning \emph{weight-$k$} guiding states are provided in \cref{app:weight_k}.

These observations highlight that the complexity of the problem depends not on abstract state equivalence, but on the encoding and input model.
They further motivate exploring structured regions of the guiding state parameter space, situated between highly unstructured subset states and overly expressive constructions.

A natural extension in this direction is provided by \emph{multi-alphabet subset states} (see \cref{def:multi_alphabet_subset_states}), which generalise subset states by allowing local alphabets beyond the computational basis.
Such states admit succinct descriptions and can be efficiently prepared using constant-depth circuits following computational basis preparation.
They encompass a broad class of guiding states and bear strong resemblance to advanced semi-classical encoded subset states (SCESS).
However, allowing arbitrarily many alphabets risks encoding problem-specific structure directly into the input, thereby undermining physical relevance and interpretability.

The necessity of extending fixed-weight constructions becomes apparent when incorporating perturbative gadget reductions.
Additional ancilla qubits introduced by such gadgets generally violate strict fixed-weight constraints.
To accommodate this, we introduce \emph{weight-$(k\!\to\!q)$ encoded states} (see \cref{def:weight_k_to_q_encoded_states}), which preserve fixed-weight structure under local isometric embeddings.
These encoded states retain the physical interpretation of fixed particle number while enabling more sophisticated reductions.

Relaxing this constraint further, we define \emph{windowed weight states} (see \cref{def:windowed_weight_states}), which allow superpositions over a polynomially bounded range of Hamming weights.
Physically, this corresponds to exploring a controlled range of occupancies, akin to active space methods in quantum chemistry.
Although this relaxation reduces the achievable fidelity, it remains sufficient to recover \clw{BQP}{completeness}, even for geometrically constrained Hamiltonians such as those defined on square lattices.

Finally, while multi-alphabet encodings can be used to absorb additional ancilla structure and preserve hardness, such constructions approach the boundary of meaningful input models.
In particular, hardness arising solely from basis misalignment, rather than intrinsic interaction complexity, may obscure the distinction between quantum and classical regimes.
This observation aligns with conjectures that the \sc{Guided Local Hamiltonian} problem becomes classically tractable for diagonal Hamiltonians when the guiding state is specified in the same basis~\cite{waite2025guided}.

\subsubsection{Relations to Other Hamiltonians}
The no low-energy trivial states (NLTS) theorem~\cite{freedman2014quantum, anshu2023nlts} identifies families of local Hamiltonians for which \emph{every} low-energy state, i.e., those within a constant additive error of the ground energy, must be non-trivial, in the sense that they cannot be prepared by constant-depth circuits acting on product states.
In contrast, our hardness constructions for the \sc{Guided Local Hamiltonian} problem demonstrate a subtler form of complexity: the existence of local Hamiltonians for which \emph{trivial} states (such as product states) achieve only a polynomially-small overlap with the ground state, whereas guiding states with non-trivial structure achieve fidelity inverse-polynomially close to one.

Thus, the Hamiltonians arising in our \clw{BQP}{hardness} reductions are \emph{not} NLTS in the strict sense --- trivial states do provide some non-negligible guidance --- but they nonetheless exhibit a regime where achieving algorithmically meaningful overlap (for polynomially-small precision) requires moving beyond trivial ansatze.
This suggests that rigorous algorithms for constructing guiding states cannot rely solely on trivial guiding states if they aim to succeed with high-quality approximations.

As a final remark, let $\boldsymbol{\rm H}_{\mu,N}^{x}$ denote the set of local Hamiltonians built using the Feynman-Kitaev construction (as outlined above) on input $x$, with mapping $\mu$ and $N$ pre-idling steps.
The state $\ket{R_{\sigma, \mu}^\nu}$ serves as a guiding state for every $H \in \boldsymbol{\rm H}_{\mu,N}^{x}$, without requiring any modification or unitary transformation.
That is, we always use the pair $(\ket{R_{\sigma, \mu}^\nu}, H)$, not $(U\ket{R_{\sigma, \mu}^\nu}, H')$ for some unitary $U$.
This is strictly stronger than the observation that unitarily equivalent Hamiltonians share guiding states up to rotation.
Instead, we identify a single fixed state that guides a large set of Hamiltonians directly, highlighting a deeper structural uniformity across the family.
\cref{cor:pinned_bqp_comp} is an extension of the idea that a fixed guiding state, when composed with a fixed set of qubits, can be used to guide a wider family of Hamiltonians.

\section{Containment of the Guided Local Hamiltonian Problem in BQP}\label{sec:bqp_containment}
The construction of trial states to guide the search for low-energy solutions is a valuable task in ground-state energy estimation problems.
Though widely used in practice, we lack a comprehensive theoretical understanding of which classes of states are most effective for this purpose when presented with an arbitrary local Hamiltonian.
Moreover, deciding which basis to use when developing potential guiding states is itself a computationally hard problem~\cite{barahona1982computational,ogorman2021electronic}.

To alleviate the challenge of finding a suitable guiding state from scratch, the \sc{Guided Local Hamiltonian} problem assumes that such a state is provided as part of the input.
The question then becomes: given a classical description of a guiding state with non-negligible overlap with the ground state, can we efficiently estimate the ground-state energy using a quantum computer?
The problem's utility, when precision and overlap are at least inverse-polynomial, hinges on a quantum algorithm that decides the promise problem efficiently, i.e., a \cl{BQP} algorithm.
Such an algorithm must prepare a trial state $\ket{\psi}$ from the classical description and estimate the ground-state energy within an additive error $|\lambda_0 - \hat{\lambda}| \leq \varepsilon$.
Unlike settings involving interactive proofs (Merlin-Arthur classes), the trial state must be constructed directly from the classical input.

As outlined in \cref{def:guided_local_hamiltonian}, we consider guiding states that are specified via a classical description rather than access to a black-box that prepares the state~\cite{cade2022complexity}.
While \clw{BQP}{hardness} has been shown for semi-classical subset states~\cite{gharibian2023dequantizing} and semi-classical encoded subset states~\cite{gharibian2022improved,cade2023improved,waite2025guided}, \cl{BQP} containment has so far been stated without a formal proof.
An important contribution of this work is to provide such an algorithm constructively.

In this section, we prove that both semi-classical subset states and, with few exceptions, the physically motivated structured variants considered in this work can be prepared efficiently from an appropriate classical description.
Combining this preparation procedure with repeated applications of quantum phase estimation yields a constructive \cl{BQP} algorithm for the \sc{Guided Local Hamiltonian} problem.
The preparation procedures for the various guiding state families are detailed in \cref{app:bqp_containment}.

We address \cl{BQP} containment in two parts.
First, we show that if a guiding state with inverse-polynomial overlap to the ground state is efficiently preparable and the Hamiltonian is efficiently accessible, then quantum phase estimation provides a \cl{BQP} algorithm for estimating the ground-state energy.
Second, we prove that the guiding state families considered here admit such efficient preparation procedures from their classical descriptions.

A subtle but essential point concerns the uniformity of the preparation procedure.
As discussed below, \cl{BQP} containment does not follow automatically from the existence of a succinct classical description of the guiding state, and requires additional uniformity assumptions.

\begin{remark}[Uniformity and Efficient Preparation]\label{rmk:uniformity}
    An instance of the \sc{Guided Local Hamiltonian} problem is of the form
    $\langle H, a, b, C_\xi \rangle$, where $C_\xi$ is a classical description of a guiding state $\ket{\xi}$.
    This description is part of the problem input and is not required to be generated uniformly across system sizes.
    However, to establish containment of the problem in \cl{BQP}, one must exhibit a \emph{uniform} polynomial-time quantum algorithm that correctly decides the problem for all valid instances.
    In particular, given $C_\xi$, the algorithm must be able to efficiently construct a quantum circuit $U_\xi$ that prepares $\ket{\xi}$ from the all-zero state.
    Thus, \cl{BQP} containment requires the existence of a uniform polynomial-time procedure that, on input $C_\xi$, outputs such a preparation circuit.
    In the absence of such a uniform compilation procedure, \cl{BQP} containment cannot be guaranteed.
\end{remark}

To illustrate the issue, consider the following pathological scenario.
For each system size $n$ and each instance $\langle H_n, a_n, b_n, C_{\xi_n} \rangle$, suppose there exists a polynomial-size quantum circuit $U_{\xi_n}$ that prepares the guiding state $\ket{\xi_n}$.
However, assume there is no uniform polynomial-time quantum algorithm that, given $C_{\xi_n}$, constructs $U_{\xi_n}$.
In this case, while each instance admits an efficient preparation circuit, these circuits cannot be generated by a single uniform algorithm across all instances.
If the circuits $U_{\xi_n}$ were instead supplied externally as size-dependent information, the resulting computational model would correspond to non-uniform computation, placing the problem closer to the class \cl{BQP}{/poly}, which permits polynomial-size advice strings depending only on the system size.
Therefore, to ensure \cl{BQP} containment, we must demonstrate a uniform polynomial-time compilation procedure for each guiding state family considered.

\subsection{Quantum Phase Estimation}\label{sec:qpe}
Quantum phase estimation (QPE) is a standard tool for estimating the eigenvalues of a unitary operator; the algorithm can be easily adapted for eigenvalue estimation of a Hermitian operator.
The algorithm outputs an estimate of the eigenvalue, corresponding to an input eigenstate, with high probability and high precision.
The precision achievable with a polynomial number of ancilla qubits is $2^{-\poly{n}}$, which suffices to resolve the ground-state energy to within the promised gap $b-a$.
Moreover, if the input state has overlap at least $\delta$ with the target eigenstate, it suffices to repeat the QPE procedure $\tilde{O}(\delta^{-1})$ times to have a high probability of success.
We formalise this in \cref{lma:qpe} below.

However, this assumes we can efficiently prepare such a state, which may not hold under restricted access models.
In the \sc{Guided Local Hamiltonian} problem, we are only given a classical description of the guiding state, not access to the state itself via a quantum oracle, say.
Thus, assuming for the moment that we can efficiently prepare the guiding state from its classical description, we now outline how QPE can be used to estimate the ground-state energy and prove \cl{BQP} containment.

We define a unitary operator, generated by the Hermitian operator $H$, via 
\begin{equation*}
    U = {\rm e}^{-\i H} = \sum_{j=0}^{2^n-1} {\rm e}^{-\i\lambda_j} \ketbra{\phi_j},
\end{equation*}
where $\{(\lambda_j,\ket{\phi_j})\}_{j=0}^{2^n-1}$ is the eigensystem of $H$.
Given an input eigenstate $\ket{\phi_j}$, the QPE algorithm outputs a bit string $\hat{\lambda}_j$ that encodes an approximation to the eigenvalue $\lambda_j$ of $H$.
If the eigenstate $\ket{\phi_j}$ is unknown or cannot be prepared, it is possible to use an approximate state $\ket{\xi}$ that has a guaranteed lower-bound on the overlap with the target eigenstate. 
Via repetition of the QPE algorithm, with an appropriate number of ancilla qubits, the energy of the approximate state can be estimated to within a desired precision.

\begin{lemma}[QPE~\cite{lin2022lecture}]\label{lma:qpe}
    Consider a $k$-local Hamiltonian $H$ over $n$ qubits with ground-state energy $\lambda_0$ and ground state $\ket{\phi_0}$.
    Let $\ket{\xi}$ be a state such that $F_{\xi,\phi_0} \geq \delta$, for some $\delta \geq 1/\poly{n}$.
    There is a quantum algorithm (quantum phase estimation) that obtains an $\varepsilon$-additive approximation to the ground-state energy $\lambda_0$, with probability at least $1-\eta$, requiring $O(\delta^{-1} \, \log(1/\eta))$ repetitions.
    The total cost of the algorithm is $O((\varepsilon \eta \delta^2)^{-1}\,(\log(1/\eta))^2)$.
\end{lemma}

A proof of this theorem can be found in Ref.~\cite{lin2022lecture}.
Provided the Hamiltonian $H$ is row-sparse and row-computable, the total cost for the QPE routine is polynomial when the parameters $\delta,\eta,\epsilon$ scale as inverse-polynomials~\cite{aharonov2003adiabatic}.
A large initial overlap requires fewer iterations of QPE and hence a decreased cost, while a polynomially-small overlap requires more iterations but with a cost at most polynomial.
This implies a bound from below of $1/\poly{n}$ on the overlap is needed to ensure the QPE algorithm can be applied.
By performing efficient state preparation in the event of a guarantee on a lower-bound on the overlap, the approximate state can be boosted closer to the ground state using techniques from Ref.~\cite{lin2020near}.

\subsection{Efficient State Preparation}\label{sec:efficient_state_preparation}
We show there exists a uniform polynomial-time quantum algorithm that, given a classical description of either a semi-classical subset state (see \cref{sec:state_type_variations}) or a semi-classical encoded subset state (see \cref{sec:state_type_variations}), constructs a quantum circuit that prepares the corresponding state.
We further show that even if the state preparation is imperfect, the overlap with the ground space can remain sufficient to apply QPE and resolve the problem.
In particular, if the prepared state $\ket{\psi}$ satisfies $\|\ket{\psi}-\ket{\xi}\|\le\varepsilon$ for some target state $\ket{\xi}$ and $\varepsilon\ge1/\poly{n}$, then its fidelity with the ground state is within $(1/\poly{n}, 1-1/\poly{n})$. 
This inverse-polynomial lower-bound on the overlap ensures that \cref{lma:qpe} can still be applied to estimate the ground-state energy.

We begin by stating the following lemma regarding semi-classical subset states.

\begin{restatable}{lemma}{EfficientStatePreparation}\label{lma:efficient_state_preparation}
    The state $\ket{\hat{C}}$ can be efficiently prepared from a classical description of the subset $C$.
\end{restatable}

A full proof of this lemma appears in \cref{app:permutation-grover-rudolph}; we sketch the main ideas here.
Recall that semi-classical subset states are defined as uniform superpositions over a subset of computational basis states and thus have a straightforward structure and can be fully described using a polynomial number of bits. 
Specifically, each state is defined by a polynomial-size subset of binary strings $C \subseteq \B^n$.
When $|C| = \poly{n}$, we refer to $C$ as \emph{sparse}.

To prepare the state in polynomial time, we use the {\tt PermutationGrover-Rudolph} algorithm from Ref.~\cite{ramacciotti2024simple}.
When parameterised by the size of the subset $C$, the algorithm is near-optimal, running in linear time.
At a high level, this approach combines Grover-Rudolph state preparation with a permutation subroutine over computational basis states.
The Grover-Rudolph algorithm prepares a superposition state over a subset $B$ of the same cardinality as $C$, i.e.,
\begin{equation*}
    \ket{\tilde{B}} = \sum_{x\in B} \ket{x} = \sum_{k=0}^{\abs{B}-1} \ket{{\rm bin}(k)},
\end{equation*}
where ${\rm bin}(k) = b_0 b_1 \ldots b_{n-1}$ is the standard binary representation of the natural number $k = \sum_{j=0}^{n-1} b_j 2^j$.
An intuitive way to understand the Grover-Rudolph algorithm is via conditional rotations implied by prefix counting.
Moreover, a series of rotation gates are sequentially applied over $n$ rounds, where each given round is conditional on the previous.
This technique is sufficient to prepare large families of arbitrary quantum states, provided with information about the amplitudes, in exponential time.

The permutation subroutine is then used to map elements of $B$ to $C$.
There is no unique way to approach this, and hence, the primary technical hurdle is proving that this permutation subroutine can be implemented and found efficiently. 
The method used in Ref.~\cite{ramacciotti2024simple} relies on the decomposition of the permutation into (disjoint) cycles, e.g., $\sigma = c_0\,c_1\ldots c_m$, where $c_i$ is a cycle, of length $l_i$, that permutes $l_i$ elements over the set.
Each cycle is implemented via a unitary operator $U_c$ that performs a Gray code rotation over the bit strings in the cycle.
Moreover, $U_c$ is decomposed into a sequence of $l$ unitary rotations $g_j$, where each $g_j$ contains one multi-controlled $X$ gate (controlled by the bit string $x_j$) and a subsequent controlled Gray code rotation $V_j$.
The purpose of $g_j$ is to perform the permutation of $x_j$ to $x_{j+1}$ within the cycle $c$.
The full permutation is then implemented by applying the cycle operators in series (see \cref{app:bqp_containment} for details).
Assuming each gate has a cost of $O(1)$, both the classical and quantum complexities of the algorithm scale as $O(|C|n)$.

In addition to the standard semi-classical subset states, we also consider the preparation of \textit{semi-classical encoded subset states} (SCESS).

\begin{restatable}{lemma}{EfficientEncodedStatePreparation}\label{lma:efficient_encoded_state_preparation}
    The state $\ket{C_{\mathcal{V}}}$ can be efficiently prepared from a classical description of the subset $C$ and the set of isometries $\mathcal{V}$.
\end{restatable}

The detailed proof of this lemma is provided in \cref{app:isometries}.
The description of these states comes in two parts: the subset $C$ and the set of isometries $\mathcal{V}$.
Since the set of isometries is restricted to be constant in size, the relevant decomposition can be implemented efficiently.
The total number of gates required to implement each isometry is $O(1)$, and thus the sequence of isometries can be efficiently implemented, requiring $O(n)$ gates. 

\subsection{Main Result}\label{sec:main_result_b}
Combining the state preparation procedures above with \cref{lma:qpe}, we obtain a constructive \cl{BQP} algorithm for the \sc{Guided Local Hamiltonian} problem under suitable guiding state families; we clarify the details below.

Given an instance $\langle H,a,b,C_\xi\rangle$, the algorithm proceeds as follows.
First, the guiding state $\ket{\xi}$ is prepared efficiently from its classical description (\cref{lma:efficient_state_preparation,lma:efficient_encoded_state_preparation}).
Quantum phase estimation (\cref{lma:qpe}) is then applied to $H$ using $\ket{\xi}$ as input, and the resulting energy estimate is compared to the thresholds $a$ and $b$.
This procedure is repeated a polynomial number of times, and the final decision is made by classical post-processing of the measurement outcomes.

For \sc{yes} instances, the ground-state energy satisfies $\lambda_0 \le a$ and the guiding state has overlap at least $\delta \ge 1/\poly{n}$ with the ground state.
Consequently, each execution of QPE yields an $\varepsilon$-accurate estimate of $\lambda_0$ with probability at least $\delta$, and repeating the procedure $O(\delta^{-1}\log(1/\eta))$ times ensures acceptance with probability at least $1-\eta$.

For \sc{no} instances, all eigenvalues satisfy $\lambda \ge b$, so any $\varepsilon$-accurate energy estimate exceeds the acceptance threshold except with probability at most $\eta$.
Thus the algorithm rejects with probability at least $1-\eta$.
Since each execution of QPE has polynomial cost and the total number of repetitions is polynomial, the overall procedure constitutes a uniform polynomial-time quantum algorithm with bounded error, and therefore lies in \cl{BQP}.

By combining \cref{lma:efficient_state_preparation,lma:efficient_encoded_state_preparation} with the fact that all physically motivated semi-classical guiding states defined in \cref{sec:bqp_hardness}, with the exception of Fendley states, admit efficient preparation from their respective classical descriptions (\cref{app:bqp_containment}), we formalise this conclusion as the following theorem.

\begin{theorem}\label{thm:bqp_contain}
    For any $\delta \in (1/\poly{n},1-1/\poly{n})$, there exist $a,b \in [0,1]$ with $b-a \geq 1/\poly{n}$ such that the \sc{Guided Local Hamiltonian} problem is contained in \cl{BQP} using either:
    \begin{inparaenum}
        \item[{\rm (a)}] SCSSs,
        \item[{\rm (b)}] SCESSs,
        \item[{\rm (c)}] fixed-weight states,
        \item[{\rm (d)}] MPSs, or
        \item[{\rm (e)}] Gaussian states.
    \end{inparaenum}
\end{theorem}

Our state preparation results additionally conclude the following corollary~\cite{waite2025guided}.

\begin{corollary}\label{cor:pinned_bqp_comp}
    The \sc{Guided Pinned (Stoquastic) Local Hamiltonian} problem is \clw{BQP}{complete}
\end{corollary}

Furthermore, we partially resolve~\cite[Conjecture 1]{waite2025guided}, concluding the duality between \clw{QMA}{complete} Hamiltonian families; completeness for stoquastic Hamiltonians remains open and is not resolved by our results.

\begin{remark}[On the Exclusion of Fendley States]
    The families of states for which we have proven \cl{BQP} containment does not currently include Fendley states. 
    This is because, in general, it has not been proven that Fendley states admit a succinct classical description from which they can be prepared efficiently~\cite{ruh2025furthering}, and thus they do not fit the criteria of~\cref{def:guided_local_hamiltonian}. 
    We note that in the event an efficient classical description or preparation protocol is developed, such an inclusion would also be confirmed. 
    The exclusion of these states for completeness underscores the importance of constructively proving class containment given the problem's confines.
\end{remark}

\section{Classical Tractability within the Goldilocks Zone}\label{sec:classical_tractability}
Establishing quantum advantage for ground-state energy estimation problems often involves comparing quantum algorithms against their classical counterparts.
Moreover, understanding the regimes where classical algorithms can efficiently approximate ground-state energies is crucial for delineating the boundaries of quantum computational supremacy.
Typically quantum advantage is demonstrated when the precision of the estimate is improved from constant to inverse-polynomial, as is the case for the \sc{Guided Local Hamiltonian} problem~\cite{gharibian2023dequantizing}.
Though we remark that this result is a worst-case complexity statement and does not preclude the existence of classical heuristics that perform well on average or for specific instances.
It is an open problem to study the average-case complexity of the \sc{Guided Local Hamiltonian} problem, particularly in physically relevant regimes.

A main result of this work concerns discussions between the \clw{BQP}{complete} and \cl{BPP} result of the \sc{Guided Local Hamiltonian} problem, particularly in the context of the newly defined guiding state types.
To compare quantum and classical regimes, we define a region of guiding states with dual relevance to both settings.
In classical settings, the \sc{Guided Local Hamiltonian} problem is tractable when the guiding state admits efficient sample- and query- (sample-query) access (see \cref{app:sample-query} for formal definitions), without strong parameter constraints.
In contrast, the quantum regime requires efficient preparation and a succinct classical description of the guiding state.
The intersection of these conditions --- the guiding states that are efficiently preparable, succinctly described, and sample-query accessible --- defines a provable quantum advantage comparison zone, which we refer to as the \emph{Goldilocks zone}; illustrated in \cref{fig:venn}.

\begin{figure}[!ht]
    \centering
    \begin{tikzpicture}
        \pic{venn};
    \end{tikzpicture}
    \caption{The overlap between the different state types producing the Goldilocks zone.
    The upper region represents those states that recover the quantum result of this work.
    The lower region represents those states that recover the classical result of Ref.~\cite{gharibian2023dequantizing}.
    States lying in the intersection (dashed) are those that can be used to prove both results, under the right conditions.}
    \label{fig:venn}
\end{figure}

\subsection{The Goldilocks Zone}\label{sec:goldilocks_zone}
Our goal in defining the Goldilocks zone is to establish a regime where the classical and quantum problem variants share comparable input information, allowing a fair assessment of their computational differences.
Since the inclusion $\cl{BQP} \subseteq \cl{BPP}$ is unresolved, unifying the input models is crucial to ensure claims of quantum advantage rest on differences in computational complexity rather than on asymmetric information supplied to the two algorithms.
Encouragingly, most of the guiding state families we introduce lie within this Goldilocks zone, and hence maintain classical tractability under appropriate parameter conditions.
We argue that states lying outside this regime could forfeit classical tractability, weakening the ability to contrast quantum and classical approaches.
\cref{fig:venn} presented a region of states where proper comparisons can be made between the classical and quantum settings.

The \cl{BPP} of Ref.\cite{gharibian2023dequantizing}, which solves the ground-state energy decision problem to constant accuracy, is achieved via a so-called \emph{dequantisation} algorithm that requires the guiding state to have a constant overlap with the ground state.
The input model for this regime only requires efficient classical sample-query access to the guiding state.
Whereas the quantum setting requires a succinct description of the guiding state that allows for efficient state preparation.
We refer to the overlap of these two input models as the \emph{Goldilocks zone}, \cref{fig:goldilocks}, where the guiding state types are sufficiently general to sustain coherent comparisons between the two settings.

\begin{figure}[!ht]
    \centering
    \begin{tikzpicture}
        \pic{goldilocks};
    \end{tikzpicture}
    \caption{The Goldilocks zone --- the outer limit of guiding state that recovers both the classical and quantum results for the \sc{Guided Local Hamiltonian} problem.
    The conjectured relationship between the physically-motivated states and the semi-classical subset states for which we have proven \cl{BQP} completeness.}
    \label{fig:goldilocks}
\end{figure}

The states lying within the Goldilocks zone are those that are `just right' for the \sc{Guided Local Hamiltonian} problem's \cl{BQP} and \cl{BPP} results.

\begin{definition}[Goldilocks State]
    Let $\ket{\Upsilon}$ be a normalised state such that there exists an efficient classical description of the state allowing for: 
    \begin{inparaenum}[(1)]
        \item efficient classical sample-query access of the state and,
        \item the existence of a polynomial-time state preparation procedure.
    \end{inparaenum}
\end{definition}

Let the set of all Goldilocks states be denoted by $\mathcal{F}(\Upsilon)$.
As we will show, the guiding state types: 
semi-classical subset states, 
semi-classical encoded subset states, 
unitarily transformed subset states,
fixed-weight states, 
matrix product states and 
Gaussian states all lie within the family $\mathcal{F}(\Upsilon)$.
Note that general unitarily transformed subset states do not lie within this family from the description of the polynomial size subset $S$ and unitary set $U = \{U_z\}_{z\in S}$ otherwise $\cl{BQP} \subseteq \cl{BPP}$.

This definition omits key properties like ground-state overlap and precision, as the Goldilocks state is a general construct applicable to both settings, not a parameter-dependent hierarchy.
Moreover, the classical description and state preparation for the state are sufficient requirements to prove at least \cl{BQP} containment, and the classical sample-query access is sufficient to prove \cl{BPP} containment, under the appropriate overlap, precision and Hamiltonian conditions, respectively.
Goldilocks states bear a close resemblance to the \emph{classically tractable} states of Refs.~\cite{vandennest2011simulating,schwarz2013simulating} and the \emph{classically evaluatable} states considered in Ref.~\cite{legall2024classical}.
A main difference is that our states are defined in terms of an overlapping region of the parameter space, rather than a subclass of states.
We also make no assumptions about the structure of the classical description, which can be a key difference between other classes of similar states.
Furthermore, this lack of structure can result in a state with variable amplitudes and an exponential number of computational basis states.
It can be shown that our newly defined guiding state types fall within the Goldilocks zone, as they are all efficiently preparable and have a classical description that allows for efficient sample-query access.
See \cref{app:sample-query} for a detailed discussion.
The following subsection summarises the classical tractability results of the guiding state types.

It has been shown that it is possible to prepare guiding states, useful for quantum chemistry, using classical pre-processing algorithms and quantum state preparation methods~\cite{fomichev2024initial,feniou2024sparse,morchen2024classification}.
While these ideas are practically relevant, they are unfortunately not amenable to the problem at hand.
This is due to a failure of quantum complexity theory, specifically in our understanding of the \sc{Local Hamiltonian} problem and the complexity classification of realistic local Hamiltonians.
While it has been shown that the likes of the Fermi-Hubbard model~\cite{schuch2009computational}, Bose-Hubbard model~\cite{childs2014bose} and antiferromagnetic Heisenberg model~\cite{cubitt2018universal} are \clw{QMA}{complete} in general, it is apparent that these classifications are not always in the realm of practicality and physical relevance (see Ref.~\cite{hamiltonianjungle2023} for a summary of the criteria).
Additionally, unless a given state type has potential for its ability to be found and strong fidelity guarantees, proving containment and hardness requires a more complete understanding of algebraic freedoms.
As already mentioned, loss of certain structure can result in a lack of classical tractability, rendering comparisons and claims of ``quantum advantage'' difficult.
Moreover, while the set $R_{\sigma, \mu'}^\nu$ provides a good recipe for constructing guiding states, it does not inform many reasonable structures, beyond the requirements of \clw{BQP}{hardness}.

\subsection{Dequantisation Algorithm}\label{sec:dequantisation_algorithm}
To prove our newly defined guiding state types are sufficient to recover the \cl{BPP} result, we outline the dequantisation algorithm used to solve the \sc{Guided Local Hamiltonian} problem classically.
We begin by informally describing what query- and sample- access to a state means.
Specifically, for a given quantum state $\ket\psi$, sample-access typically refers to the ability to obtain a random sample $\ket{x}$ from the state $\ket{\psi}$ with probability $|\langle{x}|{\psi}\rangle|^2$.
Query-access refers to the ability to obtain the amplitude $\langle{x}|{\psi}\rangle$ for any $x \in \B^n$.
In this work, we do not discuss computational constraints on the storage or encoding of these elements and values.
We assume a sufficient precision and at most polynomial access cost.
Sample-query access refers to having both sample- and query-access to the state as defined above.
\cref{app:sample-query} provides formal definitions of these access models and proves that the guiding state types considered in this work admit efficient sample-query access.

The Quantum Singular Value Transformation (QSVT)~\cite{gilyen2019quantum} provides a framework for performing polynomial transformations to the singular values of a matrix that is embedded in a higher-dimensional Hilbert space unitary operator.
Under mild access conditions and algebraic restrictions, such as element query access and polynomial sparsity, to the matrix in question, the QSVT can be dequantised~\cite{jethwani2020quantum, chia2020sampling}, i.e., there exists a classical algorithm that can approximately perform the same task.
A drawback to the classical dequantisation algorithm is the inability to maintain the same precision or overlap guarantees as the quantum algorithm.
Furthermore, for the present context, the guiding state must allow for efficient classical sample-query access.

To translate \cref{def:guided_local_hamiltonian} into the classical setting, recall that for a function $f$ and a Hamiltonian $H$, spectral decomposition implies
\begin{equation*}
    f(H) = \sum_{j=0}^{2^n-1} f(\lambda_j) \ketbra{\phi_j}.
\end{equation*}
If we assume that $\norm{H} \leq 1$, then the spectrum of $\frac{1}{2}(H+I)$ is bounded between $0$ and $1$.
We can therefore assume without loss of generality that the spectrum of $H$ is contained in the interval $[0,1]$.
The chosen action of the function $f$, on the Hamiltonian $H$, is to filter out high-energy eigenvalues; for the low- and high-energy sectors, we pick $f(x) = 1$ for $x \in [0,a]$ and $f(x) = 0$ for $x \in [b,1]$, where $a < b$.
This yields,
\begin{equation*}
    f(H) \succeq \sum_{j : \lambda_j \in [0,a]} \ketbra{\phi_j}.
\end{equation*}
For an interval $I \subseteq [0,1]$, let $N_I$ denote the number of eigenvalues of $H$ in $I$.
It follows that in general, $\sigma(H) \subset [0,a] \cup (a,b) \cup [b,1]$.
For the \sc{yes} case: $N_{[0,a]} \geq 1$ with $N_{(a,b)} \geq 0$, $N_{[b,1]} \geq 0$, and for the \sc{no} case: $N_{[0,a]}, N_{(a,b)} = 0$ and $N_{[b,1]} = 2^n$.
Take $\ket{\xi}$ to be a valid guiding state with (constant) overlap $\delta$.
When $\lambda_0 \geq b$, the quantity $\norm{f(H)\ket{\xi}}$ is $0$, and when $\lambda_0 \leq a$, we find 
\begin{equation*}
    \norm{f(H)\ket{\xi}} \geq \norm{\sum_{j : \lambda_j \in [0,a]} \braket{\phi_j}{\xi} \ket{\phi_j}} \geq \delta.
\end{equation*}
The classical algorithm must then estimate the quantity $\norm{f(H)\ket{\xi}}$.
To achieve this, sample-query access to the guiding state $\xi$, query access to $H$ and a polynomial approximation to the function $f$ is required.

One important step in the dequantisation algorithm is to construct an efficient classical routine to output coefficients of the vector $p(H) \ket{\xi}$, where $p$ is a polynomial approximation to a filter function~\cite{low2017hamiltonian}.
Assuming sample-query access to the guiding state $\ket{\xi}$ and noting that $k$-local Hamiltonians have a sparsity of $O(n^k)$, the cost of this procedure can be shown to scale as $\tilde{O}(n^{kd})$, where $d$ is the degree of the polynomial $p$ and $\tilde{O}$ hides logarithmic factors~\cite{gharibian2023dequantizing}.
Furthermore, the required degree of the polynomial $p$ roughly scales as the inverse of the precision multiplied by the logarithm of the overlap, i.e., $O(n^{(k\log(1/\delta))/(b-a)})$.
For a polynomially-small overlap but constant precision, the dequantisation algorithm runs in quasi-polynomial time and for a constant overlap but polynomially-small precision, the dequantisation algorithm runs in exponential time.
It then follows that the dequantisation algorithm runs in polynomial time for constant overlap and precision, which is sufficient to prove \cl{BPP} containment of the \sc{Guided Local Hamiltonian} problem.

We can also conclude the classical tractability of the \sc{Guided Local Hamiltonian} problem for the state types considered in this work, under the appropriate conditions.

\begin{corollary}
    For any constants $a,b \in [0,1]$ such that $a < b$ and any constant $\delta \in (0,1]$, the \sc{Guided Local Hamiltonian} problem can be efficiently solved classically with probability at least $1-1/\exp(n)$ for the following state types: 
    \begin{inparaenum}
        \item[{\rm (a)}] SCSSs,
        \item[{\rm (b)}] SCESSs,
        \item[{\rm (c)}] fixed-weight states,
        \item[{\rm (d)}] MPSs, and
        \item[{\rm (e)}] Gaussian states,
    \end{inparaenum}
    for constant local Hamiltonians.
\end{corollary}

\section{Complexity of the Fermi-Hubbard Model with a Guiding State}\label{sec:fh_model}
The Fermi-Hubbard model is fundamental in condensed matter physics, capturing essential features of strongly correlated electron systems~\cite{hubbard1963electron}.
The spin-$\frac{1}{2}$ Fermi-Hubbard model $\tilde{H} \coloneqq T + V$, typically defined on a lattice, describes electrons hopping between lattice sites with on-site interactions:
\begin{equation*}
    T \coloneqq -t \sum_{\substack{\langle i, j \rangle \\ \sigma \in \{\uparrow, \downarrow\}}} \left( a_{i, \sigma}^\dagger a_{j, \sigma} + a_{j, \sigma}^\dagger a_{i, \sigma} \right)
\end{equation*}
is the kinetic term representing electron hopping between nearest-neighbour sites $\langle i, j \rangle$ with hopping amplitude $t > 0$, and
\begin{equation*}
    V \coloneqq U \sum_{i} n_{i, \uparrow} n_{i, \downarrow}
\end{equation*}
is the potential term representing on-site Coulomb repulsion with interaction strength $U > 0$.
See \cref{app:gaussian} for details on fermionic operators.

At \emph{half-filling}, the number of electrons equals the number of lattice sites.
In the \emph{Mott insulating phase} the interaction strength $U$ dominates the hopping amplitude $t$, leading to a regime where electron mobility is suppressed due to strong repulsive interactions.
We focus on the case where the Fermi-Hubbard Hamiltonian describes a spin-$\frac{1}{2}$ fermionic system at half-filling in the Mott insulating phase.
Let $\mu_0$ denote the ground-state energy of the Fermi-Hubbard Hamiltonian $\tilde{H}$, and let $\ket{\psi_0}$ denote its corresponding ground state.

We now introduce the problem statement for the Fermi-Hubbard model.

\begin{definition}[The \sc{Guided Fermi-Hubbard Hamiltonian} Problem]\label{def:guided-fermi-hubbard}
    Given a Fermi-Hubbard Hamiltonian $\tilde{H}$ acting on $\eta$ spin-$\frac{1}{2}$ fermions such that $\|\tilde{H}\| \leq 1$, parameters $a,b \in [0,1]$ such that $b-a \geq 1/\poly{\eta}$ and a description of a semi-classical encoded fermionic subset state $\ket{\zeta_{\rm f}}$ with the promise that $\abs{\braket{\zeta_{\rm f}}{\psi_0}}^2 \geq \delta$ for some $0 <\delta < 1$, decide whether $\mu_0 \leq a$ or $\mu_0 \geq b$, promised one is true.
\end{definition}

See \cref{sec:fermionic_guiding_states} for the definition of a semi-classical encoded fermionic subset state.
Note that the definition is analogous to the semi-classical encoded subset states used in prior work on guided spin Hamiltonians~\cite{cade2023improved}.

We begin by outlining the proof strategy for the \clw{BQP}{hardness} of the \sc{Guided Fermi-Hubbard Hamiltonian} problem in the absence of local fields.
The case of local fields is treated later as the proof follows an analogous approach to the one without local fields, but requires a preliminary discussion about the spin result we use.
To establish the results, we leverage the duality mapping between the antiferromagnetic Heisenberg model and the half-filled Fermi-Hubbard model~\cite{auerbach1994interacting, hubbard1963electron}.

In the spin setting, Ref.~\cite{cade2023improved} demonstrated that the \sc{Guided Antiferromagnetic Heisenberg Hamiltonian} problem is \clw{BQP}{complete} with a semi-classical encoded subset state as the guiding state, even when restricted to $2$D square lattices.
As we include both square and triangular lattices, we  improve this result to triangular lattices using the arguments of Ref.~\cite{cubitt2018universal}.\footnote{The proof for the extension to triangular lattices requires a (generalised subdivision) perturbation gadget that tessellates the triangular lattice.}
Note that we do not present a proof of the improved results to triangular lattices here, as our focus is on extending to fermionic Hamiltonians.
We state the following strengthened version of the main result from Ref.~\cite{cade2023improved}.

\begin{theorem}\label{thm:guided-heisenberg}
    The \sc{Guided Antiferromagnetic Heisenberg Hamiltonian} problem is \clw{BQP}{complete} 
    for any $\delta \in (1/\poly{n}, 1 - 1/\poly{n})$, 
    given a semi-classical encoded subset guiding state, 
    even when restricted to $2$D square or triangular lattices.
\end{theorem}

In the setting with local fields, we use an alternative result; for the specific details of the local field case, see \cref{sec:fermi_hubbard_local_fields}.

\subsection{Heisenberg-Fermi-Hubbard Duality}\label{sec:duality_mapping}
We now review the duality mapping between the antiferromagnetic Heisenberg model and the half-filled Fermi-Hubbard model~\cite{auerbach1994interacting, hubbard1963electron}.
This mapping uses a well-known result of \citet{liu2007quantum} showing that fermionic occupation numbers at a site can encode a logical qubit.
More specific mathematical details can be found in \cref{app:fermi_hubbard_complexity}.

We summarise the implications of the duality mapping as follows.

\begin{lemma}\label{lem:duality-mapping}
    There exists an efficiently computable reduction between an instance of the Fermi-Hubbard Hamiltonian $\tilde{H}$ to an instance of the antiferromagnetic Heisenberg Hamiltonian $H$ such that: 
    the ground-state energies $\mu_0$ and $\lambda_0$ and ground states $\ket{\psi_0}$ and $\ket{\phi_0}$ of $\tilde{H}$ and $H$ respectively, are inverse-polynomially close.
\end{lemma}

\begin{remark}
    The geometry of the local interactions is preserved under the reduction in \cref{lem:duality-mapping}.
    Therefore, any lattice structure present in the Fermi-Hubbard model is also present in the corresponding Heisenberg model.
\end{remark}

We sketch the proof of \cref{lem:duality-mapping}; see \cref{app:fermi_hubbard_complexity} for full details.
To begin, we perform a Jordan-Wigner transformation to express the Fermi-Hubbard Hamiltonian in terms of local spin operators acting on $2\eta$ qubits.
This requires an ordering of the fermionic modes --- we choose a one-dimensional ordering.
Specifically, for $n$ lattice sites, each supports two fermionic modes corresponding to the spin-$\frac{1}{2}$ degrees of freedom; let the tuple $(j, \sigma)$ denote the mode corresponding to spin $\sigma$ at site $j$.
Our ordering is given by the lexicographic ordering of the tuples, i.e., $(1, \uparrow) < (1, \downarrow) < (2, \uparrow) < (2, \downarrow) < \cdots < (n, \uparrow) < (n, \downarrow)$.
It follows that the fermionic ladder operators can be expressed as products of Pauli operators.
Implicit in this transformation is the introduction of non-local string operators to account for the fermionic anticommutation relations.

We then define a logical encoding between qubits and occupation states of the fermionic modes.
The reason for this stems from the particular regime we have chosen to work in, i.e., the Mott insulating phase at half-filling, which allows us to restrict our attention to when there exists exactly one fermion per site.
Specifically, we encode the states $\ket{10}$ and $\ket{01}$ as $\ket{\boldsymbol{0}}$ and $\ket{\boldsymbol{1}}$ respectively, where the first and second entries of the occupation state correspond to the spin-$\uparrow$ and spin-$\downarrow$ modes at a given site.\footnote{We use boldface to distinguish the logical states from the physical states.}
The global encoding, denoted as $\mathcal{J}$, is given by the tensor product of local encodings $J_j$ at each site.
It then suffices to use second-order perturbation theory to prove the appropriate simulation of the Heisenberg Hamiltonian by the Fermi-Hubbard Hamiltonian within the logical subspace defined by the encoding.
We define an effective Hamiltonian acting on the low-energy subspace of the Fermi-Hubbard potential term $V$:
\begin{equation*}
    \tilde{H}_{\tn{eff}} = - T_{-+}\ V_{++}^{-1}\ T_{+-} + O\Big(\frac{t^3}{U^2}\Big),
\end{equation*}
where $\Pi_-$ and $\Pi_+$ are the projectors onto the low-energy and high-energy subspaces of $V$ respectively and $A_{\pm\pm} = \Pi_\pm A \Pi_\pm$ for any operator $A$.
A straightforward calculation shows that this effective Hamiltonian is proportional to the antiferromagnetic Heisenberg Hamiltonian $H$ up to an additive constant.
Then, by setting the parameters $t$ and $U$ appropriately~\cite{ogorman2021electronic}, i.e., scaling at most polynomially in $\eta$, we can tune the simulation error to be at most inverse-polynomial in $\eta$.
That is, the ground-state energies $\mu_0$ and $\lambda_0$ and ground states $\ket{\psi_0}$ and $\ket{\phi_0}$ of $\tilde{H}$ and $H$ respectively, are inverse-polynomially close.
Manual renormalisation of the Hamiltonians ensures that $\|\tilde{H}\|, \|H\| \leq 1$.

A subsequent use of \cref{prop:simulator-hardness} allows us to compare the target and simulator Hamiltonians in an appropriate complexity-theoretic sense.
Moreover, we can safely alter the parameters of an instance $x = \langle H, a, b \rangle$ to generate a valid instance for $\tilde{H}$; it suffices to $a' = a + \epsilon$ and $b' = b - \epsilon$ for some $\epsilon < (b-a)/2$.

\subsection{Fermionic Guiding States}\label{sec:fermionic_guiding_states}
To complete the \clw{BQP}{hardness} reduction we must prove that there exists a suitable fermionic guiding state for the Fermi-Hubbard Hamiltonian $\tilde{H}$ and that the overlap with the ground state $\ket{\psi_0}$ is sufficient.
Taking $\ket{\xi_{\rm s}}$ to be the semi-classical encoded subset guiding state for the Heisenberg Hamiltonian instance (\cref{thm:guided-heisenberg}), we define the corresponding fermionic guiding state $\ket{\zeta_{\rm f}}$ via the duality mapping:
\begin{equation*}
    \ket{\zeta_{\rm f}} \coloneqq \mathcal{J}\!\ket{\xi_{\rm s}} = \left(\bigotimes_{k} J_k \right) \ket{\xi_{\rm s}}.
\end{equation*}
It follows that $\ket{\zeta_{\rm f}}$ is a \emph{semi-classical encoded fermionic subset state} described by the same data as $\ket{\xi_{\rm s}}$ appended with the local encodings $J_j$ at each site.
In particular, if $\ket{\xi_{\rm s}}$ is defined over the subset $S \subset \B^n$ with local isometries $\{V_j\}_{j\in[n]}$, then $\ket{\zeta_{\rm f}}$ is defined over the same subset $S$ with local isometries $\{W_j\}_{j\in[n]}$, where $W_j = (J_k)^{\otimes m_j} \circ V_j$ for each $j$.
It follows that $W_j : \mathbb{C}^2 \to (\mathbb{C}^2)^{\ell_j}$ where $\ell_j = 2 m_j$, i.e., $V_j$ takes $\ket{\boldsymbol{x}_j} \mapsto \sum_{y} \bigotimes_{i=1}^{m_j} \ket{\boldsymbol{y}_i}$ and the set of $J$'s then maps $\ket{\boldsymbol{y}_i} \mapsto \ket{y_i^\uparrow y_i^\downarrow}$, where $y_i^\uparrow$ and $y_i^\downarrow$ are the occupation numbers of the spin-$\uparrow$ and spin-$\downarrow$ modes at site $i$ respectively.
For example, $W\ket{\boldsymbol{0}} = \frac{1}{\sqrt{2}} (\ket{10,01} - \ket{01,10})$ is sufficient for our purposes (cf. \cref{app:hubbard_heisenberg}).

Following \cref{sec:2-local_reduction} and Ref.~\cite{bravyi2016complexity}, we relate the ground state $\ket{\psi_0}$ of $\tilde{H}$ to the ground state $\ket{\phi_0}$ of $H$, that is,
\begin{equation*}
    \norm{\ket{\psi_0} - \mathcal{J}\!\ket{\phi_0}} = O(1/\poly{n}).
\end{equation*}
Using \cref{lma:geometric-lemma} we can bound the overlap between the fermionic guiding state $\ket{\zeta_{\rm f}}$ and the Fermi-Hubbard ground state $\ket{\psi_0}$ as follows:
\begin{equation*}
    \abs{\braket{\zeta_{\rm f}}{\psi_0}}^2 \ge \kappa,
\end{equation*}
for some $\kappa \in (1/\poly{n}, 1 - 1/\poly{n})$.
Polynomial bounds in $n$ can be converted to polynomial bounds in $\eta$ since $\eta \propto n$.
More specific details can be found in \cref{app:fermi_hubbard_complexity}.

\subsection{Fermi-Hubbard with Local Fields}\label{sec:fermi_hubbard_local_fields}
\citet{schuch2009computational} proved a \clw{QMA}{completeness} result for the Fermi-Hubbard model on a $2$D square lattice with local magnetic fields.
The result was obtained via a reduction chain, comprised of four unique perturbative gadget, from general $2$-local Pauli Hamiltonians on a square lattice to Heisenberg Hamiltonians with local magnetic fields on a square lattice (see Ref.~\cite{hamiltonianjungle2023} for a review of perturbation gadgets).
A subsequent use of the duality map outlined in \cref{sec:duality_mapping} then yields a Fermi-Hubbard Hamiltonian with local magnetic fields.
The literature so far has not explored the \sc{Guided Local Hamiltonian} problem following this chain of reductions, in favour of results on similar models without local fields~\cite{cubitt2018universal, piddock2017complexity, cubitt2016complexity}.
For the sake of completeness, we now outline the necessary arguments to extend the \clw{BQP}{hardness} result for the Fermi-Hubbard model to the case where local magnetic fields are present.

The chain of gadgets from Ref.~\cite{schuch2009computational} apply in the same standard manner as other perturbative gadget constructions~\cite{oliveira2008complexity}, therefore, the results of previous sections suggest we need only verify the form of the guiding state throughout the reduction.
Recall that a semi-classical subset state (or trivial SCESS) is a valid guiding state for the initial $2$-local Pauli Hamiltonian instance~\cite{cade2023improved}.
The square lattice embedding of Ref.~\cite{oliveira2008complexity} can be straightforwardly adapted to a triangular lattice embedding without significant changes to the gadget constructions.
It then suffices we can tessellate the triangular lattice with the gadget chains of Ref.~\cite{schuch2009computational}; \cref{fig:triangular-lattice-gadget} shows a visualisation of such a tessellation.
Since each gadget in the chain requires the addition of ancilla qubits in states of the form $\ket{\theta} = \frac{1}{\sqrt{2}}(\ket{0} + e^{i\theta}\ket{1})$, which are single-qubit states, we can conclude that the final Heisenberg Hamiltonian with local magnetic fields admits a SCESS guiding state.
For example, the local isometries $V(\theta)$ can be defined to act as 
\begin{equation*}
    \bigotimes_{j \in [l]} \ket{\theta_j} = \bigotimes_{j \in [l]} V(\theta_j) \ket{0},
\end{equation*}
where $l$ is the number of ancilla qubits.
Thus, the resulting SCESS is defined via $I^n \otimes V(\theta_1) \otimes \cdots \otimes V(\theta_l)$ acting on the initial semi-classical subset guiding state padded with $\ket{0^l}$ for the ancilla qubits.

Finally, we apply the duality mapping of \cref{lem:duality-mapping} to obtain a Fermi-Hubbard Hamiltonian with local magnetic fields.

\begin{figure}[!ht]
    \centering
    \begin{tikzpicture}
        \pic[scale=0.6]{triangle-gadgets};
    \end{tikzpicture}
    \caption{
        The gadget chain from Ref.~\cite{schuch2009computational} adapted to the triangular lattice. 
        For brevity, we only show one gadget chain cell (the large connected square and triangle) and its interior local external field contributions (gray nodes). 
        Each coloured nodes on the cell perimeter represents one of the four mediator qubits. 
        The black nodes represent the qubits of the original Hamiltonian.
    }
    \label{fig:triangular-lattice-gadget}
\end{figure}

\subsection{Main Results}\label{sec:main_result_c}
It is clear from the above arguments that our reductions and constructions are computable in polynomial time and that both results (with and without local magnetic fields) are shown to hold for $2$D square and triangular lattices.
An intermediate outcome of our results is a slight strengthening of \clw{BQP}{hardness} results presented in prior work~\cite{cade2023improved}.
It follows from \cref{sec:bqp_containment} that the \sc{Guided Fermi-Hubbard Hamiltonian} problem is also in \cl{BQP} via standard techniques.
Thus, combing the results of \cref{sec:duality_mapping,sec:fermionic_guiding_states} with \cref{thm:guided-heisenberg} yields the following result for the Fermi-Hubbard model without local magnetic fields.

\begin{theorem}\label{thm:guided-fermi-hubbard-no-fields}
    The \sc{Guided Fermi-Hubbard Hamiltonian} problem is \clw{BQP}{complete} 
    for any $\delta \in (1/\poly{\eta}, 1 - 1/\poly{\eta})$, 
    given a semi-classical encoded fermionic subset guiding state, 
    even when restricted to $2$D square or triangular lattices.
\end{theorem}

The second result extends \clw{BQP}{completeness} to the case where external fields are present and follows from the arguments outlined in \cref{sec:fermi_hubbard_local_fields}.

\begin{theorem}\label{thm:guided-fermi-hubbard-local-fields}
    The \sc{Guided Fermi-Hubbard Hamiltonian} problem is \clw{BQP}{complete} 
    for any $\delta \in (1/\poly{\eta}, 1 - 1/\poly{\eta})$, 
    given a semi-classical encoded fermionic subset guiding state, 
    even when restricted to $2$D square or triangular lattices with local magnetic fields.
\end{theorem}

Both \cref{thm:guided-fermi-hubbard-no-fields} and \cref{thm:guided-fermi-hubbard-local-fields} close an important open question raised in Ref.~\cite{cade2023improved} regarding the computational complexity of fermionic Hamiltonians in the guided setting.

\begin{remark}[Fermionic Guiding State Families]
    In this work, we have introduced a range of new physically-motivated guiding state types for the \sc{Guided Local Hamiltonian} problem.
    Our results for the Fermi-Hubbard model do not utilise these new state types, but rather rely on the semi-classical encoded subset states used in prior work on guided spin Hamiltonians.
    It remains an open question whether \clw{BQP}{completeness} can be established for the Fermi-Hubbard model with Gaussian guiding states; constructing such a reduction presents non-trivial obstacles that we leave for future work.
\end{remark}

\section{Conclusion}\label{sec:conclusion}
Our results provide further insight into how complexity theory can elucidate the potential advantages of quantum algorithms in addressing practical problems. 
In prior work, the state classes considered for the \sc{Guided Local Hamiltonian} problem were rather abstract. 
In contrast, we have demonstrated that states relevant to quantum chemistry and condensed-matter physics serve as strong candidates for guiding states. 
We hope that future research will build on these results and address additional open questions; for example, one interesting direction would be to explore extensions to semi-classical physically motivated states, including placing Fendley states in \cl{BQP} by developing an efficient preparation algorithm from a succinct description.
The results in this work leverage recent advances in state preparation procedures and complexity-theoretic tools (such as perturbation theory and gadgets), to establish connections between the complexity of the \sc{Guided Local Hamiltonian} problem and the underlying guiding state. 
Furthermore, we present a framework that characterises the relationship between the guiding state and the resulting Feynman-Kitaev Hamiltonian, enabling us to identify several characteristics necessary for \clw{BQP}{hardness}.
We have also established a complexity result for physically-relevant models beyond the spin setting --- the Fermi-Hubbard model --- in the guided setting, closing an important gap in the literature.
In this setting, we have shown that the \sc{Guided Fermi-Hubbard Hamiltonian} problem is \clw{BQP}{complete} for the Fermi-Hubbard model at half-filling in the Mott insulating phase, even when local magnetic fields are present, with a semi-classical encoded fermionic subset state as the guiding state.

Technicalities surrounding the \textsc{Guided Local Hamiltonian} problem necessitate careful analysis when classifying its complexity.
This work was conducted under the conventional definition of the problem --- where the input is a classical description of the state.
We further found that such descriptions were sufficiently detailed to allow for sample- and query-access to the state. 
Such conditions facilitate a comparison between the classical and quantum hardness results of the problem, as in Ref.~\cite{gharibian2023dequantizing}. 
Under this access model, one can define an upper limit on the types of states that reside in the ``Goldilocks zone''; that is, those states that permit both the \clw{BQP}{completeness} and \cl{BPP} results.
Outside of this regime, the problem may lose its classical tractability; this renders comparisons between the classical and quantum results less meaningful.

Our results indicate that certain parameter regimes correspond to more structured guiding states, in contrast to broader scenarios that make minimal assumptions about the state or the Hamiltonian.
We have found three interesting classes of guiding states: fixed-weight states, matrix product states and Gaussian states, that result in both \clw{BQP}{completeness} and \cl{BPP} results.
These new guiding state types prove that there is a broader set of parameters to explore for both theoretical and practical efforts in the context of the \sc{Guided Local Hamiltonian} problem.

\subparagraph{Discussion.} 
In Ref.~\cite{gharibian2023dequantizing}, two related variants of the problem are introduced with different input models to the states: \sc{GLH\textsuperscript{*}}, requiring only efficient classical sample-query access, and \sc{GLH}, using a classical description.\footnote{This variation is equivalent to \sc{SCSS-GLH}, which uses a semi-classical subset state as the guiding state.}
The requirement for sample-query appears artificial from a physical perspective, but is useful for analysing the problem's complexity via dequantisation arguments.
For certain instances, it was shown that \sc{GLH} is classically tractable~\cite[Proposition 4.5]{gharibian2023dequantizing}; interestingly, the SCSS description inherently provides sample-query access, so the same results apply to \sc{GLH\textsuperscript{*}} (under the appropriate conditions).
A critical insight is that losing sample-query access could result in a loss of classical tractability.
The absence of such a feature would make it difficult to compare claims of ``quantum advantage'' for this problem.
In this work, we adopt the original convention established in Ref.~\cite{gharibian2023dequantizing}.

Finally, we note that under our definitions, the \clw{BQP}{completeness} result of Ref.~\cite{cade2022complexity} applies only to semi-classical \emph{encoded} subset states rather than standard semi-classical subset states. 
This is because a step in the proof requires the use of $O(n^3)$ ancilla qubits, in the $\ket{+}$ state, which the standard definition of semi-classical subset states does not support.
As we discuss in \cref{sec:state_type_variations}, the standard definition of the subset states assumes a fixed-basis encoding, i.e., using the binary alphabet. 
A natural extension is to consider a multi-alphabet encoding. However, this may not be a natural approach and may lead to complications in decoding the description of a potential guiding state prepared by a quantum algorithm.

\subparagraph{Open Problems.}
An important open problem we pose is to explore the relationship and equivalence between different access models for the guiding state.
This can potentially close the gap between problem variants that subtly differ in their access models.
An extension to the Goldilocks zone suggests an input model providing a classically efficient description of the quantum circuit that prepares the guiding state may suffice to establish both a \cl{BQP}~\cite{cade2022complexity} and \cl{BPP}~\cite{gharibian2023dequantizing} result, under appropriate (respective) conditions.

Additionally, we ask whether there are more practically relevant Hamiltonians that fit within the framework we have established.
A well-known limitation of the Schrieffer-Wolff transformation and perturbative gadgets is the polynomial blow-up in the number of qubits and the strength of the interactions.
This limitation implies that the Hamiltonians within the \clw{BQP}{complete} framework are far from physically relevant, even if the underlying interactions are. 
However, in the direction of reducing interaction strength overhead, perturbative gadgets have been constructed that only introduce a constant increase in interaction strength at the expense of a polynomial increase in the number of interactions per particle~\cite{cao2015perturbative}.
Using such gadgets in conjunction with previous results~\cite{kempe2006complexity,nagaj2007new}, we can prove the \clw{BQP}{completeness} of the \sc{Guided Local Hamiltonian} problem for local Hamiltonians with $O(1)$-strength interactions and an $O(1)$ promise gap.
One other direction is to consider less restricted classes of fermionic Hamiltonians, for example, those that are not necessarily at half-filling or in the Mott insulating phase.
In fact, this is an important open question in the broader context of Hamiltonian complexity theory --- it appears existing tools are insufficient to establish further results for fermionic Hamiltonians beyond the half-filled Mott insulating phase.

Geometrical restrictions beyond $2$-dimensions have yet to be considered for the \sc{Guided Local Hamiltonian} problem.
We make the following conjecture on the complexity classification of the \sc{Guided Local Hamiltonian} problem for one-dimensional Hamiltonians on eight-state qudit systems.

\begin{conjecture}
    The \sc{Guided $8$-state $2$-Local Hamiltonian} problem on a one-dimensional lattice is \clw{BQP}{complete}.
\end{conjecture}

We suspect that demonstrating the required state preparation for qudit semi-classical encoded states may be challenging. 
However, the hardness of the problem should be achievable using appropriate modifications of the results of this work and Ref.~\cite{hallgren2013local}. 
Furthermore, this result would need a qudit extension to a semi-classical state, but this should follow straightforwardly.
Additionally, we expect this work and simple adaptions to the arguments of Ref.~\cite{kay2007quantum} to yield a \clw{BQP}{completeness} result for the guided problem for translational-invariant qudit Hamiltonians on a one-dimensional lattice.

Additionally, problems that rise above the standard \clw{QMA}{completeness} of the \sc{Local Hamiltonian} problem may also admit interesting results when provided with a guiding state.
For example, problems that consider quantities beyond the ground-state energy, such as the \sc{Approximate Simulation} problem~\cite{ambainis2014physical} (among others in this work), concern the estimation of the expectation value of a local observable with respect to the ground state of a given Hamiltonian.
Making the appropriate modifications to endow this problem with a guiding state may shift the problem to within \cl{BQP}, either by a direct calculation using the guiding state or by proving that the \sc{Guided Approximate Simulation} problem is contained in \cl{P\textsuperscript{BQP[$\log$]}}.
Since \cl{BQP} is self-low, it follows that \cl{P\textsuperscript{BQP[$\log$]}}$=$\cl{BQP}.
However, it may prove difficult to obtain the appropriate bounds on parameters if we are only provided with a single guiding state for one local Hamiltonian instance.

A final open problem we mention is to consider the effect of using different reductions from circuits to Hamiltonians. 
For example, what are the effects of using different clock states?
In particular, if the mapping produces a superposition of computational basis states or employs an alternative basis (for example, the Bell basis), these modifications could affect both the potential guiding states and the complexity of the problem~\cite{nagaj2007new}.
It does seem, however, that semi-classical encoded subset states may be sufficient to resolve these cases.
Three other constructions that may be of interest are: the graphical approach of Childs, Gossett and Webb~\cite{childs2014bose}, the injective tensor network reduction technique, generalised by Anshu, Breuckmann and Nguyen~\cite{anshu2024circuit} and the quantum Thue system framework of Bausch, Cubitt and Ozols~\cite{bausch2017complexity}.
However, we caution that over-engineering any specific family of states may be a superfluous task. 
Doing so may rob the Hamiltonian of physical realism or prevent the reduction to known models; thus, the problem may lose its practical relevance.

\section*{acknowledgments}\label{sec:acknowledgments}
We thank Mauro Morales for helpful feedback on the manuscript and useful discussions on future directions.
GW would like to thank Thinh Le for insightful discussions concerning state preparation procedures and Ryan Mann for comments on the connection between SCSS and fixed-weight states.
GW and KL are supported by a scholarship from the Sydney Quantum Academy. This work was supported by the ARC Centre of Excellence for Quantum Computation and Communication Technology (CQC2T), project number CE170100012. SJE and MJB were also partially supported with funding from the Defense Advanced Research Projects Agency under the Quantum Benchmarking (QB) program under award no. HR00112230007, HR001121S0026, and HR001122C0074 contracts. The views, opinions and/or findings expressed are those of the authors and should not be interpreted as representing the official views or policies of the Department of Defense or the U.S. Government. 

\newpage 
\bibliographystyle{apsrev4-2}
\bibliography{ref}\label{sec:refs}

\onecolumngrid
\appendix
\addtocontents{toc}{\protect\setcounter{tocdepth}{0}}
\section*{Table of Contents}\label{app:toc}

\begin{center}
\begin{minipage}{0.9\textwidth}
\manualtocentryapp{app:background}{Background and Preliminaries}
    \manualtocsubentryapp{app:complexity}{Complexity Theory}
    \manualtocsubentryapp{app:problem_statement}{Problem Statements}
    \manualtocsubentryapp{app:gaussian}{Gaussian States}
    \manualtocsubentryapp{app:fendley}{Fendley states}
\manualtocentryapp{app:bqp_hardness_framework}{BQP Hardness Framework}
    \manualtocsubentryapp{app:alt_idling_reductions}{Alternate Idling Reductions}
\manualtocentryapp{app:extended_state_families}{Extended State Families}
\manualtocentryapp{app:bqp_containment}{Efficient State Preparation from Sparse Classical Data}
    \manualtocsubentryapp{app:approx-state-prep}{Approximate State Preparation}
    \manualtocsubentryapp{app:permutation-grover-rudolph}{Permutation Grover-Rudolph Proposal}
    \manualtocsubentryapp{app:isometries}{Inclusion of Isometries}
    \manualtocsubentryapp{app:mps_prep}{Preparation of Matrix Product States}
    \manualtocsubentryapp{app:Gauss_prep}{Preparation of Gaussian States}
    \manualtocsubentryapp{app:bqp_containment_main}{Main Result of \cref{sec:bqp_containment}}
    \manualtocsubentryapp{app:history-state-prep}{History State Preparation}
\manualtocentryapp{app:sample-query}{Sample and Query Access}
    \manualtocsubentryapp{app:sample-query-proofs}{Proofs of Classically Efficient Sample and Query Access}
        \manualtocsubsubentryapp{app:SCESS-sample-query}{Semi-Classical Encoded Subset States}
        \manualtocsubsubentryapp{app:adv-scess}{Advanced Subset States}
        \manualtocsubsubentryapp{app:mps-sample-query}{Matrix Product States}
        \manualtocsubsubentryapp{app:gauss-sample-query}{Gaussian States}
\manualtocentryapp{app:STLHP}{Local Hamiltonian Problems with Different States}
\manualtocentryapp{app:weight_k}{Weight-k Guiding States}
\manualtocentryapp{app:uniform}{Optimality of Uniform Amplitudes}
\manualtocentryapp{app:fermi_hubbard_complexity}{Hardness Reduction for the Fermi-Hubbard Model}
    \manualtocsubentryapp{app:jw}{Jordan-Wigner Transformation and Local Encoding}
    \manualtocsubentryapp{app:hubbard_heisenberg}{Perturbative Reduction to the Heisenberg Model}
    \manualtocsubentryapp{app:guiding_state}{Guiding State Construction and Overlap Preservation}
\end{minipage}
\end{center}

\section{Background and Preliminaries}\label{app:background}
A local Hamiltonian over $n$ qubits is a self-adjoint operator $H = \sum_j h_j$, where $\abs{{\rm supp}(h_j)} \leq k$ for some ${k = O(1)}$ and $j \in [\poly{n}]$.
The ground-state energy $\lambda_0$ of a local Hamiltonian is the minimum eigenvalue of the Hamiltonian.

Let $X_{n,k}$ denote a subset of binary strings of length $n$ with Hamming weight $k$.
Given a symbolic representation of a state, e.g., $\psi, X_{n,k}, \eta$, we will denote \emph{normalised} states with no marker: $\psi, X_{n,k}, \eta$, \emph{un-normalised} states with a tilde: $\tilde{\psi}, \tilde{X}_{n,k}, \tilde{\eta}$, and (normalised) \emph{uniform amplitude} states with a hat: $\hat{\psi}, \hat{X}_{n,k}, \hat{\eta}$.
Define the fidelity between two states $\ket{\psi}$ and $\ket{\phi}$ as $F_{\psi,\phi} \coloneqq \abs{\braket{\psi}{\phi}}^2$.

For a given quantum state $\ket\psi$, sample-access typically refers to the ability to obtain a random sample $\ket{x}$ from the state $\ket{\psi}$ with probability $|\langle{x}|{\psi}\rangle|^2$.
Query-access refers to the ability to obtain the amplitude $\langle{x}|{\psi}\rangle$ for any $x \in \B^n$.
In this work, we do not discuss computational constraints on the storage or encoding of these elements and values.
We assume a sufficient precision and at most polynomial access cost.
Sample-query access refers to having both sample- and query-access to the state as defined above.
See \cref{app:sample-query} for formal definitions of these access models.

\subsection{Complexity Theory}\label{app:complexity}
\begin{definition}[Polynomial-time Generated Quantum Circuit]
    Let $L \subseteq \B^*$ be any set of strings. 
    Then a collection $\{Q_{|x|} : x \in L\}$ of quantum circuits is said to be polynomial-time generated if there exists a polynomial-time deterministic Turing machine that, on every input $x \in L$, outputs an encoding of $Q_{|x|}$.    
\end{definition}

\begin{definition}[\cl{BQP}]
    Let $L = (L_{\sc{yes}}, L_{\sc{no}})$ be a promise problem and $a,b : \mathbb{N} \to [0,1]$ be functions.
    A problem $L$ belongs to the class \cl{BQP}$(a,b)$ if and only if there exists a polynomial-time generated quantum circuit family $Q = \{Q_n\}_{n\in\mathbb{N}}$ that acts on $n + \poly{n}$ input qubits and produces one output qubit, such that:
    \begin{itemize}
        \item If $x \in L_{\sc{yes}}$, then $\Pr[Q_n(x) = {\tt 1}] \geq a(n)$.
        \item If $x \in L_{\sc{no}}$, then $\Pr[Q_n(x) = {\tt 1}] \leq b(n)$.
    \end{itemize} 
\end{definition}

The class \cl{BQP} is defined via $\cl{BQP} \coloneqq \cl{BQP}(2/3,1/3)$.
Via repetition and majority voting, the class \cl{BQP} has error reduction allowing for $\cl{BQP}=\cl{BQP}(1-2^{-q},2^{-q})$, for any polynomially-bounded function $q \geq 2$.
We say that a problem is \clw{BQP}{hard} if every problem in \cl{BQP} can be reduced to it using a polynomial-time (Karp) reduction.
A problem is \clw{BQP}{complete} if it is both in \cl{BQP} and \clw{BQP}{hard}.
We assume all \cl{BQP} circuits are composed of $2$-local gates from a finite universal gate set, e.g., $\{I,\tn{H, T, CNOT}\}$.
Note that we chose to include the primitive $I$ in the gate set to allow for the pre-idle sequences in the history state construction.
When this element is not included, the pre-idle sequence can be implemented by a sequence of gates that act trivially on the computational state, e.g., $\tn{H}^2$ or $\tn{T}^8$ --- see \cref{app:alt_idling_reductions} for more details.

\subsection{Problem Statements}\label{app:problem_statement}

The main problem considered in this work is the \sc{Guided Local Hamiltonian} problem.
\guidedlocalhamiltonian*

Later in this appendix, we consider other problems related to the \sc{Local Hamiltonian} problem.

\begin{restatable}[\sc{[State Type] Local Hamiltonian} problem]{definition}{statetypelocalhamiltonianproblem}\label{def:STLHP_general}
    Given a $k$-local Hamiltonian $H$ defined over $n$ qubits and parameters $a,b \in \mathbb{R}$ such that $b - a \geq 1/\poly{n}$, the problem is to decide whether 
    \begin{itemize}
        \item \sc{yes}: there exists an $n$-qubit [state type] state $\ket{\psi}$ such that $\bra{\psi}{H}\ket{\psi} \leq a$,
        \item \sc{no}: for all $n$-qubit [state type] states $\ket{\psi}$, $\bra{\psi}{H}\ket{\psi} \geq b$,
    \end{itemize}
    promised that one of these is the case.
\end{restatable}

\subsection{Gaussian States}\label{app:gaussian} 
Following Refs.~\cite{bravyi2017complexity,terhal2002classical,bravyi2004lagrangian,chapman2023unified}, consider the $2^n$-dimensional Hilbert space $\mathcal{H}_n$ describing $n$ fermionic modes. 
A convenient basis is given by the Fock states
\begin{equation*}
    \ket{x_0,x_1,\dots,x_{n-1}}=(a_0^\dagger)^{x_0} (a_1^\dagger)^{x_1}\dots(a_{n-1}^\dagger)^{x_{n-1}} \ket{0^n} \quad (x_j \in \B),
\end{equation*}
where $\{a_j\}_{j\in \mathbb{Z}_n}$ are complex fermionic annihilation operators satisfying the canonical anticommutation relations:
\begin{align*}
    \{a_i,a_j^\dagger\} &= \delta_{i,j}, & \{a_i,a_j\} &= 0, & \{a_i^\dagger,a_j^\dagger\} &= 0.
\end{align*}
The vacuum $\ket{0^n}$ is annihilated by every $a_j$, and each number operator $n_j = a_j^\dagger a_j$ has eigenvalues in $\B$.

It is often convenient to use \emph{Majorana-fermion operators}:
\begin{align*}
    c_{2j-1} &= \frac{1}{2}\bigl(a_j + a_j^\dagger\bigr), & c_{2j} &= \frac{1}{2i}\bigl(a_j - a_j^\dagger\bigr),
\end{align*}
which satisfy
\begin{equation*}
    \{c_k, c_\ell\} = 2\,\delta_{k,\ell}.
\end{equation*}
Majorana operators can simplify how certain observables and Hamiltonians are expressed.

A fermionic Hamiltonian is called \emph{free} or \emph{non-interacting} if it is quadratic in the fermionic operators:
\begin{equation*}
    H_{\tn{free}} = \frac{i}{2} \sum_{j,k} h_{jk} c_j c_k,
\end{equation*}
for some real, antisymmetric matrix $\boldsymbol{h}$ where $h_{jk}$ are the matrix elements. 
Diagonalising $\boldsymbol{h}$ reduces $H_{\tn{free}}$ to
\begin{equation*}
    \tilde{H}_{\tn{free}}=\sum_{k}\varepsilon_kb^\dagger_kb_k,
\end{equation*}
where $\{b_k\}$ are the canonical fermionic modes. 
Such Hamiltonians are diagonalised by \emph{Gaussian} (match-gate) unitaries $U_{\mathrm{mg}}$ whose conjugation sends each $c_k$ to a linear combination of all $c_\ell$:
\begin{equation*}
    U_{\tn{mg}}c_kU_{\tn{mg}}^\dagger=\sum_{\ell}R_{k\ell}c_\ell
\end{equation*}
Gaussian unitaries can be classically simulated~\cite{terhal2002classical,valiant2008holographic}.

Let $H_{\tn{free}}$ be a quadratic Hamiltonian, with energies
\begin{equation*}
    0\leq\varepsilon_0\le\dots\le\varepsilon_{n-1}
\end{equation*}
and let $m$ be the number of zero-energy modes, that is, $\varepsilon_0=\dots=\varepsilon_{m-1}=0$, with $\varepsilon_{m}>0$. 
Then the ground space of $H_{\tn{free}}$ has the form
\begin{equation*}
    \tn{span}(U_{\tn{mg}}\ket{x_0\dots x_{m-1}0\dots 0}:x_k\in \B).
\end{equation*}

Through the Jordan-Wigner transformation~\cite{jordan1928uber} (and its generalisations~\cite{chapman2020characterization}), we may express fermionic operators in terms of spins, giving free-fermionic states a computational meaning.

\begin{definition}[Gaussian states]
    A state $\varphi\in \H$ is called \textit{Gaussian} if and only if it is obtained from a bit string of the form $\ket{\boldsymbol{x},0}$ by a Gaussian unitary, i.e., $\ket{\varphi}=U_{\tn{mg}}\ket{x_1,\dots,x_m0.\dots 0}$. 
    Since Hamiltonians that admit an exact solution via a monomial-to-monomial mapping to free fermions are diagonalised by match gate circuits and Gaussian states are produced by match-gate circuits, it follows that for all $\varphi\in \H$ there exists a free Hamiltonian for which $\varphi$ is a ground state. 
\end{definition}

Let $\Gauss$ denote the set of all Gaussian states. 
Any state $\varphi\in \Gauss$ can be specified up to an overall phase by its covariance matrix $M$ of size $2n \times 2n$, which is defined as
\begin{equation*}
    M_{k,\ell}=-\tfrac{i}{2}\tr(\varphi[c_k,c_\ell])
\end{equation*}
The expectation value of any observable on a Gaussian state $\varphi\in \Gauss$ can be efficiently computed using Wick's theorem. 
For Gaussian states, all higher-order correlation functions factorise into two-point functions by Wick's theorem. 

Because Gaussian states have an efficient classical description by their Covariance matrix, the Gaussian-state variant of the \sc{Local Hamiltonian} problem is \clw{NP}{complete}, see~\cref{thm:gauss_lhp}.

\subsection{Fendley states}\label{app:fendley}
Recently, a new type of free-fermion solvable model was discovered by Fendley~\cite{fendley2019free}. 
Here, the fermions correspond to non-linear polynomials of the Pauli terms of the spin Hamiltonian. 
This solution holistically maps the spin Hamiltonian onto the free-fermion Hamiltonian, yet remains generic despite apparently transcending the monomial-to-monomial structure of the Jordan-Wigner map. 
In Refs.~\cite{elman2021free,chapman2023unified}, the solution by Fendley was extended to an entire family of models. 
The ground states of the Hamiltonians in this class are also `free-fermion', though not Gaussian. 
We call these states Fendley states.

\begin{definition}[Fendley states]
    Let $\vartheta\in \H$ be a quantum state. 
We say that $\vartheta\in \Fend$ if $\vartheta$ is the ground state of a Hamiltonian which is exactly solvable via a Fendley-type mapping.
\end{definition}

Much like Gaussian states, the physical motivation for considering the Fendley states here comes from the arguments made in Ref.~\cite{bravyi2017complexity} regarding the overlap between the low-energy eigenspace of the free Hamiltonian and the true ground state of the interacting system. 
Importantly, when compared to Gaussian free-fermion Hamiltonians, fewer terms need to be removed to reduce a general (interacting) model into one which admits a free-fermion solution via Fendley's method; thus, we expect Fendley states to provide an important superset of Gaussian states. 
As we will see, we can determine that the \sc{Guided Local Hamiltonian} problem with Fendley states is \clw{BQP}{hard}; however, placing it within \cl{BQP} is non-trivial since these states, as of yet, falter on their preparability from a classical description.

It was shown in Ref~\cite{chapman2023unified} that for systems of spatial dimension greater than one, such states may be prepared from Gaussian states via a constant-depth circuit. 
For one-dimensional models, however, it is known that there exist states for which the depth of the circuit is at least logarithmic in system size. 
Nevertheless, it is conjectured that all Fendley states may be prepared efficiently (no greater than polynomial depth)~\cite{fukai2025quantum}.

\section{BQP Hardness Framework}\label{app:bqp_hardness_framework}
For clarity, we restate \cref{lma:history_state_guiding}; however, we prove a slightly stronger version that includes the general clock mapping $\mu$.
\lmahistorystateguiding*

We now prove the following lemma.

\begin{restatable}{lemma}{lmageneralhistorystateguiding}\label{lma:general_history_state_guiding}
    Let $L \subseteq \B^*$ be a \cl{BQP} promise problem.
    Take $\mathcal{U}_x$ be a sequence of $K = \poly{|x|}$ unitary gates for instance $x \in L$.
    Using the Feynman-Kitaev construction, define the Hamiltonian
    \begin{equation*}
        \hat{H}_\mu = \Delta(H_{{\tn{in},\mu}} + H_{{\tn{clock},\mu}} + H_{{\tn{prop},\mu}}) + V_{{\tn{out},\mu}}.
    \end{equation*}
    Taking $\Delta > 112 K^3$ and $V_{{\tn{out},\mu}} = \ketbra{0}_1 \ketbra{K}_C$ such that $\norm{V_{{\tn{out},\mu}}} = 1$, results in $\hat{H}_\mu$ having a one-dimensional ground space spanned by $\ket{\hat{\phi}_0}$.
    Furthermore, it follows that 
    \begin{equation*}
        \norm{\ket{\hat{\phi}_0} - \ket{\eta_\mu}} = O(1/\poly{|x|}),
    \end{equation*}
    where $\ket{\eta_\mu}$ is the history state of the circuit.
\end{restatable}

\begin{proof}
    The well-known lower bound on the smallest non-zero eigenvalue of the Hamiltonian $\Delta(H_{{\tn{in},\mu}} + H_{{\tn{clock},\mu}} + H_{{\tn{prop},\mu}})$ is $\Omega(\Delta/K^3)$.
    The term $H_{{\tn{clock},\mu}}$ does not affect the lower bound~\cite{kitaev2002classical}
    Thus, the spectral gap is $\gamma = \Omega(\Delta/K^3)$ since the history state is the zero-eigenvector of $H_{{\tn{in},\mu}} + H_{{\tn{clock},\mu}} + H_{{\tn{prop},\mu}}$.
    The constant fact can be shown to be roughly $1/7$, though we do not require this.
    Standard Schrieffer-Wolff transformation results~\cite{bravyi2016complexity} demonstrate that by taking $\Delta > 112 K^3$, we can avoid mixing of the low-energy subspace with the high-energy subspace and obtain absolute convergence for the ensuing series expansion for the transformation.
    The ground state of the unperturbed Hamiltonian $H_{{\tn{in},\mu}} + H_{{\tn{clock},\mu}} + H_{{\tn{prop},\mu}}$ is the history state
    \begin{equation*}
        \ket{\eta_\mu} = \frac{1}{\sqrt{K+1}} \sum_{t=0}^{K} \ket{\varphi_{t}}\ket{\mu(\boldsymbol{t})},
    \end{equation*}
    where $\ket{\varphi_{t}} = U_{t}\ket{\varphi_{\boldsymbol{t-1}}}$ and $\ket{\varphi_{\boldsymbol{0}}} = \ket{x,0^m}$.
    For the appropriately chosen $\Delta$, the result is that the ground state of $\hat{H}_\mu$ is spanned by $\ket{\g}$.
    It then follows~\cite[Lemma 2]{bravyi2016complexity} that 
    \begin{equation*}
        \norm{\ket{\g} - \ket{\eta_\mu}} = O({\norm{V_{{\tn{out},\mu}}}}/{\Delta}) = O(1/K^3).
    \end{equation*}
\end{proof}

As we proceed, we will presume references to the Hamiltonian $\hat{H}_\mu$ are implicitly referring to Hamiltonians of the form in this lemma and the summary contained in the main text.
Note that the lower-bound on the spectral gap is $\Omega(\Delta/K^3)$, which holds for binary encodings of the clock states.
The specific coefficients in this bound may differ; however, for this work, this is not worthy of merit.

We now prove a geometric lemma relating the fidelity between a well-chosen guiding state and the ground state of an unperturbed Hamiltonian is sufficiently high, the subsequent fidelity between the same guiding state and a perturbed version of the Hamiltonian is also sufficiently high.
Therefore, the guiding state can be used to approximate the ground state of the perturbed Hamiltonian.

\lmageneralguidingstate*

\begin{proof}
    The case where the first $N$ unitaries are the identity is referred to as \emph{pre-idling}.
    Let $K = N + T$.
    We split the history state into two sectors: $\ket{\eta_\mu} = \ket{\eta_{\mu,1}} + \ket{\eta_{\mu,2}}$, where
    \begin{align*}
        \ket{\eta_{\mu,1}} &= \frac{1}{\sqrt{N + T + 1}} \sum_{\boldsymbol{t}=0}^{N} \ket{x,0^m}\ket{\mu(\boldsymbol{t})}, &
        \ket{\eta_{\mu,2}} &= \frac{1}{\sqrt{N + T + 1}} \sum_{\boldsymbol{t}=N+1}^{N + T + 1} U_{\boldsymbol{t-N}} \cdots U_{\boldsymbol{1}} \ket{x,0^m}\ket{\mu(\boldsymbol{t})}.
    \end{align*}
    Let $\sigma = \{x\} \times \{0\}^m \eqqcolon \{\varphi_0\}$, $\nu = N$ and $\mu' = \mu$, then we have the subset
    $R_{\varphi_0, \mu}^N = \{\varphi_0\} \times \{\mu(\boldsymbol{1}), \dots, \mu(\boldsymbol{N})\}$.
    The corresponding state is defined as 
    \begin{equation*}
        \ket{\hat{R}_{\varphi_0, \mu}^N} = \frac{1}{\sqrt{N}} \sum_{(z,t) \in R_{\varphi_0, \mu}^N} \ket{z,t}.
    \end{equation*}
    It is straightforward to verify the fidelity between $\ket{\hat{R}_{\varphi_0, \mu}^N}$ and $\ket{\eta_{\mu}}$ is 
    $F_{\hat{R}_{\varphi_0, \mu}^N, \eta_{\mu}} = N/({N + T + 1})$.
    Fix any polynomials $p$ and $p'$.
    Choose $N \geq T p(|x|)$, then 
    $N/({N + T + 1}) \geq 1 - 1/{p(|x|)}$.
    Using \cref{lma:general_history_state_guiding} and \cref{lma:geometric-lemma}, assuming $\norm{\ket{a}-\ket{b}} \le 1/p'(|x|)$, we have 
    $F_{\hat{R}_{\varphi_0, \mu}^N, \gamma} \geq 1 - 1/{p(|x|)} - 2/{p'(|x|)}$.
    In particular, for suitable polynomial choices of $p$ and $p'$, the fidelity satisfies
    \begin{equation*}
        F_{\hat{R}_{\varphi_0, \mu}^N, \gamma} \geq 1 - \frac{1}{q(|x|)}
    \end{equation*}
    for some polynomial $q$.
\end{proof}

To prove \cref{thm:R_state_hardness}, we require the following lemma that unifies the ideas of \cref{lma:general_history_state_guiding} and \cref{lma:general_guiding_state}.

\begin{restatable}{lemma}{jointlemma}\label{lma:joint_lemma}
    Let $L \subseteq \B^*$ be a \cl{BQP} promise problem.
    Take $\mathcal{U}_x$ be a sequence of $K = \poly{|x|}$ unitary gates for instance $x \in L$.
    Assume in the \sc{yes} case the acceptance probability of the circuit is at least $1 - \varepsilon$ and in the \sc{no} case the acceptance probability is at most $\varepsilon$.
    Define $\hat{H}_\mu = \Delta H_\mu + V_\mu$ be a Hamiltonian with a one-dimensional ground space spanned by $\ket{\hat{\phi}_0}$ constructed via the Feynman-Kitaev circuit-to-Hamiltonian construction with a subsequent Schrieffer-Wolf transformation for a sufficiently large $\Delta$.
    Suppose $\ket{\hat{R}_{\sigma, \mu}^\nu}$ is a state defined over a subset $R_{\sigma, \mu}^\nu$ such that $F_{\hat{R}_{\sigma, \mu}^\nu, \hat{\phi}_0} \geq 1 - 1/\poly{|x|}$.
    Then, $\hat{H}_\mu$ has an inverse-polynomial spectral gap $\hat{\gamma}$ and parameters $a,b$ such that $b-a = \Omega(1/\poly{|x|})$.
\end{restatable}

\begin{proof}
For an operator $A$, let $\lambda_j(A)$ denote its $j$-th smallest eigenvalue.
Define $\hat{\lambda}_j \coloneqq \lambda_j(\hat{H}_\mu)$ and $\lambda_j \coloneqq \lambda_j(H_\mu)$ for $j=0,1,\dots$.
Fix an instance $x$ and let $K=\poly{|x|}$ be the number of circuit gates.
Recall that the frustration-free ground state of $H_\mu$ is the history state $\ket{\eta_\mu}$.

Choose $\Delta=\poly{|x|}$ sufficiently large for the Schrieffer-Wolff (SW) approximation to hold with inverse-polynomial accuracy (see \cref{lma:general_history_state_guiding}).
Concretely, there exists a polynomial $p$ such that
\begin{equation*}
    \abs{\lambda_0(\hat{H}_\mu) - \lambda_0(P V_{{\tn{out},\mu}} P)} \leq \frac{1}{p(|x|)},
\end{equation*}
where $P = \ketbra{\eta_{\mu}}$
Therefore,
\begin{equation*}
    \hat{\lambda}_0\in \left[\mel{\eta_\mu}{V_{{\tn{out},\mu}}}{\eta_\mu}-\frac{1}{p(|x|)},\;
    \mel{\eta_\mu}{V_{{\tn{out},\mu}}}{\eta_\mu}+\frac{1}{p(|x|)}\right].
\end{equation*}

In the Feynman-Kitaev construction, $\mel{\eta_\mu}{V_{{\tn{out},\mu}}}{\eta_\mu}$ equals the (clock-averaged) rejection probability:
\begin{equation*}
    \mel{\eta_\mu}{V_{{\tn{out},\mu}}}{\eta_\mu} = \frac{\Pr[\tn{reject}]}{K+1}.
\end{equation*}
Thus, using $\Pr[\tn{reject}]\le \varepsilon$ in the \textsc{yes} case and $\Pr[\tn{reject}]\ge 1-\varepsilon$ in the \textsc{no} case, we obtain
\begin{align*}
    \textsc{yes}:&\quad \hat{\lambda}_0\in \comm{-\frac{1}{p(|x|)}+\frac{\varepsilon}{K+1}}{\frac{1}{p(|x|)}+\frac{\varepsilon}{K+1}},\\
    \textsc{no}:&\quad\hat{\lambda}_0\in \comm{-\frac{1}{p(|x|)}+\frac{1-\varepsilon}{K+1}}{\frac{1}{p(|x|)}+\frac{1-\varepsilon}{K+1}}.
\end{align*}

We now bound the spectral gap $\hat{\gamma}\coloneqq\hat{\lambda}_1-\hat{\lambda}_0$.
Let $\hat{H}_\mu=\Delta H_\mu+V_\mu$ with $\|V_\mu\|\le O(1)$ (in particular, $\|V_{{\tn{out},\mu}}\|\le 1$).
By the standard FK spectral analysis (see \cref{sec:bqp_hardness}), the first excited eigenvalue of $\Delta H_\mu$ obeys
\begin{equation*}
    \Delta \lambda_1\ge c\frac{\Delta}{K^3}
\end{equation*}
for some constant $c>0$.
By Weyl's inequality,
\begin{equation*}
    \hat{\lambda}_1 \ge \lambda_1(\Delta H_\mu)-\|V_\mu\| \ge c\frac{\Delta}{K^3}-\|V_\mu\|.
\end{equation*}
Choosing $\Delta=\poly{|x|}$ sufficiently large (e.g. see \cref{lma:general_history_state_guiding}) ensures $\hat{\lambda}_1 = \Omega(K^{-3}\Delta)$.
Combining with the above bounds on $\hat{\lambda}_0$ yields
\begin{equation*}
    \hat{\gamma} \ge \frac{1}{q(|x|)}
\end{equation*}
for some polynomial $q$, after taking $\Delta$ to dominate the inverse-polynomial error terms.

Finally, to produce a standard local-Hamiltonian instance with bounded operator norm (or bounded local terms), rescale $\hat{H}_\mu$ by a polynomial factor:
define $H'_\mu\coloneqq\hat{H}_\mu/s(|x|)$ where $s(|x|)\ge \|\hat{H}_\mu\|$ is polynomial.
Then $\|H'_\mu\|\le 1$ and the gap rescales as $\gamma'=\hat{\gamma}/s(|x|)\ge 1/\poly{|x|}$.

For the promise parameters, choose thresholds $a<b$ that separate the \textsc{yes} and \textsc{no} intervals above.
For example, take
\begin{equation*}
    a \coloneqq \frac{\varepsilon}{K+1}+\frac{2}{p(|x|)},\qquad
    b \coloneqq \frac{1-\varepsilon}{K+1}-\frac{2}{p(|x|)}.
\end{equation*}
For $\Delta$ sufficiently large we find $b-a=\Omega(1/\poly{|x|})$.
After rescaling by $s(|x|)$, the promise gap becomes $(b-a)/s(|x|)=\Omega(1/\poly{|x|})$ as required.
\end{proof}

The proof of \cref{thm:R_state_hardness} now follows directly from \cref{lma:joint_lemma} by taking $R_{\sigma, \mu}^\nu$ to be the subset used to define the state $\ket{\hat{R}_{\sigma, \mu}^\nu}$ in the theorem statement.

\subsection{Alternate Idling Reductions}\label{app:alt_idling_reductions}
In the main text, we discussed the pre-idling reduction, where the circuit is idled \emph{before} the sequence of $K$ local gates.
This implicitly assumes that we have access to the primitive gate $I$ in our gate set, which allows us to implement the pre-idle sequence by applying $N$ identity gates.
It is easy to see that when such an element is absent that there exists a sequence of gates that act trivially on the computational state; for example, the sequence $\tn{H}\ \tn{T}^8 \ \tn{H}$ acts as the identity on any computational basis state.
However, this will affect the structure of the history state and therefore the subset $R_{\sigma, \mu}^\nu$ (consequently the guiding state) that we use to prove \clw{BQP}{hardness}.
Using this fact, we anticipate that a number of different state forms can be encoded into the pre-idle sector of the history state.
More generally, we can also consider the alternative approach of idling \emph{during} and \emph{after} the circuit; alternative arguments can be constructed to prove \clw{BQP}{hardness} for the problem.

Consider a circuit $U$ composed of $K$ local gates and append $M$ identity gates at the end of the circuit.
This yields the history state
\begin{equation*}
    \ket{\eta} \propto \sum_{t=0}^{K} \ket{\varphi_t}\ket{t} + \sum_{t=K+1}^{K+M}\ket{\varphi_K} \ket{t},
\end{equation*}
where for the time period $t>K$, the workspace qubits are static.
An appropriate guiding state might be of the form 
\begin{align*}
    \ket{\xi} &\propto \ket{\varphi_K} \ket{K} + \ket{\varphi_K} \ket{K+1} + \cdots + \ket{\varphi_K} \ket{K+M}, \\
              &= U \ket{\varphi_0}\big(\ket{K} + \ket{K+1} + \cdots + \ket{K+M}\big).
\end{align*}
The preparation for which follows \cref{thmt@@EfficientStatePreparation}.
We show the circuit implementing this state is efficient in \cref{fig:static}.

Now consider a circuit $U$ composed of $K$ local gates and take $0 \leq s \leq K$.
Define a new circuit $U'$ that is identical to $U$ except after the $s$-th local gate, we apply a sequence of $S$ identify gates, e.g., 
\begin{equation*}
    U' = \prod_{i=1}^{s} U_i \left(\prod_{j=1}^{S} I_j\right) \prod_{k=s+1}^{K} U_k.
\end{equation*}
The history state for the Hamiltonian constructed from $U'$ is then
\begin{equation*}
    \sqrt{1+K+S}\,\ket{\eta} = \sum_{t=0}^{s-1} \ket{\varphi_t}\ket{t} + \ket{\varphi_s}\sum_{t=s}^{S} \ket{t} + \sum_{t=S+1}^{K}\ket{\varphi_K} \ket{t}. 
\end{equation*}
The workspace qubit state $\ket{\varphi_s}$ is static for the time-period $t \in [s,S]$.
A guiding state can be constructed as
\begin{equation*}
    \ket{\xi} \propto \left(\prod_{i=1}^{s} U_i\right) \ket{\varphi_0}\big(\ket{s} + \ket{s+1} + \cdots + \ket{S}\big).
\end{equation*}

\begin{figure}[!ht]
    \centering
    \begin{tikzpicture}
        \pic[scale=1.5]{static};
    \end{tikzpicture}
    \caption{
        An example of the circuit that prepares the guiding state $\ket{\xi}$ having overlap with the static sector of the history state. 
        The circuit element labelled {\tt PGR} is the preparation circuit for the {\tt PermutationGrover-Rudolph} algorithm of Ref.~\cite{ramacciotti2024simple}, outlined in \cref{app:permutation-grover-rudolph}.
        The circuit element labelled $U$ is a sequence of local gates that are applied to the workspace qubits after an initial state preparation procedure that prepares an underlying computational basis state. 
        Specifically, $\ket{\psi_1} = \ket{\underbar{0}}\ket{\underbar{0}}$, $\ket{\psi_2} = \ket{\varphi_0}(\sum_{t \in [s,S]} \ket{t})$, and $\ket{\psi_3} = \ket{\varphi_K}(\sum_{t \in [s,S]} \ket{t})$.
    }
    \label{fig:static}
\end{figure}

While the overlap and precision parameters can be made to match those of the pre-idling reduction, the type of guiding state is very different.
For one, our classical description is now similar to that of semi-classical encoded subset states.
Yet, any particular structure is necessarily lost due to the circuit $U$ since not much can be inferred about what the state $\ket{\varphi_S}$ is.
This is in contrast to the previous case, where the guiding state had a simple structure.
Therefore, in the spirit of this work, we rule out such guiding states as being physically relevant (in general).

In addition to the padding of identify gates, one could also consider inserting idling periods at various points throughout the circuit.
Even more so, one could consider sequences of subcircuits that are equivalent to the identify operation on the workspace qubits.
In this case, it may be possible to use non-trivial aspects of specific state families to yield additional \clw{BQP}{hardness} results.

\section{Extended State Families}\label{app:extended_state_families}
The standard definition of subset states assumes a fixed-basis encoding, i.e., using the binary alphabet $\B = \B$.
A natural extension is to consider a multi-alphabet encoding, where different portions of the state are encoded using different alphabets.
Such state families can be useful when considering perturbative gadgets that require ancilla qubits prepared in different basis states~\cite{waite2025guided,waite2025complexityb}.

\begin{definition}[Multi-Alphabet subset states]\label{def:multi_alphabet_subset_states}
    Let $\varSigma_j$ be an alphabet of size at most four.
    For $\ell_j \in \mathbb{N}$, take $\varSigma_j^{\ell_j}$ to be the set of strings of length $\ell_j$ over the alphabet $\varSigma_j$.
    Define $\Sigma \subset \bigtimes_{j=1}^{m} \varSigma_j^{\ell_j}$ to be a subset of the Cartesian product of $m$ sets $\varSigma_j^{\ell_j}$
    Assume that $\sum_{j=1}^{m}\ell_j = n$ and $\abs{\Sigma} = \poly{n}$, the state $\ket{\Sigma}$ defined over $\Sigma$ is defined as
    \begin{equation*}
        \ket{\Sigma} \coloneqq \sum_{x\in \Sigma} \alpha_x \ket{x},
    \end{equation*}
    where $\alpha_x$ are the amplitudes of the state.
\end{definition}

An example of such a state might be $\Sigma \subset \B^3 \times \{+,-\}^3 \times \{\Phi^+,\Phi^-,\Psi^+,\Psi^-\}$, resulting is a superposition state of the form 
\begin{equation*}
    \ket{\Sigma} = \frac{1}{\sqrt{2}}\big( \ket{010}\ket{+\!++}\ket{\Phi^-} + \ket{111}\ket{+\!--}\ket{\Psi^+}\big).
\end{equation*}
A multi-alphabet state is not particularly realistic as it would require some form of processing to partially disentangle parts of the state.
Though, it may be possible that certain decompositions of physical states could be approximated in this manner.
The state can be efficiently prepared using a constant-depth unitary transformation after an initial computational basis state preparation.
For the example above, preparing the state $\ket{010\;000\;10} + \ket{111\;011\;01}$ followed by the unitary transformation $U = \sc{Had}_4 \sc{Had}_5 \sc{Had}_6 \sc{Had}_7 \sc{CNOT}_{7,8}$ is sufficient to prepare the state $\ket{\Sigma}$.
Furthermore, multi-alphabet subset states using $\B$ and $\{+,-\}$ are appropriate for stoquastic Hamiltonians and efficient classical solutions for diagonal Hamiltonians~\cite{waite2025guided}.
But in a more general manner, with many alphabets, the endeavour becomes futile as it begins to stray from the physical relevance of the problem.
There is a strong similarity between multi-alphabet subset states and advanced SCESS --- defined in \cref{app:adv-scess}.
The latter has a more refined underlying structure, while the former can approach more general states.
Yet, the preparability of multi-alphabet subset states is up for debate in many `simple' circumstances.

Certain perturbative gadgets require state families with additional structure to account for ancilla qubits.
Even with isometries, fixed-weight states may not be sufficient to capture the necessary structure; for example using the gadgets from Ref.~\cite{schuch2009computational}.
It is an open problem to develop weight-preserving gadgets for these constructions.
Here, we define two natural extensions to fixed-weight states that may be useful in this context.

\begin{definition}[Weight-$(k\to q)$ Encoded State]\label{def:weight_k_to_q_encoded_states}
    Let $Y_{n,k} \subset \B^n$ be a set of binary strings of length $n$ with Hamming weight $k$ such that $\abs{Y_{n,k}} = \poly{n}$.
    Let $\mathcal{V} = \{V_j\}_{j\in[n]}$ be an ordered set of isometries such that: for each $j$ we have $V_j : \mathbb{C}^2 \to (\mathbb{C}^2)^{\otimes m_j}$ with $m_j = O(1)$ and for any $\ket{y}$ in the image of $\mathcal{V}$ the Hamming weight of $y$ is $q$.
    The encoded weight-$k$ state over $(Y_{n,k},\mathcal{V})$ is defined as
    \begin{equation*}
        \ket{Y_{n,k,\mathcal{V}}} \coloneqq \frac{1}{\sqrt{\abs{Y_{n,k}}}}\sum_{x\in Y_{n,k}} \bigotimes_{j\in[n]} V_j\ket{x_j}.
    \end{equation*}
\end{definition}

Weight-$(k\to q)$ encoded states are capable of retaining the fixed-weight structure of the input state under the sophisticated perturbative gadgets~\cite{cubitt2018universal,piddock2017complexity,schuch2009computational}.
The idea of using a fixed-weight state is to demonstrate that the Hamming weight of the state is reasonably small relative to the size of the input.
The state $\ket{Y_{n,k,\mathcal{V}}}$ is a fixed-weight state for the encoded system and therefore carries the same physical interpretation as the original state.
However, the geometrical perturbative gadgets outlined above do not fit this model since a tensor product of a polynomial number of $\ket{+}$ states does not have a fixed weight.

It is possible to add additional elements into fixed-weight states to cover a larger portion of the Hilbert space.
Specifically, we consider a window of Hamming weights and a state defined accordingly.
This has the physical significance of ``checking'' over an informed range of possible occupancies --- such information may be utilised in active space calculations in quantum chemistry.

\begin{definition}[Windowed Weight States]\label{def:windowed_weight_states}
    Let $\boldsymbol{k} = (k_1,k_2,\dots,k_d)$ be an increasing sequence of integers where $0 \leq k_j < k_{j+1} \leq n$.
    Consider the subset $\mathcal{X}_{n,\boldsymbol{k}} = \bigcup_{k\in \boldsymbol{k}} X_{n,k} \subset \B^n$ such that $\abs{\mathcal{X}_{n,\boldsymbol{k}}} = O(\poly{n})$.
    The polynomially-sized spectrum weight state over $\mathcal{X}_{n,\boldsymbol{k}}$ is defined as
    \begin{equation*}
        \ket{\mathcal{X}_{n,\boldsymbol{k}}} \coloneqq \frac{1}{\sqrt{\abs{\mathcal{X}_{n,\boldsymbol{k}}}}}\sum_{k \in \boldsymbol{k}} \sqrt{\abs{X_{n,k}}}\,\ket{\hat{X}_{n,k}}.
    \end{equation*}
\end{definition}

\section{Efficient State Preparation from Sparse Classical Data}\label{app:bqp_containment}
This appendix provides a proof of the statements in \cref{sec:bqp_containment}.
We prove that all the classes of states, except Fendley states, considered in this work can be efficiently prepared from a classical description.
The inclusion of the problem, with a given state type, in \cl{BQP} is then followed by an application of \cref{lma:qpe}.
As noted in \cref{sec:bqp_containment}, Ref.~\cite{cade2022complexity} adopts an alternative problem definition, so our results do not conflict with their framework; however, they are essential for the statements in Ref.~\cite{cade2023improved}.

\begin{definition}[\sc{[State Type] Preparation} Problem]
    Given a classically efficient description of a target state $\ket{\phi}$, define a unitary $U$ that prepares a state $\ket{\psi}$ such that $\|\ket\psi-\ket\phi\|\leq1/\poly{n}$, where both the description of $U$ and its implementation require only polynomial space and time.
\end{definition}

It is well-known that preparing an arbitrary quantum state is challenging; even when an amplitude description is provided, such descriptions generally require exponential space. 
Consequently, the state types considered here are not arbitrary; the states we consider have a specific structure that allows for efficient classical descriptions.

\subsection{Approximate State Preparation}\label{app:approx-state-prep}
We first consider the implications of approximate state preparation; in particular, we provide guarantees on the overlap between the prepared state and the ground space of the local Hamiltonian, even when the state is not exact.
\cref{lma:geometric-lemma} places a bound from below on the overlap between the prepared state and the ground state of the Hamiltonian.
Furthermore, assume that we prepare the state $\ket{\psi}$ such that $\norm{\ket{\psi} - \ket{\phi}} \leq \varepsilon$, for some $\varepsilon \geq 1/\poly{n}$.
Then, the overlap between the prepared state and the ground state of the Hamiltonian is at least $F_{\psi,\g} \geq (\sqrt{\delta} - \varepsilon)^2 \eqqcolon \kappa$.
For $\delta \in (1/\poly{n},1-1/\poly{n})$ (such that $\varepsilon < \sqrt{\delta}$), the resultant overlap $\kappa$ also lies in a similar range, i.e., $\kappa \in (1/\poly{n},1-1/\poly{n})$.
This is at least an inverse-polynomial lower-bound on the overlap between the prepared state and the ground state of the Hamiltonian, and hence \cref{lma:qpe} can be applied to estimate the ground-state energy of the Hamiltonian.
An even tighter bound on $\varepsilon$, say $\varepsilon \geq 1/\exp(n)$, will also be sufficient to apply \cref{lma:qpe}.

\subsection{Permutation Grover-Rudolph Proposal}\label{app:permutation-grover-rudolph}
We now describe elements of the algorithm that prepare a semi-classical subset state $\ket{\hat{C}}$ from a classical description of the subset $C$.
We call a subset $C$ as \emph{sparse} if $|C|=\poly{n}$.
\citet{ramacciotti2024simple} proposed a two step algorithm for preparing a state $\ket{\hat{C}}$ from $C$.
Specifically, the algorithm uses the Grover-Rudolph algorithm to prepare a uniform superposition state of ordered indices, i.e.,
\begin{equation*}
    \ket{\psi} = \frac{1}{\sqrt{|B|}}\sum_{j=0}^{|B|-1} \ket{j},
\end{equation*}
where the size of the subset $B$ is equal to the size of the subset $C$.
The Grover-Rudolph portion of the algorithm is sufficient to prepare the state with non-uniform, but normalised, coefficients.
However, for the present purposes, we simply require a uniform superposition.

The second step of the algorithm is to permute the computational basis states, mapping each $\ket{j}$ to a corresponding $\ket{x}$ for some $x\in C$.
This is achieved by using cycle-based permutations. 
Specifically, \cite[Algorithm 4]{ramacciotti2024simple} outlines the classical algorithm for constructing each cycle.
The more interesting aspect of this algorithm is the transformation of the permutation into a quantum circuit.

\subparagraph{Permutation Unitary.} For a finite set $X$ of $k$ elements, we define a permutation $\sigma$ to be a bijective mapping from the set to itself. 
Moreover, $\sigma(X)$ is a trivial or non-trivial rearrangement of the elements in $X$.
A standard theorem of abstract algebra states that \emph{every permutation of a finite set can be written as a cycle or as a product of disjoint cycles}~\cite{gallian2021contemporary}.
Therefore, for any permutation $\sigma$ we expressed it as a product of disjoint cycles:
\begin{equation*}
    \sigma = c_0\,c_1\cdots c_m,
\end{equation*}
where $c_i$ are the cycles of $\sigma$ and $m \leq k-1$ is the number of cycles.
We say that a cycle $c_i$ has length $l_i$ if it permutes $l_i$ elements of the set $X$.
For example, $c_i = (x_{i_0}, x_{i_1}, x_{i_2})$ is a cycle of length $3$ that permutes the elements $x_{i_0} \mapsto x_{i_1}$, $x_{i_1} \mapsto x_{i_2}$, and $x_{i_2} \mapsto x_{i_0}$.

Let the elements of a set $X$ be $n$-bit strings, i.e., $X = \{x_0, x_1, \ldots, x_{k-1}\}$, where $x_j \in \{0, 1\}^n$.
Define the computational basis states as $\ket{x_j} = \bigotimes_{r=0}^{n-1} \ket{x_j[r]}$, where $x_j[r]$ is the $r$-th bit of the string $x_j$.
Our goal is to construct a unitary operator $U_\sigma$ that implements the permutation $\sigma$ on computational basis states in the set $X$.
It follows that $U_\sigma = \prod_{i=0}^{m} U_{c_i}$, where $U_{c_i}$ is the unitary operator that implements the cycle $c_i$.

Let $c$ be a cycle of length $l$ such that $c = (x_0, x_1, \ldots, x_{l-1})$.
It follows that we can decompose the unitary operator $U_c$ as $U_c = (\prod_{i=0}^{l-1} g_{i})B_0$, where $g_i$ is a unitary operator that implements a conditional Gray code rotation on the $i$-th and $(i+1)$-th elements of the cycle and $B_0$ is a boundary gate to reset an ancilla qubit.
Moreover, consider the cycle step $x_i \mapsto x_{i+1}$.
The Gray code rotation from $x_i$ to $x_{i+1}$ is defined as a sequence of bit flip operations that transforms the $n$-bit string $x_i$ into the $n$-bit string $x_{i+1}$.
For the $r$-th bit of the strings, define $F_i^r = x_i[r] \oplus x_{i+1}[r]$.
Define a unitary operator
\begin{equation*}
    V_{i} = \prod_{r=0}^{n-1} X_r^{F_i^r},
\end{equation*}
which is a string of at most $n$ Pauli-$X$ operators.
Let $CV_{i}$ be the controlled-$V_i$ operator, which applies $V_i$ to the target qubits if the control qubit is $\ket{1}$.
Let $\boldsymbol{C}_bX$ denote the multi-controlled $X$ gate with $n$ control qubits dictated by the bit string $b$ and one target qubit.
For example, when $b = 1^2$, the multi-controlled $X$ gate $\boldsymbol{C}_{1^2}X$ is equivalent to the Toffoli gate.
Let the unitary operator $g_i$ be defined over $n$ workspace and one ancilla qubit, such that
\begin{equation*}
    g_i = (\boldsymbol{C}_{x_i}X) \cdot (CV_{i}),
\end{equation*}
where for the first gate the target qubits are the $n$ workspace qubits and for the second gate the target qubit is the ancilla qubit.
The operator $g_i$ implements the cycle step $x_i \mapsto x_{i+1}$.
Therefore, the unitary operator $U_c$ can be expressed as above where we include the gate 
\begin{equation*}
    B_0 = \boldsymbol{C}_{x_0}X,
\end{equation*}
to reset the ancilla qubit due to $g_{l-1}$.
The sequence of cycles $c_i$ must be implemented in series over the workspace and ancilla registers.
Hence, we conclude that 
\begin{equation*}
    U_\sigma = \prod_{i=0}^{m} U_{c_i} = \prod_{i=0}^{m} \bigg(\big(\prod_{j=0}^{l_i-1} g_{j_i}\big)B_{0_i} \bigg) = \prod_{i=0}^{m} \bigg(\big(\prod_{j=0}^{l_i-1} (\boldsymbol{C}_{x_{j_i}}X) \cdot (CV_{j_i})\big)\boldsymbol{C}_{x_{0_i}}X \bigg),
\end{equation*}
which is a sequence dominated by \Gate{Toffoli}, \Gate{Cnot} and single-qubit unitaries when decomposed into elementary gates.
A straightforward analysis shows that the bound from above on the gate cost scales as $O(ln)$ for a given cycle of length $l$ and $n$ workspace qubits.
Therefore, to implement the permutation $U_\sigma$ for our subset $C$, we need to implement $m\cdot O(ln) = O(|C| n)$ gates.

\begin{figure}[!ht]
    \centering
    \begin{tikzpicture}
        \pic{cycle};
    \end{tikzpicture}
    \caption{An example of part of the unitary operator $U_c$ that implements the cycle increment $x_i \mapsto x_{i+1}$. The first multi-controlled $X$ gate is applied to the workspace qubits, controlled by the $n$-bit string $x_i$, with the target qubit as the ancilla qubit. Notice that the ancilla qubit is only flipped if the workspace qubits are in the state $\ket{x_i}$. The second gate is a controlled-$V_i$ operator that applies the Gray code rotation $V_i$ to the workspace qubits, conditioned on the ancilla qubit being $\ket{1}$.}
    \label{fig:cycle}
\end{figure}

\EfficientStatePreparation*

\begin{proof}
    We begin by applying the standard {\tt Grover-Rudolph} algorithm to prepare the state
    \begin{equation*}
        \ket{\psi}=\frac{1}{\sqrt{|C|}}\sum_{j=0}^{|C|-1} \ket{j}, 
    \end{equation*}
    which requires $\lceil \log_2(\abs{C}) \rceil$ qubits, each initialised in the state $\ket{0}$. 
    This step has complexity $O(|C|\, \log_2(\abs{C}))$. 
    Next, we pad the register with $r = O(n)$ ancilla qubits initialised in the state $\ket{0}$ to extend the Hilbert space to dimension $2^n$. 
    We then apply the {\tt SparsePermutation} algorithm~\cite[Algorithm 4]{ramacciotti2024simple} to the state $\ket{\psi}\ket{0^r}$.
    The purpose of this algorithm is to permute the computational basis states, mapping each $\ket{j} $ to a corresponding $\ket{x}$for some $x\in C$, by using cycle-based permutations. 
    This classical algorithm has a worst-case runtime of  $O(|C| n)$ that returns a list of cycles.
    The final step uses~\cite[Algorithm 7]{ramacciotti2024simple} and the above procedure to implement the permutation as a quantum circuit.
    Assuming each gate incurs a constant cost, both the classical and quantum complexities of this step scale as $O(|C| n)$.
    
    In conclusion, encompassing all the steps described above into a single algorithm --- {\tt PermutationGrover-Rudolph} --- produces an efficient quantum algorithm that prepares the state $\ket{\hat{C}}$ from a classical description of the subset $C$.
    This proves \cref{thmt@@EfficientStatePreparation}.
\end{proof}

\subsection{Inclusion of Isometries}\label{app:isometries}
Semi-classical encoded subset states are extensions of the semi-classical subset states, where each qubit is acted on by an isometry $V_j$ which is defined from a global isometry $\mathcal{V}$, e.g., $\mathcal{V} = \bigotimes_{j=1}^{n} V_j$.
Each isometry $V_j$ maps a single qubit to a constant number of qubits $m_j$.

\EfficientEncodedStatePreparation*

\begin{proof}
    Isometries are transformations between Hilbert spaces that preserve the inner product. 
    In particular, an isometry $V$ from $p$ to $q$ qubits is a $2^q \times 2^p$ complex matrix satisfying $V^\dagger V = I_{2^p}.$ A complete description of $V$ requires $2^{p+q+1} - 2^{2p} - 1$ real parameters. 
    Since any isometry can be seen as a collection of columns from a unitary matrix, there is considerable freedom in its embedding, and the associated unitary is not unique. 
    Additionally, the isometry can be embedded into a unitary using a block-encoding. 
    As a result, any isometry from $p$ to $q$ qubits can be efficiently decomposed into \Gate{Cnot} and single-qubit gates.
    In the event $p=q$, the isometry is a unitary operator, and the Solovay-Kitaev theorem ensures that it can be approximated to within inverse-exponential error in polynomial time.

    Under our definition of \textit{semi-classical encoded subset states}, each isometry $V_j : \mathbb{C}^2 \to (\mathbb{C}^2)^{\otimes m_j}$ maps single qubit to a set of $m_j\ge2$ qubits. 
    Consequently, we define the global isometry as $\mathcal{V} = \bigotimes_{j=1}^{n} V_j = \prod_{j=1}^{n} (I\otimes \cdots \otimes V_j \otimes \cdots \otimes I)$. 
    Known upper- and lower-bounds on the number of gates required to implement an isometry of size $2^q \times 2^p$ are given in Ref.~\cite[Table 2]{iten2016quantum}.
    In general, the cost is exponential in both $p$ and $q$. 
    However, the cost in our setting will not be exponential in $n$ since both $p$ and $q$ are constant.
    The total number of gates required to implement each isometry is $O(1)$, and thus the sequence of isometries can be efficiently implemented, requiring $O(n)$ gates. 
\end{proof}

It follows trivially that if the isometry $\mathcal{V}$ is replaced with a unitary $U$ comprised of a polynomial number of single-qubit and $2$-local gates, the resulting state can be prepared efficiently.
For simple examples, see \cref{fig:easyprep}, appropriate for perturbative gadgets that require ancilla qubits.
This prompts an alternative reduction to the \sc{Guided Local Hamiltonian} problem, where the guiding state has overlap with a different sector of the history state, see \cref{app:alt_idling_reductions}.

\begin{figure}[!ht]
    \centering
    \begin{tikzpicture}
        \pic{easyprep};
    \end{tikzpicture}
    \caption{
        Simple examples of how to prepare guiding states using the subroutines of \cref{sec:bqp_containment} and additional constant-depth circuits to add the ancillae qubits used in the gadget reductions.
    }
    \label{fig:easyprep}
\end{figure}

\subsection{Preparation of Matrix Product States}\label{app:mps_prep}

\begin{lemma}
    \label{lemma:mps_prep}
    MPSs can be efficiently prepared from the classical description detailed in \cref{sec:state_type_variations}
\end{lemma}
\begin{proof}
    The general structure for a Matrix Product State (MPS) is a state of the form
    \begin{equation*}
        \ket{\Psi} = \sum_{\underline{\sigma}} {\rm Tr}\big[(\prod_{j\in[n]} A^{\underline{\sigma}_j}) \big]\, \ket{\underline{\sigma}}.
    \end{equation*}
    MPSs are completely specified by the set of tensors $\{A_j^{\underline{\sigma}_j}\}$ and the physical qudits $\underline{\sigma}_j$.
    It requires a classical space complexity of $\Theta(n\,\chi^2\,\dim{(\underline{\sigma}_j)})$ to specify the state (where $\chi$ is the bond dimension).
    This is, of course, efficient provided both the bond dimension and physical dimension are bounded by a polynomial in $n$.
    It can also be \sc{MPS Preparation} problem can be resolved using $O(n \, \poly{\chi})$-time classical pre-processing followed by a quantum circuit of $O(n \, \chi \, \log(\chi)^2 \, \log(n/\varepsilon))$ gates with $\lceil \log(\chi) \rceil$ ancillae qubits~\cite{schon2005sequential}.
\end{proof}
\subsection{Preparation of Gaussian States}\label{app:Gauss_prep}
\begin{lemma}
    \label{lma:Gauss_prep}
    Gaussian states can be prepared efficiently from the classical description detailed in \cref{sec:state_type_variations}
\end{lemma}
\begin{proof}
    For any orthogonal matrix $Q\in O(2n)$, a match-gate circuit $U_{\tn{Q}}$ is a unitary that satisfies $U_{\tn{Q}} c_v U_{\tn{Q}}^\dagger=\sum_vQ_{v,u}c_u$ for any $v$. 
    From this definition, it follows that $U_{\tn{Q}^T}=U_{\tn{Q}}^\dagger$. 
    Since the product of Majorana operators $c_v$ forms a complete basis for linear operators, it suffices to specify how a unitary acts on the $2n$ Majoranas $\{c_v\}_v$ to uniquely specify the unitary up to a phase~\cite{bittel2025optimal}
    Specifically, for a given orthogonal matrix, there exists a known exact implementation of the associated match-gate circuit, which can be implemented using $O(n^2)$ local 2-qubit gates~\cite{jiang2018quantum}.
\end{proof}

\subsection{Main Result of \cref{sec:bqp_containment}}\label{app:bqp_containment_main}
\begin{theorem}\label{thrm:SCSS-GLHP}
    The \sc{Semi-Classical Subset State Guided Local Hamiltonian} problem is in \cl{BQP}.
\end{theorem}

\begin{theorem}\label{thrm:SCESS-GLHP}
    The \sc{Semi-Classical Encoded Subset State Guided Local Hamiltonian} problem is in \cl{BQP}.
\end{theorem}

These follow from an application of \cref{lma:qpe} preceded by \cref{thmt@@EfficientStatePreparation} and \cref{thmt@@EfficientEncodedStatePreparation} respectively
Since the \textit{fixed-weight states} defined in \cref{sec:state_type_variations} are a special case of subset states, we can trivially infer the following Corollary.

\begin{corollary}\label{thrm:kweight-GLHP}
    The \sc{Fixed-Weight State Guided Local Hamiltonian} problem is in \cl{BQP}.
\end{corollary}

Combining \cref{lma:qpe} with \cref{thrm:SCSS-GLHP,thrm:kweight-GLHP,thrm:SCESS-GLHP,lemma:mps_prep,lma:Gauss_prep}, is sufficient to conclude \cref{thm:bqp_contain}.

\subsection{History State Preparation}\label{app:history-state-prep}
Here we consider a family of states we call \emph{Unitarily Transformed Subset States}, defined as 
\begin{equation*}
    \ket{S_{U}} = \frac{1}{\sqrt{\abs{S}}}\sum_{z \in S} U_z \ket{z},
\end{equation*}
where $U = \{U_z\}_{z \in S}$ is a set of unitaries that act on the state $\ket{z}$ and $S$ is a subset of at most a polynomial number of computational basis states.
Given the classical description of the set $S$ and the unitaries $U_z$, it is possible to prepare the state $\ket{S_{U}}$ efficiently.

The history state $\ket{\eta}$ is an example of a unitarily transformed subset state.
The description follows as: $S = \{x\} \times \{0^m\} \times D_K$ where $D_K$ is the set of all unary strings from $0$ to $K$ and $U_z = U_t \cdots U_1$ is the sequence of local gates applied to the workspace qubits.
Notice, for simplicity, we are assuming the clock register is unary.
Recall that the history state is defined as
\begin{equation*}
    \ket{\eta} = \frac{1}{\sqrt{K+1}}\sum_{t=0}^{K} U_t \cdots U_1 \ket{x,0^m}\ket{1^t\,0^{K-t}}.
\end{equation*}

\begin{figure}[!ht]
    \centering
    \begin{tikzpicture}
        \pic{unary};
    \end{tikzpicture}
    \caption{An example of a circuit that prepares a uniform superposition of time steps in unary.}
    \label{fig:unary-time-prep}
\end{figure}
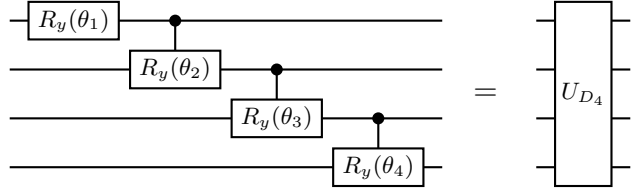

We perform a sequence of controlled rotation gates to prepare the superposition of time steps.
Our state preparation for the clock register has linear depth in $K$ with each rotation $\theta_j$ computable in classical polynomial time.
An example of such a circuit is shown in \cref{fig:unary-time-prep} for $\frac{1}{\sqrt{5}}(\ket{0000} + \ket{1000} + \ket{1100} + \ket{1110} + \ket{1111})$.
Each rotation gate $R_y(\theta_i)$ acts as $R_y(\theta_i) \ket{0} = \alpha_i\ket{0} + \beta_i\ket{1}$, where $\alpha_i = \cos(\theta_i/2)$ and $\beta_i = \sin(\theta_i/2)$.
Our circuit prepares the normalised state 
\begin{equation*}
    \ket{\psi_K} = \sum_{z \in D_K} f(z) \ket{z},
\end{equation*}
such that 
\begin{equation*}
    f(x) = \prod_{j=1}^{K} \alpha_j^{z_j \oplus 1} \beta_j^{z_j}.
\end{equation*}
To prepare the instance $\ket{x}$ we perform a series of Pauli-$X$ gates, i.e., $\boldsymbol{X} = \prod_{i=1}^{n} X_i^{x_i}$.
The history state is then prepared as shown in \cref{fig:history-state-prep}.
The depth of the circuit is $\Theta(K)$, where $K$ is the number of local gates in the circuit.

\begin{figure}[!ht]
    \centering
    \begin{tikzpicture}
        \pic{historystate};
    \end{tikzpicture}
    \caption{An example of a circuit that prepares the history state $\ket{\eta}$ from the instance $\ket{x}$ and the clock mapping $\mu$. The circuit $U$ is a sequence of local gates that are applied to the workspace qubits after an initial state preparation procedure that prepares an underlying computational basis state.}
    \label{fig:history-state-prep}
\end{figure}
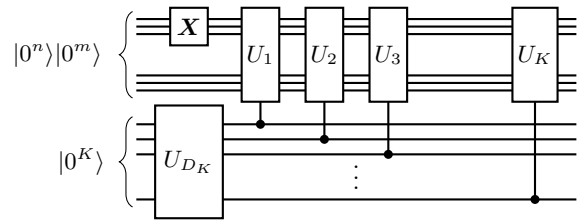

It follows that the ``non-pre-idled'' \clw{BQP}{hardness} proof discussed in \cref{sec:bqp_hardness} follows exactly from the results above.
Though since we define a perturbed Hamiltonian, the history state is not the exact ground state of the final Hamiltonian.
Our fidelity overlap regime is, therefore, the same as that of the pre-idled case.
This implies that though the history state is the optimal candidate for a guiding state, it is not the only candidate capable of producing high-fidelity overlaps.
The interesting point is that we do not require such a complicated state to solve the problem in quantum polynomial-time.

We can also conclude that non-flat Feynman-Kitaev constructions will likely also be \clw{BQP}{hard} for the \sc{Guided Local Hamiltonian} problem when using a guiding state from the family of \emph{Modified Unitarily Transformed Subset States}, i.e., where there is an additional input describing the amplitudes of the state.
For example, $\{(\alpha_z,z)\}_{z \in S}$, where $\alpha_z$ are the amplitudes of the state $\ket{z}$.
The interesting scenario is whether there are subclasses of states that recover the \clw{BQP}{completeness} result.
More importantly, are there any subclasses of states that differ from those already considered?

\section{Sample and Query Access}\label{app:sample-query}
In this appendix, we introduce formal definitions for sample and query access to quantum states.
We then proceed to show that certain classes of guiding states have efficient classical algorithms allowing for both sample and query access.

\begin{definition}[(Classically Efficient) Sample Access]
    \label{def:sample-access}
    Given a normalised state $\ket{\psi}$, we say that there exists classically efficient sample-access to $\ket{\psi}$ if there exists a classical algorithm that, given the description of $\ket{\psi}$, can output a sample from the probability distribution $\abs{\braket{z}{\psi}}^2$.
\end{definition}

\begin{definition}[(Classically Efficient) Query Access]
    \label{def:query-access}
    Given a normalised state $\ket{\psi}$, we say that there exists classically efficient query-access to $\ket{\psi}$ if there exists a classical algorithm that, given the description of $\ket{\psi}$, can compute the amplitude $\braket{z}{\psi}$ for any $z \in \B^n$.
\end{definition}

\begin{definition}[(Classically Efficient) Sample-Query Access]
    Given a normalised state $\ket{\psi}$, we say that there exists sample-query access to $\ket{\psi}$ if \cref{def:sample-access} and \cref{def:query-access} are satisfied.
\end{definition}

\subsection{Proofs of Classically Efficient Sample and Query Access}\label{app:sample-query-proofs}

\subsubsection{Semi-Classical Encoded Subset States}\label{app:SCESS-sample-query}
\begin{lemma}[Efficient Sampling from SCESS States~\cite{cade2023improved}]
    \label{lma:SCESS-sample}
    Given the description of a semi-classical encoded subset state $\ket{C_{\mathcal{V}}}$, one can sample classically from the distribution on $z \in \B^{M}$ given by $\Pr[z] = \abs{\braket{z}{C_{\mathcal{V}}}}^2$ in time $\poly{n}$.
\end{lemma}

\begin{proof}
    Assume we are given the encoding of the subset $C \subset \B^n$ where $|C| = O(\poly{n})$.
    Additionally, assume we are also given an encoding of the $n$ isometries within $\mathcal{V} = \{V_j\}_{j\in[n]}$.
    The SCESS is defined as
    \begin{equation*}
        \ket{C_{\mathcal{V}}} \coloneqq \frac{1}{\sqrt{\abs{C}}} \sum_{x \in C} \bigotimes_{j\in[n]}\, V_j\ket{x_j}.
    \end{equation*}
    Let the computational basis state images lie in $\B^{M}$.
    Let $p(y_0,y_1,\dots,y_{j-1}) = |(\bra{y_0,y_1,\dots,y_{j-1}} \otimes I)\ket{C_{\mathcal{V}}}|^2$, be the probability of measuring the first $j$ qubits of the state $\ket{C_{\mathcal{V}}}$ to be in the state $\ket{y_0,y_1,\dots,y_{j-1}}$.
    For each $j \in [M]$, we can efficiently calculate $p(y_0,y_1,\dots,y_{j-1})$ since $|C| = O(\poly{n})$ and $\bigotimes_{j\in[n]}\, V_j\ket{x_j}$ is a product state of $O(1)$-size qubit states.
    We can therefore also efficiently calculate the conditional probability
    \begin{equation*}
        p(z | y_0,y_1,\dots,y_{j-1}) = \frac{p(z,y_0,y_1,\dots,y_{j-1})}{p(y_0,y_1,\dots,y_{j-1})}.
    \end{equation*}
    If the bits $y_0,y_1,\dots,y_{j-1}$ have already been sampled, we compute $p(z | y_0,y_1,\dots,y_{j-1})$ and sample the next bit by tossing a coin with bias $p(1 | y_0,y_1,\dots,y_{j-1})$.
    This process is repeated until all $M$ bits have been sampled.
    The probability of sampling the string $z$ is then $\abs{\braket{z}{C_{\mathcal{V}}}}^2$.
\end{proof}

\begin{lemma}[Efficient Query Access for SCESS States]
    \label{lma:SCESS-query}
    Given the description of a semi-classical encoded subset state $\ket{C_{\mathcal{V}}}$, one can compute $\braket{z}{C_{\mathcal{V}}}$ and $\abs{\braket{z}{C_{\mathcal{V}}}}^2$ for any $z \in \B^{M}$ in time $\poly{n}$.
\end{lemma}

\begin{proof}
    For a given $z \in \B^M$, we wish to compute the amplitude $\braket{z}{C_{\mathcal{V}}}$.
    Since each isometry $V_j$ maps a single qubit to $m_j$ qubits, we denote the substrings of $z$ as $z_{\boldsymbol{j}} \in \B^{m_j}$ for each $j \in [n]$.
    Trivially,
    \begin{equation*}
        \braket{z}{C_{\mathcal{V}}} = \frac{1}{\sqrt{|C|}} \sum_{x \in C} \prod_{j=1}^n w_j,
    \end{equation*}
    where $w_j = \bra{z_{\boldsymbol{j}}}V_j\ket{x_j}$.
    The summation over $x \in C$ is efficient since $|C| = O(\poly{n})$.
    Furthermore, each term depends on the product of overlaps $\bra{z_{\boldsymbol{j}}}V_j\ket{x_j}$, where $V_j$ is an isometry acting on $O(1)$ qubits.
    For any fixed $x_j \in \B$, the vector $V_j\ket{x_j}$ has size $O(2^{m_j})$, where $m_j = O(1)$.
    Therefore, $\bra{z_{\boldsymbol{j}}}V_j\ket{x_j}$ can be computed in constant time.
    For a given $x \in C$, the product of overlaps $\prod_{j=1}^n \bra{z_{\boldsymbol{j}}}V_j\ket{x_j}$ is efficiently computable with $O(n)$ operations.
    Normalisation follows straightforwardly.
\end{proof}

\subsubsection{Advanced Subset States}\label{app:adv-scess}
Here, we consider a more advanced class of semi-classical encoded subset states.
These are defined with the potential for a more complex perturbative reduction and also to account for the potential need for multiple different few-qubit states (see \cref{sec:lattice_reduction})
Specifically, the isometries defined as per the SCESS are globally set --- this is not the case for the advanced SCESS.

\begin{definition}[Advanced SCESS]
    \label{def:adv-SCESS}
    For any subset $C \subset \B^n$ such that $\abs{C} = O(\poly{n})$, consider a collection $\mathcal{W}_C$ of $\abs{C}$ sets ordered isometries $\mathcal{V}_{x} = \{V_{x,j}\}_{j\in[n]}$ where, for each $j$ we have $V_{x,j} : \mathbb{C}^2 \to (\mathbb{C}^2)^{\otimes m_{j}}$ with $m_{j} = O(1)$.
    The advanced semi-classical encoded subset state $\ket{C_{\mathcal{W}}}$ over $(C,\mathcal{W}_C)$ is defined as
    \begin{equation}
        \ket{C_{\mathcal{W}}} \coloneqq \frac{1}{\sqrt{\abs{C}}} \sum_{x \in C} \mathcal{V}_x \ket{x} = \frac{1}{\sqrt{\abs{C}}} \sum_{x \in C} \bigotimes_{j\in[n]} V_{x,j}\ket{x_j}.
    \end{equation}
\end{definition}

Like SCESS, the images of $\mathcal{V}_x \ket{x}$ are product states over $O(1)$ qubits.
Notice in the definition that each isometry $V_{x,j}$ maps from $\mathbb{C}^2$ to $(\mathbb{C}^2)^{\otimes m_j}$, where $m_j$ is the same for any $x \in C$.
This ensures the resulting state has components with the same dimensionality.
To show the applicability of these states, especially for the \sc{Guided Local Hamiltonian} problem's classical algorithm, we prove the following two lemmas (in \cref{app:sample-query}).

Notice that the sample-query access to advanced SCESS will be inherited by state types that are a subset.
Furthermore, it is clear that unitarily transformed subset states, as defined in the previous section, are a special case of advanced SCESS.
Therefore, the history state can be defined with respect to an isometry of this form.

\begin{lemma}[Efficient Sampling from Advanced SCESS States]
    \label{lma:adv-SCESS-sampling}
    Given the description of an advanced semi-classical encoded subset state $\ket{C_{\mathcal{W}}}$, one can sample classically from the distribution on $z \in \B^{M}$ given by $\Pr[z] = \abs{\braket{z}{C_{\mathcal{W}}}}^2$ in time $\poly{n}$.
\end{lemma}

\begin{proof}
    The proof of this theorem follows directly from the reasoning in \cref{lma:SCESS-sample}. 
    By observing that in the Advanced SCESS state, the isometries $\{V_{x,j}\}_{j \in [n]}$ depend on the subset element $x \in C$, the sampling scheme can be adapted by incorporating the $x$-dependence into the calculation of the probabilities $p(y_0, y_1, \dots, y_{j-1})$ and $p(z \mid y_0, y_1, \dots, y_{j-1})$. 
    The key alteration is ensuring that for each $x$, the ordered isometries $\mathcal{V}_x = \{V_{x,j}\}_{j \in [n]}$ are used in computing the probabilities.

    Since the images of $\mathcal{V}_x \ket{x}$ remain product states over $O(n)$ qubits, and $\abs{C} = O(\poly{n})$, the efficient sampling scheme described in \cref{lma:SCESS-sample} applies without significant modification, ensuring that the string $z \in \B^M$ is sampled with probability $\abs{\braket{z}{C_{\mathcal{W}}}}^2$.
\end{proof}

\begin{lemma}[Efficient Query Access for SCESS States]
    \label{lma:adv-SCESS-query}
    Given the description of an advanced semi-classical encoded subset state $\ket{C_{\mathcal{W}}}$, one can compute $\braket{z}{C_{\mathcal{W}}}$ and $\abs{\braket{z}{C_{\mathcal{W}}}}^2$ for any $z \in \B^{M}$ in time $\poly{n}$.
\end{lemma}

\begin{proof}
    Similar logic to the proof of \cref{lma:SCESS-query} and \cref{lma:adv-SCESS-sampling} can be applied.
\end{proof}

\subsubsection{Matrix Product States}\label{app:mps-sample-query}

\begin{lemma}[Efficient Sampling from $\mathcal{F}(\mathrm{MPS}^*)$]\label{lma:mpsstar-sample}
    Given the classical tensor description of a state $\ket{\xi_{\rm MPS}}\in\mathcal{F}(\mathrm{MPS}^*)$ of the form \cref{eq:good-MPS}, one can sample classically from the distribution on $z\in\B^n$ given by $\Pr[z]=\abs{\braket{z}{\xi_{\rm MPS}}}^2$ in time $\poly{n,\chi}$.
\end{lemma}

\begin{proof}
    Because $\tilde A_j^{1}=0$ for all $j>\lambda$, every basis string with $z_j=1$ for some $j>\lambda$ has probability $0$.
    Hence the sampler may deterministically output $z_{\lambda+1}=\cdots=z_n=0$ and only needs to sample $z_1,\dots,z_\lambda$.

    We sample the first $\lambda$ bits sequentially.
    For a prefix $y_1,\dots,y_{j-1}$ define
    \begin{equation*}
    p(y_1,\dots,y_{j-1}) \coloneqq \|(\bra{y_1\cdots y_{j-1}}\otimes I)\ket{\xi_{\rm MPS}}\|^2.
    \end{equation*}
    For MPS, such prefix probabilities can be computed efficiently by contracting the corresponding partial tensor network using precomputed environments (transfer matrices), with cost $\poly{\chi}$ per step.
    Thus we can compute the conditional probabilities
    \begin{equation*}
    p(z_j=b\mid y_1,\dots,y_{j-1})=
    \frac{p(y_1,\dots,y_{j-1},b)}{p(y_1,\dots,y_{j-1})}
    \qquad (b\in\B),
    \end{equation*}
    sample $z_j$ by a biased coin flip, update the environment, and iterate.
    The resulting string $z$ is produced with probability $|\braket{z}{\xi_{\rm MPS}}|^2$.
    The total runtime is polynomial in $n$ and $\chi$.
\end{proof}

\begin{lemma}[Efficient Query Access for $\mathcal{F}(\mathrm{MPS}^*)$]\label{lma:mpsstar-query}
    Given the classical tensor description of a state $\ket{\xi_{\rm MPS}}\in\mathcal{F}(\mathrm{MPS}^*)$ of the form \cref{eq:good-MPS}, one can compute $\braket{z}{\xi_{\rm MPS}}$ and $\abs{\braket{z}{\xi_{\rm MPS}}}^2$ for any $z\in\B^n$ in time $\poly{n,\chi}$.
\end{lemma}

\begin{proof}
    Let $z\in\B^n$.
    If $z_j=1$ for some $j>\lambda$, then the corresponding tensor $\tilde A_j^{1}$ is zero by definition, hence the full contraction vanishes and $\braket{z}{\xi_{\rm MPS}}=0$.
    Otherwise,
    \begin{equation*}
        \braket{z}{\xi_{\rm MPS}} = {\rm Tr}\Big[\Big(\prod_{j=1}^{\lambda} A^{z_j}\Big)\Big(\prod_{j=\lambda+1}^{n}\tilde A_j^{0}\Big)B\Big],
    \end{equation*}
    which can be evaluated by sequential matrix multiplication and trace evaluation.
    Each multiplication involves matrices of dimension at most $\chi$, and the total cost is $O(n\chi^3)$ (or better depending on boundary conventions).
    Normalisation is computed analogously from the tensor description.
\end{proof}

\subsubsection{Gaussian States}\label{app:gauss-sample-query}

\begin{lemma}[Efficient Sampling from Gaussian States]\label{lma:Gauss-sample}
    Let $\ket{\varphi}$ be an $n$-mode fermionic Gaussian state specified by a classical description of its covariance matrix $M\in\mathbb{R}^{2n\times 2n}$ given to polynomially many bits of precision.
    Then one can sample classically from the distribution on $z\in\B^n$ given by $\Pr[z]=|\braket{z}{\varphi}|^2$ (in the occupation-number/computational basis) in time $\poly{n}$.
\end{lemma}

\begin{proof}
    We give a sequential sampling procedure based on conditional probabilities computed from Gaussian correlation data.

    For a prefix $y_1,\dots,y_{j-1}\in\B^{j-1}$ define
    \begin{equation*}
        p(y_1,\dots,y_{j-1}) \coloneqq \|(\bra{y_1\cdots y_{j-1}}\otimes I)\ket{\varphi}\|^2,
    \end{equation*}
    the probability of observing the first $j-1$ occupation outcomes to be $y_1,\dots,y_{j-1}$ when measuring in the computational basis.
    We aim to compute the conditional probabilities
    \begin{equation*}
        p(z_j=b \mid y_1,\dots,y_{j-1}) = \frac{p(y_1,\dots,y_{j-1},b)}{p(y_1,\dots,y_{j-1})}, \qquad b\in\B,
    \end{equation*}
    and sample $z_j$ by a biased coin flip.

    A key structural property of fermionic Gaussian states is that they are closed under conditioning on occupation-number measurements: conditioning on measurement outcomes on a subset of modes yields a (possibly unnormalised) post-measurement state on the remaining modes that is again Gaussian.
    Moreover, the covariance matrix of this conditioned state can be updated efficiently from the original covariance matrix using only polynomial-time linear-algebra operations on matrices of dimension at most $2n$.
    In particular, at each step $j$ we can:
    (i) compute the marginal probability $\Pr[z_j=1 \mid y_1,\dots,y_{j-1}]$ from the current covariance data of the conditioned Gaussian state, and
    (ii) update the covariance matrix conditioned on the sampled value $z_j\in\B$.

    Thus we can compute $p(z_j=b \mid y_1,\dots,y_{j-1})$ in time $\poly{n}$ at each step and sample $z_j$ accordingly.
    Iterating for $j=1,\dots,n$ produces a full string $z$ distributed exactly according to $\Pr[z]=|\braket{z}{\varphi}|^2$.
    The overall running time is polynomial in $n$.
\end{proof}

\begin{lemma}[Efficient Query Access for Gaussian States]\label{lma:Gauss-query}
    Let $\ket{\varphi}$ be an $n$-mode fermionic Gaussian state specified (up to global phase) by a classical description of its covariance matrix $M\in\mathbb{R}^{2n\times 2n}$ given to polynomially many bits of precision.
    Then for any $z\in\B^n$ one can compute the amplitude $\braket{z}{\varphi}$ and probability $|\braket{z}{\varphi}|^2$ (in the occupation-number/computational basis) in time $\poly{n}$.
\end{lemma}

\begin{proof}
    Fix $z\in\B^n$ with Hamming weight $k$, and let $J=\{j_1<\cdots<j_k\}$ denote the occupied modes.
    Gaussian states are completely determined by their two-point correlators, equivalently by the covariance matrix $M$ of Majorana operators.
    From $M$ one can efficiently compute the reduced two-point correlators restricted to the subset of modes $J$.

    For fermionic Gaussian states, occupation-basis amplitudes admit closed-form expressions in terms of these two-point correlators.
    In the number-conserving case (Slater determinants), $\braket{z}{\varphi}$ equals the determinant of a $k\times k$ matrix obtained from the corresponding one-particle data restricted to $J$.
    In the general case (allowing pairing), $\braket{z}{\varphi}$ equals the Pfaffian of a $2k\times 2k$ antisymmetric matrix $A_J$ whose entries are computable from the covariance matrix restricted to $J$.
    Both determinant and Pfaffian can be computed in polynomial time (e.g., via Gaussian elimination-type procedures in time $O(k^3) \subseteq O(n^3)$).
    Therefore $\braket{z}{\varphi}$ and $|\braket{z}{\varphi}|^2$ can be computed in time $\poly{n}$ from the description of $M$.
\end{proof}

\section{Local Hamiltonian Problems with Different States}\label{app:STLHP}
It is natural to consider variations of the \sc{Local Hamiltonian}, under which we are tasked with deciding whether states of a given type are extremal.
By this we mean, determine if there exists a state, from a family $\mathcal{F}$, such that $\tr(\rho H)$ is minimised, i.e., below a given threshold.
On the other hand, if all states in this family have energy above another threshold.
We formalise this problem with the following definition.

\statetypelocalhamiltonianproblem*

A classification of this problem's complexity has been presented by~\citet{kallaugher2024complexity} concerning the case when the state is a product state.
It was shown that for all families of $2$-local interactions $\mathcal{S}$, the problem \sc{$\mathcal{S}$-ProdLH} is \clw{NP}{complete}.
Additionally, this problem has been studied from the lens of parameterised complexity theory.
\citet{bremner2025parameterized} proved that the \sc{Weight-$k$ $l$-Local Hamiltonian} problem was contained in the class \cl{QW$[1]$} (the quantum analogue to \cl{W$[1]$} ``weft-$1$'') and hard for the class \cl{QM$[1]$}.
We now consider the case where the states are Gaussian.

The fact that Gaussian states have concise classical descriptions allows us to naturally refer to a decision problem with one parameter $a$, rather than a promise problem as for conventional \sc{$k$-Local Hamiltonian} problems. 
By convexity, a Gaussian state achieves an extreme value of $\tr(\rho H)$ if and only if there exists a pure Gaussian state $\ket{\varphi} \in \Gauss$, which achieves that value. 
By the same reasoning, mixtures of Gaussian states also cannot exceed the value attained by some pure Gaussian state.

\begin{theorem}\label{thm:gauss_lhp}
    The \sc{Gauss Local Hamiltonian} problem is \clw{NP}{complete}.
\end{theorem}

We sketch the proof of this theorem.
The \cl{NP} containment is achieved by the following lemma:

\begin{lemma}\label{lma:gauss_calc}
    Let $H$ be a $k$-local Hamiltonian for $k=O(1)$, and $\varphi\in\H$ be a Gaussian state.
    Then there is a classical efficient algorithm for calculating $\tr(\varphi H)$.
\end{lemma}

\begin{proof}
    Any $k$-local Hamiltonian term is a linear combination of at most $4^k-1$ Pauli terms. 
    The Pfaffian formalism allows us to evaluate the expectation value of any Pauli term for a Gaussian quantum state via the Jordan-Wigner transformation. 
    Since $k=O(1)$, there are only a constant number of Pfaffian calculations is required. 
    Moreover, the upper bound on the number of terms in the $k$-local Hamiltonian is $O(n^k)$. 
    Therefore, there exists an efficient algorithm for calculating the energy of a Gaussian quantum state for any $k$-local Hamiltonian.
\end{proof}

The Ising model is a $k$-local Hamiltonian that is \clw{NP}{hard} to solve and has a basis state as its ground state. 
Since basis states are Gaussian, Lemma~\ref{lma:gauss_calc} implies that the \textsc{Gauss Local Hamiltonian} problem is \clw{NP}{complete}~\cite{whitfield2014np}.

\section{Weight-k Guiding States}\label{app:weight_k}
The \emph{fixed-weight} states considered do not have their Hamming weight as a parameter; rather, the Hamming weight is fixed with respect to the instance.
Alternatively, we can consider a superposition state parameterised by the Hamming weight.
We refer to such states as \emph{weight-$k$} states, $\ket{\psi_k} = \sum_{x : \hw(x) = k} \alpha_x \ket{x}$.
The local Hamiltonian problem variant that asks to determine if there exists an extremal weight-$k$ state $\ket{\psi_k}$ is known to be in the class \cl{XP} \cite{bremner2022quantum} via straightforward projection and diagonalisation arguments.
Furthermore, this problem can be verified in constant depth, using one ``big-\sc{and}'' gate along any given path (weft-$1$), i.e., the problem is in \cl{QW}$[1]$ \cite{bremner2025parameterized}, though completeness for this class is unknown.
Though we have not been able to construct a hardness result for the \sc{Guided Local Hamiltonian} problem using weight-$k$ states, we can comment on the difficulty of proving this and the complexity of the problem in general.

To parametrise the guiding state by Hamming weight would require a modification in the Feynman-Kitaev construction to remove the instance from the input and a careful choice of clock mapping.
Or, construct a reduction that creates a history state with a parameterised Hamming weight.
Assuming guiding states of this type \emph{could} be constructed, it may then be expected that the arguments proving containment in \cl{XP} would be sufficient to prove the problem was classically tractable.
Moreover, assume there exists a state $\ket{\psi_k}$ with overlap $\delta$ with the ground state of a local Hamiltonian $H$.
Let the ground-state energy of $H$ be $\lambda_0$.
Since $k$ is known via an $f(k) \log(n)$-sized classical description of $\ket{\psi_k}$, we define the Hamiltonian $H_k = P_k H P_k$ where $P_k$ is the projector onto the subspace of Hamming weight $k$.
Let the ground-state energy of $H_k$ be $\mu_0$.
Brute force diagonalisation of $H_k$ is then possible in $n^{O(k)}$ time, so $\mu_0$ is computable classically.
In the \sc{Yes} case, if $\mu_0 \leq a$, then Courant-Fischer theorem implies that $\lambda_0 \leq a$.
In the \sc{No} case, if $\mu_0 \geq b$, then is it not clear that $\lambda_0 \geq b$.
The structure of the eigenstates of $H_k$ are difficult to determine analytically, though the guiding state $\ket{\psi_k}$ can be used as a trial state.

\begin{proposition}\label{prop:weight-k-bound}
    Consider a Hamiltonian $H$ with an eigensystem $\{(\lambda_j,\ket{\phi_j})\}_{j=0}^{2^n-1}$, where $\lambda_0$ is the ground-state energy and $\ket{\phi_0}$ is the ground state.
    Assume there exists a weight-$k$ state $\ket{\psi_k}$ such that $F_{\psi_k,\phi_0} = \delta$.
    Then, let $H_k$ be the Hamiltonian projected onto the weight-$k$ subspace with eigensystem $\{(\mu_i,\ket{\nu_i})\}_{i=0}^{\binom{n}{k}-1}$, where $\mu_0$ is the ground-state energy of $H_k$ and $\ket{\nu_0}$ is the ground state of $H_k$.
    The difference between the ground-state energy of $H$ and the ground-state energy of $H_k$ is bounded as
    \begin{equation*}
        (1-\delta)\gamma + \lambda_0 \leq \mu_0 \leq \lambda_0 + (1-\delta)\norm{H},
    \end{equation*}
    where $\gamma = \lambda_1 - \lambda_0$.
\end{proposition}

\begin{proof}
    Consider the weight-$k$ state 
    \begin{equation*}
        \ket{\psi_k} = \sum_{x : \hw(x)=k} \alpha_x \ket{x},
    \end{equation*}
    and the projector $P_k$ onto the weight-$k$ subspace.
    Take $Q_k$ to be the projector onto the orthogonal complement of the weight-$k$ subspace.
    Assume that the overlap between the ground state and the weight-$k$ state is a parameter $\delta$, i.e., $F_{\psi_k,\phi_0} = \delta$.
    It follows from the fact that the set $\{\ket{\phi_j}\}_{j=0}^{2^n-1}$ is an orthonormal basis, that 
    \begin{equation*}
        \ket{\psi_k} = \sum_{j=0}^{2^n-1} c_j \ket{\phi_j} = {\rm e}^{{\rm i} \theta_0} \sqrt{\delta} \ket{\phi_0} + \sum_{j>0} c_j \ket{\phi_j}.
    \end{equation*}
    Since $\ket{\psi_k}$ lies in the weight-$k$ subspace, we have that $P_k \ket{\psi_k} = \ket{\psi_k}$.
    Define the projected Hamiltonian $H_k = P_k H P_k$.
    The eigensystem of $H_k$ is $\{(\mu_i,\ket{\nu_i})\}_{i=0}^{\binom{n}{k}-1}$, where $\mu_0$ is the ground-state energy of $H_k$ and $\ket{\nu_0}$ is the ground state of $H_k$.
    Unless $\delta=1$, then $\ket{\nu_0}$ is not equal to $P_k \ket{\phi_0}$.
    To bound the quantity ${\mu_0 - \lambda_0}$, we use the variational principle and the trial state $\ket{\psi_k}$.
    Specifically, $\mu_0 \leq \bra{\psi_k} H_k \ket{\psi_k}$, and
    \begin{equation*}
        \bra{\psi_k} H_k \ket{\psi_k} = \bra{\psi_k} H \ket{\psi_k} = \sum_{i,j} c_i^* c_j \bra{\phi_i} H \ket{\phi_j} = \sum_{i,j} c_i^* c_j \lambda_j \braket{\phi_i}{\phi_j} = \sum_{j} \abs{c_j}^2 \lambda_j.
    \end{equation*}
    Expanding $\ket{\nu_0}$ in the eigenbasis of $H$, i.e., $\ket{\nu_0} = \sum_{j} d_j \ket{\phi_j}$, we have that 
    \begin{equation*}
        (1-\abs{d_0}^2)\gamma + \lambda_0 \leq \bra{\psi_k} H_k \ket{\psi_k} \leq \lambda_0 + (1-\delta)\norm{H},
    \end{equation*}
    where $\gamma = \lambda_1 - \lambda_0$.
\end{proof}

Therefore, unless the overlap is sufficiently large, the problem is not classically tractable.
Moreover, for the \sc{No} case, even in the event $\psi_k = \nu_0$, we find that $\mu_0 \geq \lambda_0 + (1-\delta)\gamma$ and therefore we require, at least, $\delta > 1 - \varrho_{ab}/\gamma$, where $\varrho_{ab}$ is the polynomially-small promise bound.
Though, this does not rule out a proof of classical tractability via other means.

We note that this proposition is conceptually related to the \emph{projection lemma} of~\citet{kempe2006complexity}: both bound how restricting the Hilbert space (or penalising its complement) affects low-energy eigenvalues. 
However, the projection lemma assumes an explicit large penalty (spectral separation) and treats the remainder as a small perturbation (yielding an error that scales like $\norm{H}^2 / J$), whereas \cref{prop:weight-k-bound} assumes only the existence of a single weight-$k$ state with fidelity $\delta$ to the ground state and produces $1-\delta$ bounds. 
The two results are complementary with that of Ref.~\cite{kempe2006complexity} being perturbative and constructive, whereas our bound is overlap-based and elementary.

\section{Optimality of Uniform Amplitudes}\label{app:uniform}
Here we provide arguments on the optimality of uniform amplitude states against the Feynman-Kitaev construction history state, specifically in the case of flat coefficients.
To show this, we can reframe the overlap problem as a Lagrange multiplier problem.
Specifically, consider a subset state $\ket{S}$ of the form $\ket{S} = \sum_{x\in S} \alpha_x \ket{x}$.
The fidelity of $\ket{S}$ with the history state portion $\ket{\eta_{\mu,1}}$ is given by
\begin{equation*}
    F_{S,\eta_{\mu,1}} = \frac{1}{N+K+1} \abs{\sum_{x \in S \cap E_{\mu,1}} \alpha_x^*}^2.
\end{equation*}
We have denoted $E_{\mu,1}$ as the set of computational basis states in $\ket{\eta_{\mu,1}}$.
The task is the maximisation of $\abs{\sum_{x \in S \cap E_{\mu,1}} \alpha_x^*}^2$, the elementary symmetric polynomial of order-$2$, subject to the constraint $\sum_{x\in S} \abs{\alpha_x}^2 = 1$.
Clearly, $S \neq E_{\mu,1}$ is already sub-optimal.
The Lagrange multiplier problem is then 
\begin{equation*}
    \max_{\alpha_x} \abs{\sum_{x \in S \cap E_{\mu,1}} \alpha_x^*}^2 - \lambda\left(\sum_{x\in S} \abs{\alpha_x}^2 - 1\right).
\end{equation*}
By Cauchy-Schwarz, it can be shown that the optimal solution is $\alpha_x = 1/\sqrt{\abs{S}}$ for $x\in S$.
Hence, a relaxation of the uniform amplitude condition is sub-optimal and will likely result in a smaller upper bound from \clw{BQP}{hardness} proof.
An example of this is seen in Ref.~\cite{gharibian2023dequantizing}.

\section{Hardness Reduction for the Fermi-Hubbard Model}\label{app:fermi_hubbard_complexity}

We provide the details supporting \cref{thm:guided-fermi-hubbard-no-fields}; the local-field case (\cref{thm:guided-fermi-hubbard-local-fields}) follows an identical argument once the guiding-state analysis of \cref{sec:fermi_hubbard_local_fields} is substituted.

\subsection{Jordan-Wigner Transformation and Local Encoding}\label{app:jw}

Consider $n$ lattice sites, each supporting two fermionic modes $(j,\sigma)$ with $j\in[n]$ and $\sigma\in\{\uparrow,\downarrow\}$, ordered lexicographically.
We apply the standard Jordan-Wigner (JW) transformation to represent the Fermi-Hubbard Hamiltonian $\tilde{H}$ as a $2n$-qubit operator.
The JW strings are supported only on qubits preceding $(j,\sigma)$ in the chosen ordering; in particular, nearest-neighbour hopping terms become four-qubit Pauli strings (explicit expressions are standard and omitted).

At half-filling in the singly-occupied subspace $\mathcal{L}_{-} \coloneqq \bigotimes_{j} \mathrm{span}\{\ket{10}_{j},\ket{01}_{j}\}$, we define the local isometry $J_j : \mathbb{C}^2 \to \mathbb{C}^4$ by
\begin{equation*}
  J_j\ket{\boldsymbol{0}}_j = \ket{10}_j, \qquad J_j\ket{\boldsymbol{1}}_j = \ket{01}_j,
\end{equation*}
where $\ket{10}_j$ and $\ket{01}_j$ denote occupation states $(n_{j,\uparrow}, n_{j,\downarrow}) \in \{(1,0),(0,1)\}$ respectively.
The global encoding $\mathcal{J} \coloneqq \bigotimes_{j} J_j$ embeds $n$-qubit spin states into the singly-occupied fermionic subspace.

\subsection{Perturbative Reduction to the Heisenberg Model}\label{app:hubbard_heisenberg}

Let $\Pi_{\pm}$ project onto $\mathcal{L}_{\pm}$ (low/high energy subspaces of $V$).
Since every hopping event creates or removes a doubly-occupied site, $T_{--} = 0$ and $V_{--} = 0$, while $V_{++} \succeq U\,I$.
Standard second-order perturbation theory~\cite{auerbach1994interacting} gives the effective Hamiltonian on $\mathcal{L}_{-}$:
\begin{equation*}
  \tilde{H}_{\mathrm{eff}} = -T_{-+}\,V_{++}^{-1}\,T_{+-} + O\left(U^{-2}t^3\right).
\end{equation*}
A direct local calculation on each edge $\langle i,j\rangle$, using the JW representations above, yields
\begin{equation}\label{eq:heff-heis}
  \tilde{H}_{\mathrm{eff}} = \frac{2t^2}{U}\sum_{\langle i,j\rangle} W_{ij} + k\,I + O\left(U^{-2}t^3\right),
\end{equation}
where $W_{ij} \coloneqq X_iX_j + Y_iY_j + Z_iZ_j$ is the Heisenberg interaction operator acting on the logical qubits via $\mathcal{J}$, and $k$ is an energy constant.

We invoke the second-order perturbation simulation framework of \citet{bravyi2016complexity} with the target Hamiltonian $H \coloneqq \sum_{\langle i,j\rangle} k_{ij} W_{ij}$, $0 \leq k_{ij} \leq \poly{n}$. 
Choosing $t$ and $U$ scaling at most polynomially in $n$ (following~\citet{ogorman2021electronic}), the Fermi-Hubbard Hamiltonian $\tilde{H}$ is an $(\eta,\epsilon)$-simulation of $H$ under $\mathcal{J}$, with $\eta, \epsilon = 1/\poly{n}$.
Manual normalisation then ensures $\|\tilde{H}\|, \|H\| \leq 1$ as required by the problem definition.
The geometry of local interactions is preserved under the reduction, so any lattice structure (square or triangular) present in $\tilde{H}$ is inherited by $H$.

\subsection{Guiding State Construction and Overlap Preservation}\label{app:guiding_state}
We now verify that the semi-classical encoded subset state (SCESS) guiding state $\ket{\xi_{\rm s}}$ for the Heisenberg instance induces a valid fermionic guiding state $\ket{\zeta_{\rm f}}$ for $\tilde{H}$ with sufficient overlap.

\paragraph{Construction.}
Define $\ket{\zeta_{\rm f}} \coloneqq \mathcal{J}\ket{\xi_{\rm s}}$.
If $\ket{\xi_{\rm s}}$ is a SCESS over $(S, \{V_j\}_{j\in[n]})$ with $V_j : \mathbb{C}^2 \to (\mathbb{C}^2)^{m_j}$, then composing each $V_j$ with the local encodings $(J_k)^{\otimes m_j}$ acting on its image gives updated isometries $W_j \coloneqq (J_k)^{\otimes m_j} \circ V_j$, and
\begin{equation*}
  \ket{\zeta_{\rm f}} = \frac{1}{\sqrt{|S|}} \sum_{x \in S} \bigotimes_{j \in [n]} W_j \ket{\boldsymbol{x}_j}.
\end{equation*}
This is a SCESS over $(S, \{W_j\}_{j\in[n]})$; the classical description is extended only by the fixed, efficiently computable local encodings $J_k$, preserving polynomial description complexity.
As a concrete example, the non-trivial isometry used in the Heisenberg reduction~\cite{cade2023improved} satisfies $V_j\ket{\boldsymbol{0}}_j = \frac{1}{\sqrt{2}}(\ket{\boldsymbol{01}} - \ket{\boldsymbol{10}})$, which under $J_{j^1} J_{j^2}$ maps to
\begin{equation*}
  W_j\ket{\boldsymbol{0}}_j = \tfrac{1}{\sqrt{2}}\bigl(\ket{10,01}_{j^1,j^2} - \ket{01,10}_{j^1,j^2}\bigr),
\end{equation*}
a state naturally supported on the singly-occupied subspace.

\paragraph{Overlap bound.}
Let $\ket{\phi_0}$ and $\gamma$ be the ground state and spectral gap of $H$ (with $\gamma \geq 1/\poly{n}$ for the instances produced by the
reduction~\cite{cade2023improved}), and suppose $\abs{\braket{\phi_0}{\xi_{\rm s}}}^2 \geq \delta$ for some $\delta \in (1/\poly{n}, 1 - 1/\poly{n})$.
By \cite[Lemma~2]{bravyi2016complexity}, for our choice of parameters the Fermi-Hubbard ground state $\ket{\psi_0}$ satisfies
\begin{equation*}
  \norm{\ket{\psi_0} - \mathcal{J}\ket{\phi_0}} \leq \eta + O\left(\epsilon/\gamma\right) = 1/\poly{n}.
\end{equation*}
Applying \cref{lma:geometric-lemma} then gives
\begin{equation*}
  \abs{\braket{\zeta_{\rm f}}{\psi_0}}^2 \geq \kappa,
\end{equation*}
for some $\kappa \in (1/\poly{n}, 1-1/\poly{n})$, where the interval is preserved up to the perturbative error (absorbed into the polynomial).
Since $\eta \propto n$, all polynomial bounds carry over to bounds in $\eta$.
This completes the \clw{BQP}{hardness} reduction for \cref{thm:guided-fermi-hubbard-no-fields}.
\end{document}